\documentclass[11pt]{article}
\pdfoutput=1
\usepackage{booktabs} % For formal tables
\usepackage[ruled]{algorithm2e} % For algorithms

\SetAlFnt{\small}
\SetAlCapFnt{\small}
\SetAlCapNameFnt{\small}
\SetAlCapHSkip{0pt}
\IncMargin{-\parindent}

\usepackage{booktabs}
\usepackage{amsfonts,amsmath,amssymb,amsthm,graphicx}
\usepackage{changepage}
\usepackage{enumerate}
\usepackage{scalerel}
\usepackage{accents}
\usepackage{hyperref}
\usepackage{xfrac}
\usepackage[margin=1in]{geometry}
\usepackage{fullpage}
\usepackage{multirow}
\usepackage{thm-restate}
\usepackage{natbib}
\usepackage{csquotes}
\usepackage{comment}

\usepackage[capitalize]{cleveref}

\usepackage{tikz}
\usetikzlibrary{patterns}
\usetikzlibrary{intersections}

\theoremstyle{plain}
\newtheorem{theorem}{Theorem}%[section]
\newtheorem{lemma}{Lemma}
\newtheorem{claim}{Claim}

\newtheorem{proposition}{Proposition}
\newtheorem{corollary}{Corollary}% reset theorem numbering for each chapter

\theoremstyle{plain}
\newtheorem{definition}{Definition}%[section] % definition numbers are dependent on theorem numbers
\newtheorem{example}{Example}
\newtheorem{assumption}[definition]{Assumption}

\theoremstyle{plain}
\providecommand{\keywords}[1]
{
  \small	
  \textbf{\textit{Keywords---}} #1
}

\newcommand{\val}{v}
\newcommand{\cost}{c}
\newcommand{\alloc}{x}
\newcommand{\pay}{p}
\newcommand{\util}{u}

\newcommand{\feasibles}{\mathcal{X}}
\newcommand{\type}{\theta}
\newcommand{\types}{\Theta}
\newcommand{\dist}{F}
\newcommand{\dista}{D}
\newcommand{\dists}{\mathcal{F}}
\newcommand{\mech}{\mathcal{M}}
\newcommand{\game}{\mathcal{G}}
\newcommand{\strategy}{\sigma}

\newcommand{\BNE}{{\rm BNE}}

\newcommand{\action}{a}
\newcommand{\actions}{A}
\newcommand{\first}{^1}
\newcommand{\second}{^2}
\newcommand{\combined}{^C}
\newcommand{\afirst}{\action\first}
\newcommand{\asecond}{\action\second}
\newcommand{\acombined}{\action\combined}
\newcommand{\asetfirst}{\actions\first}
\newcommand{\asetsecond}{\actions\second}
\newcommand{\asetcombined}{\actions\combined}
\newcommand{\mechf}{\mech\first}
\newcommand{\mechs}{\mech\second}
\newcommand{\mechc}{\mech\combined}
\newcommand{\combinedGame}{\mechc=\game(\mechf, \Gamma, \mechs)}

\newcommand{\combinedGameSimple}{\mechc=\game(\mechf, \mechs)}

\newcommand{\wel}{{\rm Wel}}
\newcommand{\rev}{{\rm Rev}}
\newcommand{\poa}{{\rm PoA}}

\newcommand{\bid}{b}
\newcommand{\E}{\mathbb{E}}

\newcommand{\signal}{\psi}
\newcommand{\signals}{\Psi}
\newcommand{\nosignal}{\bar{\signal}}
\newcommand{\opt}{{\rm OPT}}

\newcommand{\agents}{\mathcal{N}}
\newcommand{\items}{M} % \mathcal{M} is already taken by \mech
\newcommand{\estimate}{\psi}

\newcommand{\given}{\,\mid\,}

\newcommand{\prob}[2][]{\text{\bf Pr}\ifthenelse{\not\equal{}{#1}}{_{#1}}{}\!\left[{\def\givenn{\middle|}#2}\right]}
\newcommand{\expect}[2][]{\text{\bf E}\ifthenelse{\not\equal{}{#1}}{_{#1}}{}\!\left[{\def\givenn{\middle|}#2}\right]}

\newcommand{\tparen}{\big}
\newcommand{\tprob}[2][]{\text{\bf Pr}\ifthenelse{\not\equal{}{#1}}{_{#1}}{}\tparen[{\def\given{\tparen|}#2}\tparen]}
\newcommand{\texpect}[2][]{\text{\bf E}\ifthenelse{\not\equal{}{#1}}{_{#1}}{}\tparen[{\def\given{\tparen|}#2}\tparen]}

\newcommand{\sprob}[2][]{\text{\bf Pr}\ifthenelse{\not\equal{}{#1}}{_{#1}}{}[#2]}
\newcommand{\sexpect}[2][]{\text{\bf E}\ifthenelse{\not\equal{}{#1}}{_{#1}}{}[#2]}

\newcommand{\dd}{{\mathrm d}}

\newcommand{\reals}{\mathbb{R}}

\let\oldparagraph\paragraph
\renewcommand{\paragraph}[1]{\oldparagraph{#1.}}

\begin{document}

%\begin{titlepage}

%\title{On the Inefficiency of Auctions in Combined Markets}
%\title{On the Impact of Information Acquisition and Aftermarkets \\ on Auction Efficiency}
%\title{Making Carbon-Allowance Auctions Robust to Aftermarkets}
\title{Making Auctions Robust to Aftermarkets}

%\title{Making Auctions Robust to Aftermarkets with Application to Carbon Markets}

\author{Moshe Babaioff\thanks{Microsoft Research.
Email: \texttt{moshe@microsoft.com}}
\and Nicole Immorlica\thanks{Microsoft Research.
Email: \texttt{nicimm@gmail.com}. }
\and Yingkai Li\thanks{Cowles Foundation for Research in Economics, Yale University.
Email: \texttt{yingkai.li@yale.edu}.
Research carried out while this author was an intern in Microsoft Research New England Lab and a PhD candidate at Northwestern University.
} 
\and Brendan Lucier\thanks{Microsoft Research.
Email: \texttt{brlucier@microsoft.com}. }}

% \author{Moshe Babaioff}
% \affiliation{
% \institution{Microsoft Research}
% \city{Herzliya}
% \country{Israel}}
% \email{moshe@microsoft.com}

% \author{Nicole Immorlica}
% \affiliation{
% \institution{Microsoft Research}
% \city{New York}
% \country{USA}}
% \email{nicimm@gmail.com}

% \author{Yingkai Li}
% \affiliation{
% \institution{Yale University}
% \city{New Haven}
% \country{USA}}
% \email{yingkai.li@u.northwestern.edu}

% \author{Brendan Lucier}
% \affiliation{
% \institution{Microsoft Research}
% \city{Cambridge}
% \country{USA}}
% \email{brlucier@microsoft.com}

\date{}

\maketitle

\begin{abstract}
A prevalent assumption in auction theory is that the auctioneer has full control over the market and that the allocation she dictates is final. In practice, however, agents might be able to resell acquired items in an aftermarket. A prominent example is the market for carbon emission allowances. These allowances are commonly allocated by the government using uniform-price auctions, and firms can typically trade these allowances among themselves in an aftermarket that may not be fully under the auctioneer’s control. While the uniform-price auction is approximately efficient in isolation, we show that speculation and resale in aftermarkets might result in a significant welfare loss. Motivated by this issue, we consider three approaches, each ensuring high equilibrium welfare in the combined market. The first approach is to adopt smooth auctions such as discriminatory auctions. This approach is robust to correlated valuations and to participants acquiring information about others’ types. However, discriminatory auctions have several downsides, notably that of charging bidders different prices for identical items, resulting in fairness concerns that make the format unpopular. Two other approaches we suggest are either using posted-pricing mechanisms, or using uniform-price auctions with anonymous reserves. We show that when using balanced prices, both these approaches ensure high equilibrium welfare in the combined market. The latter also inherits many of the benefits from uniform-price auctions such as price discovery, and can be introduced with a minor modification to auctions currently in use to sell carbon emission allowances.
\end{abstract}


\section{Introduction}
\label{sec:intro}

{There is a vast literature in economics and computer science  analyzing the welfare properties of auctions at equilibrium.  A common assumption in this literature is that the auction participants are {the} \emph{consumers} who derive value from using the goods they purchase.  However, in many settings, buyers may have the option to resell their winnings in outside markets that are beyond the purview of the auction designer.  Resale possibilities can change behavior in the primary auction, such as by encouraging %\sout{socially-wasteful} %MB: at this point this feels too early % BL: agreed
speculation.  The resulting distortions can reduce final welfare even after accounting for any gains from post-auction trade.

This raises the possibility that mechanisms that appear approximately efficient in isolation might actually generate very poor outcomes {if agents anticipate aftermarket trade opportunities and adjust their strategies accordingly.}
%in practice due to aftermarket trading \mbc{I think we need to reword this. Post auction trading is always good (given the auction outcome). The welfare loss is actually due to the change of strategies in the auction due to the anticipation of the aftermarkets. Maybe change to "due to agents strategic behaviour resulted from their anticipation of aftermarket trade opportunities"}.  
Further complicating the issue, it may be that neither the auction designer nor the auctioneer have any sway over the format of the secondary market(s), and may not even be aware that they exist.  We therefore ask: how robust are different auction formats to the distortion effects of aftermarket trading?  Can existing theoretical analyses of auction efficiency at equilibrium be adapted to this setting?
}

%\nsi{can we argue that the dynamic aspect is not salient because licenses are annual, or argue that even if values change between primary and secondary, our results still are relevant?  try to wrap something about this into the intro maybe?  i think this will be the main complaint if we get an actual applied person reviewing this. note there are two ways the dynamics appear -- one is that my value might change over time and this seems like the main point; the other is that i have many more markets to participate in because there are 12 auctions and a continuum of secondary markets and so there are many more strategies than we actually discuss.  this latter thing somehow feels more okay to sweep under the rug and we also can cite brendan/adrian/et al if the reader is interested in that aspect....}
%\yl{I guess it will depend on how valuation evolves and what is the efficiency benchmark. Suppose that in primary auction the agent holds an expectation over valuation and this valuation is realized in the secondary market.If the benchmark is just the optimal allocation without information in the secondary market, then equilibrium welfare of smooth auction in combined market is not much smaller than that.}

\paragraph{Application: Markets for Carbon Emission Allowances}
% climate change is core challenge facing humanity and is driven by corporate emissions
Our planet is warming at an alarming rate.  
%According to the Annual 2020 Global Climate Report of the National Oceanic and Atmospheric Administration \citep{noaa},\footnote{The NAOA is an American scientific and regulatory agency within the United States Department of Commerce tasked with understanding changes in climate, among other related topics.} every month that year except December was in the top four warmest on record for that month. 
%https://www.ncdc.noaa.gov/sotc/global/202013
The Swiss Re Institute, the research arm of the reinsurance company Swiss Re based in Zurich, Switzerland, estimated in an April 2021 report \citep{Swiss} on the impact of climate change that ``The world stands to lose close to 10\% of total economic value by mid-century if climate change stays on the currently-anticipated trajectory.''
% \footnote{The Swiss Re Institute report is available \href{https://www.swissre.com/dam/jcr:e73ee7c3-7f83-4c17-a2b8-8ef23a8d3312/swiss-re-institute-expertise-publication-economics-of-climate-change.pdf}{here}.}
%https://www.swissre.com/dam/jcr:e73ee7c3-7f83-4c17-a2b8-8ef23a8d3312/swiss-re-institute-expertise-publication-economics-of-climate-change.pdf
Due to these dire circumstances, the United Nations made combating climate change and its impacts one of 17 sustainable development goals \citep{un}.
%https://www.un.org/sustainabledevelopment/climate-change/
%Greenhouse gas emissions caused by human activity are a primary contributor to climate change according to the United States Environmental Protection Agency \citep{epa}.
%https://www.epa.gov/climatechange-science/causes-climate-change
%Reducing human consumption of these gasses requires heroic efforts on the part of humanity.

% carbon markets are a technique used to mitigate climate change; often run via a centralized uniform price auction
Current efforts are largely focused on imposing limits on greenhouse gas emissions.  This makes the ``right to emit'' a scarce resource.  Emission allowances are allocated via markets, with the goal of distributing them to the industries that can provide the highest value to society per unit of emission.
%subject to the emissions they generate.  
%Each year, each company must present 
%an emission allowance to account for all its pollution, or face steep fines.  While some allowances are reserved for regulated industries such as public utilities and aviation, 
A large fraction of these allowances in major markets, such as the EU and California, are distributed via auctions on a monthly or quarterly basis.  
The EU Emissions Trading System (EU ETS), for instance, allocated 57\% of allowances via auction {between 2013 and 2020}.\footnote{See the article from \citet{euauction}.}
%https://ec.europa.eu/clima/eu-action/eu-emissions-trading-system-eu-ets/auctioning_en
Each of these auctions is typically run using a single-round sealed-bid uniform-price auction.  
%As described on the European Energy Exchange website \citep{eex}, which runs the EU ETS, the auction format is:
%\begin{itemize}
%    \item Single round: Bids will be submitted during one given bidding window.
%    \item Sealed bid: Bids will be submitted without seeing other participant's bids.
%    \item Uniform price: All successful bidders will pay the same auction clearing price.
%\end{itemize}
%https://www.eex.com/en/markets/environmental-markets/eu-ets-auctions
This format was adopted to reflect the priorities of the EU commission which requires, according to Article 10(4) of DIRECTIVE 2003/87/EC, that auctions are designed to ensure transparency, equitable informational and procedural access, and that ``participants do not undermine the operation of the auctions.''\footnote{See the legal document from \citep{eulaw}.}
%https://eur-lex.europa.eu/legal-content/EN/TXT/?uri=CELEX:02003L0087-20180408
\iffalse
4.  The Commission is empowered to adopt delegated acts in accordance with Article 23 to supplement this Directive concerning the timing, administration and other aspects of auctioning, in order to ensure that it is conducted in an open, transparent, harmonised and non-discriminatory manner. To that end, the process shall be predictable, in particular as regards the timing and sequencing of auctions and the estimated volumes of allowances to be made available.

Those delegated acts shall ensure that auctions are designed to ensure that:

(a) operators, and in particular any small and medium-sized enterprises covered by the EU ETS, have full, fair and equitable access;

(b) all participants have access to the same information at the same time and that participants do not undermine the operation of the auctions;

(c) the organisation of, and participation in, the auctions is cost-efficient and undue administrative costs are avoided; and

(d) access to allowances is granted to small emitters.
\fi

% carbon markets have an after-market 
{The uniform-price auction is known to be approximately efficient at equilibrium when run in isolation~\citep{de2013inefficiency}.}
{However,
%In regard to the last priority, 
%one clear point of concern 
%for these auctions 
%is the existence of unregulated (or only partially regulated) secondary markets for emission allowances.  
emission allowances can also be traded via unregulated (or only partially regulated) secondary markets.}
%, such as the secondary spot market (EEX) for EU ETS licenses.}
These secondary markets are not fully controlled by the primary auctioning bodies and can take many forms -- bilateral trade, brokered trade, and exchanges, to name a few.  
%For example, EU ETS licenses can be traded on  secondary spot and futures markets (EEX).
%
%
The existence of these secondary markets is concerning because they can distort outcomes, %creating {significant} inefficiency even when the underline auctions are {(approximately)} efficient in isolation. 
%But these secondary markets are also unavoidable absent extreme regulation. 
but they are also unavoidable absent extreme regulation.
%\mbc{the next does not connect well with the prior sentence}
%The EU ETS goes so far as to even sanction an EEX platform for secondary %spot, futures and options contracts for emission allowances. 

%there is opportunity for speculation in the aftermarket
What is the impact of these secondary markets on the primary market?
There are certainly potential benefits to secondary markets, %\mbc{these are problematic as to get the efficiency bound we need all to be able to participate, and that no new items are allocating after the auction. we should discuss this. } 
such as providing {simpler market access to smaller firms} who do not feel confident participating in the primary auction.
%, or allowing licenses that were not distributed by auction to be redistributed.
But distortion effects may be present as well.
%\mbc{Maybe better say "combined market"}\bjl{I like primary market}?  
According to \citep{eudoc}, the auction clearing prices closely track the mean of the best-ask and best-bid prices on the EEX spot secondary market. Furthermore, both prices rise steadily month-over-month. 
%https://ec.europa.eu/clima/system/files/2021-10/policy_ets_auctioning_cap_report_202109_en.pdf
These facts, taken together, suggest there is room for agents to {speculate by} buying allowances in the auctions and reselling them at a later time in the secondary markets. 
As noted in prior work of \citet{quemin2021financials},
%https://papers.ssrn.com/sol3/papers.cfm?abstract_id=3985079
``regulators are currently ill-equipped to appraise the beneficial and detrimental facets of speculation, and proper warning systems are wanting.''  Their work provides a diagnostic toolkit to assess the degree and impact of speculation in these markets. 
%\nsi{someone should probably skim their intro at least.}  
In our work, we ask whether the design of the primary market itself can defend against detrimental speculative behavior.  Namely, \emph{are the welfare guarantees of certain auction formats (approximately) robust to reallocation by arbitrary secondary markets? {Can small changes to currently-used auctions achieve such robustness?}}

\paragraph{A Model of Aftermarkets}
% our model
%To address {these questions} 
%this question 
%we focus on an idealized model of this market. We model the 
{Motivated by emission allowance markets,}  
%model the market as 
our primary model is 
a multi-unit auction in which each item corresponds to an allowance for one unit of emission.  {However, we note that most of our results 
%\mbc{is this true also for the new part about balanced prices and the part about auction reserves?} 
extend to a richer class of combinatorial auctions; see Section~\ref{sec:extensions}.}  In the multi-unit model, agents (buyers/bidders) have decreasing marginal value for units, representing their value for using (consuming) the allowances.  These valuation functions are private knowledge but drawn from known distributions.  Items (allowances) acquired in the auction can be resold in a secondary market (aftermarket). To distinguish a secondary market from a general mechanism, we impose some mild conditions on the form these markets can take. Specifically, we assume that these are \emph{trade mechanisms}: mechanisms that are budget balanced\footnote{Our results will hold not only for trade mechanisms that are strongly budget balanced (net payment of 0), but also for weakly budget balanced mechanisms (mechanisms that never lose money).} and do not force participation.  
%\nsi{i think we'll get dinged for not modeling the repeated nature; do we have a better defense, like the EU ETS auctions are run infrequently so well-modeled by single-shot auctions? Probably not since they are run monthly. Maybe just delete the footnote?}\yl{Each month they sell the licence for the coming month or they sell the licence for the year at each month? If it is the former case, then probably viewing it as single-shot will not hurt us too much. }\nsi{no it's the latter; i added this to the description paragraph up front.}\bjl{I agree that it would be better to remove the footnote.  We already say our model is idealized, and I don't think it helps so much to defend this specific abstraction. I removed it for now.}

% our question (how is welfare impacted)
Since agents anticipate the secondary market, the potential for resale can change behavior in the primary auction.  
{Our equilibrium notion is perfect Bayesian equilibrium (PBE).\footnote{All of the positive results in this paper actually hold for all Bayes-Nash equilibria (and thus, in particular, for all perfect Bayesian equilibria).
We choose the stronger notion of PBE as our solution concept to make the negative results more convincing, as we discuss later.
}
%\mbc{is it true that all the positive results are for BNE (and thus for PBE) while our negative results are even for PBE? If so, we can say that somewhere. }
Roughly speaking, each agent is assumed to behave rationally in the secondary market given the auction outcome and her beliefs about other agents (subgame perfection), and her belief % OLD: their beliefs, which 
must be consistent with observed outcomes of the auction (Bayesian updating).  Moreover, each agent is forward-looking and bids at equilibrium in the auction given their beliefs about what will occur in aftermarket trading.}
%\footnote{For our lower bounds, we will also impose a refinement that no agent uses a weakly dominated bidding strategy in the auction, even after deletion of dominated strategies in the secondary market.  In particular, this excludes ``bullying'' strategies whereby one agent bids unreasonably high in the auction while others do not participate.}
%
%and may encourage socially-wasteful speculation. As a result, the final allocation and welfare might be very different than that of the auction in isolation. 
One subtlety is that behavior in the secondary market can depend on the information released after the primary auction, such as whether bids are publicly observed. 
We want results that are robust to this choice, so we allow an arbitrary revelation of signals correlated with the auction bids and outcomes before the secondary market begins.  We assume that agents are fully aware of the secondary market (and what information they'll learn about the primary auction outcome) when participating in the primary auction.  We call the resulting mechanism that combines the primary auction and the aftermarket trade mechanism the \emph{combined market}. We seek conditions on the design of the primary market that guarantee high welfare {in every equilibrium allocation}
%OLD: allocations 
in the combined market.

%One naive approach would be to choose primary market {designs}, 
%such as the uniform-price auctions run by the EU ETS, that are approximately efficient in isolation. 
%Indeed, if participants of the primary market don't anticipate the secondary market, then any trade in the secondary market only Pareto improves the utilities of all agents (since the trade mechanism satisfies voluntary participation {and budget balance}) so welfare can never be harmed.  The problem is that the existence of the secondary market impacts the strategies of agents in the primary market by, for example, enabling speculation.  \bledit{It is therefore not immediate that an auction will retain its stand-alone efficiency properties when combined with a secondary market.}

\paragraph{Aftermarkets can Substantially Reduce Welfare}
It may seem counter-intuitive that secondary trade mechanisms can reduce overall welfare.  Indeed, 
%if participants of the primary market don't anticipate the secondary market, then 
given the outcome of the primary market,
as trade is voluntary, any trade in the secondary market only Pareto improves the utilities of all agents. 
%(since the trade mechanism satisfies voluntary participation {and budget balance}) so welfare can never be harmed.
{The problem is that the 
agents, being aware of the existence of the secondary market, adjust their strategies in the primary market. In particular, aftermarkets create opportunities for strategically acquiring items in the primary market for the sole purpose of opportunistically reselling  them later.} 
% OLD: The problem is that the  existence of the secondary market impacts the strategies of agents in the primary market by, for example, enabling speculation.  
In principle this can lead to lower welfare overall.  

{To formally illustrate the problem, we show that} {even a uniform-price auction}
%markets fully efficient \mbc{do we indeed show this about fully efficient? where?} in isolation 
can suffer arbitrarily large loss of welfare due to the presence of a secondary market.  Importantly, we insist that the equilibria we construct avoid weakly dominated strategies.
{Indeed, low-welfare equilibria are already known to exist for uniform-price auctions in isolation and thus also in combined markets, but these equilibria rely on overly aggressive bids that seem not
predictive and hence cannot be viewed as negative results; these equilibria are eliminated by excluding weakly dominated strategies.\footnote{For example, suppose that one agent bids infinitely high on all units, and all other bidders bid $0$.  The first agent then wins all items and pays nothing, regardless of the valuations.  This is technically a Bayes-Nash equilibrium {(and such a BNE exists even in the dominant-strategy Second Price Auction)}.  However, since the first agent's bid in this example is dominated by bidding her true marginal values, this equilibrium does not survive the elimination of weakly dominated strategies.}
In contrast, we establish that in the presence of a secondary market, even equilibria that exclude weakly dominated strategies
%
%survive such refinements \mbc{unclear what we mean by "such". Maybe be explicit?} 
can have vanishing welfare.}

\medskip

\noindent
\textbf{Theorem} (informal).
%\textbf{Theorem. } 
There exist instances of a uniform-price auction of $m$ items followed by a trade mechanism such that the expected welfare obtained in some perfect Bayesian equilibrium is only a
$\frac{1}{\Omega(\log(m))}$ fraction of the optimal expected welfare.  This is true even when restricting to equilibria that avoid weakly dominated strategies.

\medskip

What drives such bad equilibria?    Intuitively, if an individual participant wins many items %OLD: allowances
in the auction, they can use their market power to distort prices in the secondary market and increase their own revenue at the expense of efficiency.  This incentivizes a speculator to win many items %OLD: allowances
in the auction.  In particular, it can be perfectly rational for a speculator to bid aggressively at auction to ensure they have many items %OLD: allowances
 to work with in the secondary market.
%has implications on the auction as well: a speculator may rationally behave as though they have complementary preferences over allowances,
%\mbc{Brendan, is this still true? We slightly rewarded section 3, maybe match that here}, 
%since they can exert market power in the secondary market only if they win a large number of allowances at auction.  Unfortunately, the welfare guarantees of uniform-price auctions break down in the presence of complementarities \mbc{I thought that Example \ref{example:additional} shows that complementarities is not the issue that derive the result.}.  
%
{Moreover, this type of aggressive resale strategy will not necessarily be resisted by the other agents since bidders are incentivized to keep their bids low to reduce the price of items they are already winning.}
%
%%Moreover, since bidders are incentivized to shade their bids to keep prices low in uniform-price auctions (as higher bids lead to higher prices), this type of aggressive resale strategy will not necessarily be resisted by the competing agents \mbe{(as it will cause them to pay more for items won at the auction)}. 
%{as in a uniform-price auction more aggressive bidding in the auction will also increase their own auction payments}.  
%
{In Section~\ref{sec:example} we present an example illustrating all of these effects in a perfect Bayesian equilibrium in undominated strategies.  Our example needs only a very simple form of secondary market with dominant strategies, in which a reseller makes take-it-or-leave-it price offers.}
%Moreover, our equilibrium survives the elimination of weakly dominated strategies.}\footnote{In particular, while the speculator in our example bids very aggressively at auction, this behavior is not dominated.  }

%This unfortunate effect is driven by
%... \bjl{Still writing this sentence.}
%the rise of implicit complementarities \mbc{this is only one of the issues. Example 1 presents a different issue.} in agents' valuation functions.  \nsi{Brendan, can you add a sentence here with more intuition, something like the speculators need to drive price up and so inflate demand to crowd out other bidders and this is only valuable in aggregate and then maybe throw in some of those fancy words like this shows the benefit of stockpiling to influence prices which i quote above from an article... monopolist distortion, demand reduction}  

% smooth auctions have good welfare in combined markets
\paragraph{Auction Formats that are Robust to Aftermarkets}
The example above shows that the uniform-price auction is susceptible to the presence of an aftermarket.
We argue however that some market formats do have robust welfare guarantees in combined markets.  First, we show that {\em smooth} auctions maintain their welfare guarantees in such environments.    
Smoothness is a technical condition introduced by \citet{roughgarden2009intrinsic} in the context of proving worst-case guarantees on welfare properties of equilibria.  We formally define smoothness in Section~\ref{sec:smooth}, but for now
it is enough to know that if an auction format is $(\lambda,\mu)$-smooth for some $\lambda \in (0,1]$ and $\mu \geq 1$, then any equilibrium will generate at least a $\lambda/\mu$ fraction of the optimal expected welfare~\citep{syrgkanis2013composable}.
As noted by \citet{syrgkanis2013composable}, smoothness can be thought of as a sort of approximate First Welfare Theorem\footnote{{Informally, the First Welfare Theorem states that when prices clear the market, the allocation is socially efficient.}
%\bjl{TODO: say more about what the first welfare theorem is.}
}, whereby any loss in efficiency due to a reduced allocation to one bidder can always be partially offset by the payment of another bidder. %If a mechanism has the property that in any outcome, any participant can change her bid to receive her allocation of choice by paying the price paid at the current outcome, then the equilibrium outcome and prices are market clearing, implying that the outcome is socially optimal. Smooth mechanisms satisfy an approximate analog of this, requiring the property only in aggregate and only approximately (both in the value of the outcome achieved by the deviating bid and in the price paid), but not allowing the deviating bid to depend on the current actions of other players
Our result shows that any efficiency guarantee proven for an auction in isolation using smoothness will immediately extend to any equilibrium of the combined market, no matter what trade mechanism is used in the secondary market.  

\medskip

\noindent
\textbf{Theorem} (informal).  For any combined market consisting of a $(\lambda,\mu)$-smooth auction followed by a trade mechanism, the expected welfare 
% MB: I think it is better to have the stronger claim in the statement of the theorem. If you agree we also need to slightly revise the sentence that follows the theorem. 
{of any Bayes-Nash equilibrium (and thus, in particular,  of any perfect Bayesian equilibrium)} 
%OLD: of any perfect Bayesian equilibrium 
is at least $(\lambda/\mu)$ times the optimal welfare.

\medskip

%This efficiency guarantee holds not only at any PBE of the combined market, but in fact for the broader class of Bayes-Nash equilibria as well.
{We also show that this guarantee continues to hold even when agent valuations can be arbitrarily correlated, and even if participants can choose to acquire costly information about others' types in advance of the auction (such as a potential speculator investigating market forecasts).}

Uniform-price auctions are not $(\lambda,\mu)$-smooth for any constants $\lambda$ and $\mu$; 
{so the theorem does not apply for such auctions. 
Unlike for smooth auctions,
their welfare can be severely reduced in the presence of a secondary market (as we mentioned above). }
%OLD: as we mentioned above, their welfare can be severely reduced in the presence of a secondary market. 
{Thus, current emission allowances markets are not robust to speculation opportunities created by aftermarkets.}
%\nsi{commenting out below todos as we have run out of time} 
%\nsi{simple intuitive example, then state discriminatory is smooth citing appropriate source, then show how it breaks example to help reader see that it's smooth, state that this implies discriminatory price auctions would have good overall welfare, mention that in an extension we show this holds even with information acquisition; our main technical result is that all smooth auctions have good PoA} \bjl{Todo: describe at a high level why not, forward-point to a later section for example.}
%Alternative multi-unit auctions
%Simple variants of uniform-price auctions
%are, %\mbc{should we actually state these are "simple variants" of uniform-price? I think this works against us - we want to argue this is a possible yet can be viewed as a radical change (it was ruled out already by EU ETS) and thus we come up with a really simple variant -- adding reserves to the uniform auction }, 
%however, smooth.  For example, 
However, discriminatory-price auctions, in which each buyer pays her marginal bid for each unit she wins, are $(1-1/e,1)$-smooth.
%
%\mbc{before moving to the downside, maybe we should first say what is good about these smooth auctions.}
The smooth discriminatory-price auction therefore guarantees $1-\sfrac{1}{e}$ fraction of the optimal welfare in the combined market.\footnote{In the special case of a single item, the discriminatory-price auction reduces to the first-price auction. 
In asymmetric buyers setting, \citet{jin2022first} show that the equilibrium welfare of first-price auction can be at most $1-\frac{1}{e^2}$ fraction of the optimal when there are no aftermarkets.
In \cref{apx:single efficient}, we show that when there are two symmetric buyers,
any Bayes-Nash equilibrium of the combined market created by any trade mechanism that follows the first-price auction, is efficient.
}

{Is the discriminatory price auction a viable solution for allocating carbon allowances?}
Unfortunately, this auction format has some downsides. First, bidding is rather challenging and highly depends on distributional knowledge by the bidders. %One downside of discriminatory price auctions %Second is the possibility of regret and perceived unfairness, whereby bidders who bid ``too high'' are charged more than others.  
Second, discriminatory auctions might be perceived as unfair since identical goods are sold for different prices, creating envy between buyers. Finally, winners typically realize in retrospect that they could lower their payment by lowering their bids, creating regret. 
%\mbc{not only that, but more fundamentally, people pay different price for the same good, which might be considered unfair (envy). (also, regret exists in uniform auction for the price setting agent - so maybe remove or emphasize that almost all winning agents might have regret.)
%Maybe change the second point to "Secondly, discriminatory-price auctions might be perceived as unfair as identical goods are sold for different prices, creating envy between buyers. Finally, winners typically realize in retrospect that they could lower their payment by lowering their bids, creating regret.". We then might need to revise the next as well:}. 
These issues could discourage participation in the auction, and indeed such concerns have been cited as reasons why this auction format was not adopted by the EU ETS~\citep{euetsregulation}.  

\paragraph{An Alternative Solution: Posted Prices}
{Motivated by these concerns,} we also show that one can achieve robustness to secondary market distortions in another way.  Instead of running an auction, one could use \emph{posted prices}: make a quantity of items available at a declared price and allow buyers to purchase (in an arbitrary order) while supplies last.  This combines the fixed-price feature of a carbon tax with the quantity restriction of an auction.  
%The price offered would naturally depend on aggregate demand, so this method loses the price-discovery aspect of an auction.  But if the designer is aware of overall market conditions (in the form of {just one statistic:}
%{the expected maximal welfare of the market}),
%a prior distribution over valuations \mbc{it is actually enough to know the expected welfare (only one parameter, no need to know the entire distribution. Maybe change to ''in the form of a prior distribution over valuations, or some sufficient statistics of it'' ?)}), 
%then a recent line of work has demonstrated that simply posting an appropriate price achieves a welfare guarantee comparable to that of a uniform-price auction when running in isolation. 
%\mbc{this does not connect well with the next paragraph (repetitive) } %{Are such posted-price mechanisms robust to secondary markets?}
%Yet, \emph{are these posted-price mechanism robust to secondary markets?}  }

{
In general, posted-price mechanisms are not robust to secondary markets.
%is not enough to achieve the robustness we are looking for.  
We show by way of example that even if a posted-price mechanism achieves high welfare on its own, this welfare can be significantly decreased by speculation that occurs at equilibrium in the presence of a secondary market.  This motivates us to focus on a particular form of posted-price mechanism: those that use \emph{balanced prices}, which are set proportional to the expected average welfare generated in the efficient allocation.}  We define balanced prices formally in Section~\ref{sec:balanced-prices}.  It is known that balanced prices yield strong welfare guarantees for many allocation problems, including multi-unit auctions~\citep{feldman2014combinatorial,dutting2020prophet}.  Specifically, if prices are $(\alpha,\beta)$-balanced for $\alpha,\beta \geq 1$, then the expected welfare obtained when buyers purchase sequentially is at least a $1/(1+\alpha\beta)$ fraction of the optimum~\citep{dutting2020prophet}.  We prove that the welfare guarantee from balanced prices continues to hold at any equilibrium given any arrival order of the buyers even in the presence of a secondary market.\footnote{Unlike with smooth auctions, this result does not extend to {correlated value distributions or to} 
settings where buyers can acquire additional information before the auction opens.  For example, it is problematic if the buyers become more informed than the designer who set the prices.}

\medskip

\noindent
\textbf{Theorem} (informal).
%\textbf{Theorem. } 
For any combined market consisting of a posted-price mechanism with $(\alpha,\beta)$-balanced prices followed by a trade mechanism, the expected welfare 
{of any Bayes-Nash equilibrium (and thus, in particular,  of any perfect Bayesian equilibrium)} 
%of any perfect Bayesian equilibrium 
is at least $1 / (1+\alpha\beta)$ times the optimal welfare. 
%\mbc{consider strengthening the theorem to talk about BNE, as in the prior theorem.} 

\medskip

%Similar to the result for smooth auctions, our guarantee actually holds for the more general class of Bayes-Nash equilibria.  
For our setting of selling identical items to buyers with decreasing marginal values, there is a per-item price that is $(1,1)$-balanced, and therefore guarantees at least half of the expected optimal welfare even in the presence of a secondary market.

\paragraph{A Proposal: Balanced Reserves}
To this point we have described two methods for allocating items %\mbc{should we talk about "items"? In general, we need to search for "allowances" and decide if we want to replace at each instance} 
that achieve robust welfare guarantees in the presence of secondary markets.  One is to use a smooth mechanism, such as a discriminatory auction, and the other is to find a balanced price and sell items at that price while supplies last.  We have already discussed potential drawbacks of the former solution, but the latter has its own set of practical challenges: for one thing, it is a dramatic change relative to the uniform-price auctions typically used; for another, if the government underestimates demand and sets its price too low, this could encourage a rush where buyers race to purchase items at the moment they become available, resulting in low welfare, buyer frustration,
%\mbc{I think the claim about inefficiency must be removed - it is contradicting the POA result for posted pricing (for any order of arrival). We can instead say "unfair allocation" or "allocations that are perceived as unfair" }\bjl{The claim is about what happens if the government sets its price too low, so it doesn't contradict POA results.} 
and a perception of unfairness 
%\mbc{more impotently, I think it results with inefficient outcome as items are allocated in FIFO and not by value}.  
As it turns out, one can address these issues and obtain all of the benefits of balanced posted prices with a small tweak to a uniform-price auction -- the introduction of appropriately-chosen reserves.  Reserve prices are already common in many emission auctions, such as the one administered by the California Air Resources Board~\cite{carb}, in the form of price floors.
%, typically introduced to prevent the auction from running at a loss and preventing allowances from being sold at a price significantly cheaper than the social cost of the carbon they represent. 
We show {that an appropriate choice of reserve prices}
%also 
guards against welfare loss in combined markets: 
{namely, one can augment a uniform-price auction with a per-allowance reserve price with bounded welfare loss, by setting the reserve to be}  
%OLD: namely one can augment a uniform-price auction with a per-allowance reserve price set equal to 
the balanced per-allowance price one would use in a posted-price mechanism.   
We prove in Section~\ref{sec:reserves} that the expected welfare at any equilibrium of this auction is at least half of the expected optimal welfare, and this guarantee persists in the presence of an arbitrary secondary market.  We view this as a practical solution that can be implemented with minimal effort: as long as the government can estimate {just one statistic, }
the expected average social value of 
{a carbon allowance, } %MB: this is the average for a single unit
%OLD: carbon licenses, 
they can mitigate the impact of speculation and other equilibrium effects of resale by {employing an appropriately-determined auction reserve.}
%price is at least half of that estimation \mbc{why "at least"? the claim will not hold if prices are too high. We also did not define "that estimation". So maybe simply say "the auction per-unit reserve price is appropriately determined."}.  
We further show that the welfare guarantees degrade gracefully as one adjusts the prices, % \mbc{this sounds right for underpricing, but not for overpricing}, 
meaning that unavoidable misspecifications in the price determination will have a modest effect on the welfare guarantees.

\subsection{Additional Related Work}
\label{sec:related}

\paragraph{Equilibria of Combined Markets} The challenges in analyzing equilibria in combined markets was acknowledged in \citet{haile2003auctions} due to the fact that there exist endogenously induced common value components in the auction.
%In addition, the author showed that the revenue equivalence result may fail when there exist secondary markets.
In the simple single-item setting with winner posting prices as secondary markets,\footnote{In this model, the authors also assume that no information, especially the bids, are revealed in the secondary market to avoid the ratchet effect.}
\citet{hafalir2008asymmetric} characterized the equilibrium behavior of the agents in the combined market,
and \citet{hafalir2009revenue} adopted the characterization to show that
the expected welfare of the first-price auction with secondary markets
may decrease by a multiplicative factor of $\sfrac{2e}{(2e-1)}$.
In addition to the above discussions,
there are many papers
discussing various properties of the resale model in the economics literature,
including but not limited to the observation of bid shading in the auction \citep{pagnozzi2007bidding},
and the revenue ranking of the simple auctions \citep{lebrun2010revenue}.
See the survey of \citet{susin2017auctions} for more discussions on the equilibrium properties of the resale model.
Finally, there are several recent papers focusing on designing optimal mechanisms when the seller has no control over the secondary market.
\citet{carroll2019robustly} show that second price auction with reserve prices is the robustly revenue optimal mechanisms with unknown resale opportunities.
\citet{dworczak2020mechanism} considers the design of information released to the secondary markets and show that the information structure that induces truthful behaviors are cutoff rules.
He also provides sufficient conditions for simple information structure to be optimal.

\paragraph{Sequential Auctions} One closely related line of theoretical work is price of anarchy for sequential auctions, which also study subgame perfect equilibrium outcomes \citep{syrgkanis2012bseq,leme2012sequential}. In \citet{leme2012sequential}, the authors illustrate that although price of anarchy of the sequential composition of first-price auction is small for unit-demand agents, the result breaks for agents with submodular valuations, and the price of anarchy can be unbounded in the latter case. In contrast, our results indicate that for submodular valuations, a simultaneous first-price auction followed by any trade mechanism
will have constant price of anarchy for the combined market. The main difference that allows us to handle combinatorial auctions in sequential auction format is that all items are sold only in the first market, and the secondary market is only providing the platform for agents to retrade the items, rather than selling items sequentially, with each item sold once in one of the auctions.
Recently, \citet{eden2020price} bound the price of anarchy when each agent is subject to an externality from the allocation of the other agents. The authors motivate the externality by the resale model where those resale behaviors are assumed to be fixed exogenously, 
which is substantially different from our model where 
agents behaviors depend on the format of the combined market.

\paragraph{Price of Anarchy}
As discussed earlier, our techniques leverage smoothness and balanced pricing.  
%The smoothness framework is a powerful tool for analyzing the price of anarchy in auctions (see \cite{roughgarden2017price} for a detailed discussion on the literature).
%This framework is first proposed in \cite{roughgarden2009intrinsic} for complete information games and \cite{roughgarden2012price,syrgkanis2012bayesian} for incomplete information games. The idea of smoothness is further refined and generalized for simultaneous composition and sequential composition of smooth mechanisms in multi-item settings \citep[e.g.,][]{syrgkanis2013composable, feldman2013simultaneous}.  
For multi-unit auctions in particular, this theory has been used to derive equilibrium welfare bounds for different auction formats in isolation~\citep{de2013inefficiency,markakis2015uniform}, and we extend this analysis to settings with aftermarkets. The balanced pricing framework is a general approach for designing posted-price mechanisms in a broad class of allocation problems~\citep{kleinberg2012matroid,feldman2014combinatorial,dutting2020prophet}.  These constructions employ the theory of Prophet inequalities to bound the welfare obtained when buyers sequentially purchase their preferred bundles at the proposed prices, which are calculated using the distribution of buyer values.  Similar to smoothness, we extend the existing analysis to show that the welfare guarantees attainable through balanced pricing extend to settings with aftermarkets.
%\nsi{add balanced pricing refs or alternatively only cite these things in the sections where we use them}

\paragraph{Carbon Markets} There is a rich literature exploring market and regulation-based techniques for reducing emissions and their effectiveness.  Here we discuss a small sampling of this literature, referring the reader to many excellent overviews such as \cite{Cotton15} or \cite{cramton2017global} for further details. 
%\mbc{We cannot start a sentence with a citation that looks like "\cite{weitzman74} asks". Should appear as ''Weitzman[1974] asks''. Yingkai, please fix all citations that start sentences.} 
\citet{weitzman74} asks whether it is better to control emissions via imposing standards (quantity regulation) or charging taxes (price regulation), and notes that prices tend to fare better when the social cost of emissions is close to linear, whereas quantity regulation can be preferable in the face of uncertainty when marginal costs are variable.  
We note that taxes share many similarities with the posted-pricing mechanism we study. \citet{KerrCramton98}, in turn, propose selling emission allowances in an auction (they suggest an ascending auction).  Their paper explicitly suggests these allowances be tradeable in aftermarkets to maximize liquidity.  \citet{goldner2020} study uniform-price auctions with price floors and ceilings, a common mechanism in practice, and prove welfare guarantees under certain conditions in the absence of aftermarkets. Our paper complements these by exploring the interplay of these aftermarkets and the primary auction and stating conditions under which welfare guarantees extend to the combined market.

%\pagebreak

\section{Preliminaries}
\label{sec:prelim}
%\mbc{The model is presented as extremely general, allowing abstract feasibility set and allowing for agents to have externalities. But do we have any result beyond the homogeneous items, no externalities  case? If not, maybe the model should simply be the multi-unit model, with the single item as a special case. ADDED: I guess the general result about smooth mechanism does not depend on specific settings. Can we find any example that is not multi-unit for which we can use the result. How about combinatorial  auctions? Any setting with externalities (like public good) ?}
%\yl{In section 3 we have a sentence saying that for combinatorial auction with for example monotone submodular valuations, 
%then simultaneous first price auction is smooth.}\mbc{need to add all the other corollaries as explicit statements. }

%\bjl{I recommend reordering the content in the preliminaries section as follows: (1) allocation problems, utility model, examples of particular allocation problems of interest.  (2) mechanisms, BNE, PoA, Revenue benchmark.  (3) Our framework for secondary markets: $\mechf$, $\mechs$, $\mechc$.  Assumptions on $\mechs$, solution concept for $\mechc$ (subgame perfection).}\mbc{I have made a major revision along these lines.}

\subsection{The Basic Setting}
\label{sec:basic}

For clarity we begin by describing a basic model focused on a multi-unit auction of identical items.  Our results actually apply to a more general model 
of combinatorial auctions
and different information structures (including the ability to purchase information about aggregate demand); we describe these extensions in Section~\ref{sec:extensions}.

\subsubsection{The Allocation Problem}

%We model a carbon allowances allocation problem
%ETS \mbc{remove "ETS"? or remind the reader what this stands for (why is it here?). Maybe better state this as: "We model the carbon allowances allocation problem as ... "} as an allocation problem 
%in which a government agency 
A seller initially holds a set $\items$ of $m$ identical items to be allocated among a set $\agents$ of $n$ buyers.  
%consider a seller holding a set of $m$ % indivisible? 
%items and $n$ agents (buyers)
%interested in these items. The seller has no value for the items. 
A feasible allocation is a profile $\mathbf{x} = (x_1, \dotsc, x_n)$, where $x_i \in [m]$ is the number of items obtained by agent $i$ and $\sum_i x_i \leq m$.  We write $\feasibles$ for the set of feasible allocations.

Buyer $i$ has a private valuation function $v_i \colon [m] \to \reals_{\geq 0}$
%$\val_i : \reals \to \reals$ \mbc{this allows for fractional allocation and negative values. Is this on purpose? I see we want randomization, but we can have the utility for deterministic allocation being defined for natural quantities, and utilities for lotteries being the expectation. E.g. when we define the marginal per item, we do not explicitly say that adding half an item adds half the marginal to the utility - it is derived from the expectation.   }\bjl{Not on purpose.  We can define as $v_i \colon [m] \to \reals_{\geq 0}$.}
where $\val_i(x_i)$ denotes buyer $i$'s value for obtaining $x_i$ items, normalized so that $\val_i(0) = 0$.  We emphasize that this is a consumption value.
%, driven by the ability to emit pollution as afforded by the licenses.  
Valuations are assumed to have non-increasing marginal valuations: for each $j \geq 1$, $v_i(j) - v_i(j-1)$ is non-negative and weakly decreasing in $j$.  We will sometimes refer to $v_i$ as the type of agent $i$. We write $\types = \times_i\types_i$ for the set of valuation profiles. {We assume that $\val_i$ is sampled independently from a known distribution $\dist_i$, and denote the prior product distribution over the valuations by $\dist = \times_i\dist_i$.}
%OLD: and $\dist = \times_i\dist_i$ be the prior distribution over the valuations.
%The set of feasible allocations of the items, denoted by $\feasibles$
% We assume $\feasibles$ is a finite set. 
%Let $\types = \times_i\types_i$ be the type space of the agents
%For any agent $i$, her type $\type_i\in \types_i$ is her private information
%and the marginal distribution of her private type is denoted by $\dist_i$. 
%Let $\val_i(\alloc; \type_i)$ be the public valuation function of agent~$i$
%given allocation $\alloc \in \feasibles$ and type $\type_i$,
%and her 
The utility of agent $i$ given allocation $\alloc_i$
and total payment $\pay_i$
is $u_i(\alloc_i, \pay_i) = \val_i(\alloc_i) - \pay_i$.
Buyers are assumed to be risk-neutral and seek to maximize expected utility.
%Moreover, all agents have von Neumann–Morgenstern expected utility for randomized outcomes, valuing it at its expectation. 
%Note that here we allow agents to have externalities (the valuation of an agent might depend on the entire allocation, not only her own allocation). % \mbc{do we ever present any result with that generality? we should add such results}

%We use $\val$ to denote the vector of valuation functions. 
The \emph{welfare} of an allocation $\mathbf{x} \in \feasibles$ when the valuations are $\mathbf{v}$
%and the types are $\type\in \types$ 
is defined to be 
$\wel(\mathbf{v},\mathbf{x}) = \sum_i v_i(x_i)$.
% \bjl{The following definitions are for general mechanisms, and aren't specific to our secondary market framework.  I recommend making all of the general definitions first, then having a section dedicated to our framework where we define $\mechf$, $\mechs$, and $\mechc$.  (See proposed outline at the beginning of the section.)}
For any valuation profile $\mathbf{v}$, 
let $\wel(\mathbf{v},\feasibles) = \sup_{\alloc\in \feasibles} \wel(\mathbf{v},\alloc) $ be the optimal (highest) welfare
given the valuation functions $\mathbf{v}$ and feasibility constraint $\feasibles$.
We say an allocation is \emph{efficient} if it achieves the optimal welfare. %\mbc{it is defined as the supremeum, so the optimum might not be obtained (but it is for the case we consider, i think). Maybe remark?} % MB: we have never defined the meaning of "efficient"
Let $\wel(\dist, \feasibles) = \expect[\mathbf{v}\sim\dist]{\wel(\mathbf{v},\feasibles)}$ be the expected optimal welfare.
When $\feasibles$ is clear from the context, 
we omit it in the notation and use $\wel(\mathbf{v}), \wel(\dist)$ to denote the optimal welfare and expected optimal welfare, respectively.

\subsubsection{Mechanisms}

% BJL: Ran out of time to play with the ordering.
%\bjl{I'm still not completely sure about the ordering here. Currently all of the mechanism formalization is presented first, then particular instantiations come after.  It might be clearer to interleave these more, so we can explain the mechanism formalization with examples as we go.  But for now I've left it as is.}

%\bjl{We should cut some of the mechanism design formalisms.}

%In our model, 
Agents can acquire items by participating in an auction then trading among themselves in a secondary market.  We will formally describe both the auction and the secondary market 
%\mbc{maybe also mention the combined market?} 
as \emph{mechanisms}.
%
%\bjl{Decided to stick with the single-round mechanism definition here.  Will discuss multi-round issues later.}MB: good.
%We model auctions using the theory of mechanisms.  
%A mechanism $\mech$ defines 
%a set of actions for each agent, and a (possibly random) mapping from profile of actions to a feasible allocation and payment from each agent. 
Formally, a mechanism $\mech =(\alloc^{\mech}, \pay^{\mech}): \actions \to \Delta(\feasibles \times \reals^n)$ is defined 
by an \emph{allocation rule} $\alloc^{\mech}: \actions \to \Delta(\feasibles)$ and a \emph{payment rule}
$\pay^{\mech}: \actions \to \reals^n$,  
where $\actions = \times_i \actions_i$
and $\actions_i$ is the action space of agent $i$ in the mechanism.
Thus, for action profile $\mathbf{a} =(a_1, a_2,\ldots,a_n)\in (A_1, A_2,\ldots,A_n)= A$ the \emph{outcome}
of the mechanism is the (randomized) allocation $\alloc^{\mech}(\mathbf{a})$, and each agent $i$ is charged (in expectation) a \emph{payment} of $\pay^{\mech}_i(\mathbf{a}) \geq 0$.  
%For example, in an auction each agent's action is typically a bid, and the allocation and payment rules describe how to use the submitted bids to determine how many licenses are won by each buyer and at what prices.  
The utility of agent $i$ with valuation $v_i$
%type $\type_i$ 
when participating in the mechanism $\mech$ in which agents take actions $\mathbf{a} \in \actions$ is  
$u_i(\mech(\mathbf{a})) = \val_i(\alloc^{\mech}(\mathbf{a})) - \pay^{\mech}_i(\mathbf{a})$.
% We sometimes 
%$\mech: \asetfirst \to \Delta(\feasibles \times \reals^n)$

A mechanism $\mech$ with valuation distribution $\dist$ defines a game. A \emph{strategy} $\strategy_i: v_i \to \Delta(\action_i)$ for agent~$i$ is a mapping from her valuation $v_i$ to a distribution over her actions.
%(in private information settings we consider, $\strategy_i$ may not depend on $\type_{-i}$, but may depend on any other information, as the distribution $\dist$). 
With slight abuse of notation denote by $\strategy_{-i}(\val_{-i})$ the profile of actions taken by agents other than $i$ when each $j\neq i$ has valuation $\val_j$.
%of type $\type_j$ takes action $\strategy_{j}(\type_{j})$.
A strategy $\strategy_i$ is a \emph{best response} for agent $i$ given strategies of the others $\strategy_{-i}$
if for any strategy $\strategy_i'$ it holds that 
$\expect{u_i(\mech((\strategy_i(\val_i), \strategy_{-i}(\val_{-i})))}\geq  
\expect{u_i(\mech((\strategy_i'(\val_i), \strategy_{-i}(\val_{-i})))}$
for every valuation $\val_i$, where the expectation is over the valuations of the other agents as well as any randomness in the mechanism and strategies. 
%\mbc{We say that the expectation is over $\dist$, so it is in expectation over types $\type_i$, not for every realized $\type_i$. This relates to the proof of the 2 iid FPA result (see there) }\yl{changed to interim version} MB: add the randomness in strategies and ,mechanism
A profile of strategies $\strategy=(\strategy_1,\ldots,\strategy_n)$ is a \emph{Bayesian Nash equilibrium} (BNE) for mechanism $\mech$ with distribution $\dist$, if for every agent $i$, strategy $\strategy_i$ is a best response for agent $i$ given strategies of the others $\strategy_{-i}$.

By slightly overloading the notation, 
we also denote $\wel(\mech, \strategy, \dist)$ 
as the expected welfare obtained in mechanism $\mech$ using equilibrium strategy profile $\strategy$. 
Let the \emph{price of anarchy of mechanism $\mech$ within the family of distributions $\dists$} be
\begin{align*}
\poa(\mech, \dists) = 
\sup_{\dist\in \dists} \frac{\wel(\dist)}{\inf_{\strategy\in \BNE(\dist, \mech)}\{\wel(\mech, \strategy, \dist)\}}
% OLD (using max and min): \max_{\dist\in \dists} \frac{\wel(\dist)}{\min_{\strategy\in \BNE(\dist, \mech)}\{\wel(\mech, \strategy, \dist)\}}
\end{align*}
%\mbc{I have changed max to supremum and min to infimum as these sets might be infinite and the supremum might not be obtained. Same for Revenue. }
where $\BNE(\dist, \mech)$ is the set  of Bayesian Nash equilibria given distributions $\dist$ and mechanism~$\mech$.

\paragraph{Auctions}

%\bjl{This can be shortened}
%There are several auction formats that are of interests in this paper. % For single item and multi-item auction, we consider direct revelation mechanisms with the action space of each agent being her value space. For combinatorial auctions, agents are submitting a bid for each item. 
%\begin{itemize}
%\item Single-item auctions: Each agent $i$ simultaneously submit bid $a_i\in A_i=\reals_{\geq 0}$ representing her value for the item, and the item is sold to agent $i^*\in\arg\max_i a_i$, with arbitrary tie breaking. Payments are set as follows:  
%	\begin{itemize}
%	\item
%	\emph{first-price auction}:	The highest bidder $i^*$ pays her bid $a_{i^*}$, and any agent $i\neq i^*$ pays $0$.

%	\item
%	\emph{second-price auction}: 
%	The highest bidder $i^*$ pays the second highest bid $\max_{i\neq i^*}a_i$, and any agent $i\neq i^*$ pays $0$.
	
%	\item
%	\emph{all-pay auction}: 
%	Every agent $i$ pays her bid $a_i$.
%	\end{itemize}
%\item 
We can describe multi-unit auctions as mechanisms.
%For example, in the multi-unit auctions we consider, % the action space is $A=\reals_{\geq 0}^{n\times m}$. 
For example, in most common multi-unit auctions, an action of bidder $i$ is a \emph{bid} 
of the form $(a_{i1},\dots,a_{im})\in A_i= \reals_{\geq 0}^{m}$ representing her $m$ marginal values, with $a_{i1}\geq \dots \geq a_{im} \geq 0$.  The agents simultaneously declare these bids to the auctioneer.  The $n \times m$ received marginal bids are then sorted from largest to smallest, and the
$m$ identical items are greedily allocated to the bidders of the $m$ highest marginal bids (breaking ties arbitrarily).  The two most common auction formats, uniform and discriminatory, then differ in how payments are calculated:
	\begin{itemize}
	\item
	\emph{discriminatory auction}: Each agent pays her winning bids.  That is, if agent $i$ wins $x_i$ items, then she pays her highest $x_i$ marginal bids: $p_i = \sum_{j=1}^{x_i} a_{ij}$. 
	\item \emph{uniform-price auction}: A common price $p$ per unit is chosen, and each agent pays $p$ for each license won.  This price $p$ is taken to lie between the $(m)$-th highest marginal bid and the $(m+1)$-st highest (within the $n \times m$ reported marginals).  In other words, $p$ lies between the highest losing bid and the lowest winning bid. In this paper we focus on the case where $p$ is the highest losing bid.  %\mbc{don't we want to pick a specific price "for the purpose of this paper" ? BTW - are the POA bounds (both in isolation and in  combined markets) hold for all of these prices?}
%	Agent $i$ winning $x_i$ units of items pays for each unit  
%	the $(m+1)$-th highest marginal bid within the $n\times m$ reported marginals. \mbc{ this is one version (corresponding to the SPA). Need to check if the version with price being the $m$-th highest marginal bid has good POA.}
	\end{itemize}
%Auctions also sometimes impose further restrictions on the bids that can be submitted.  An auction with \emph{standard bidding} imposes no restrictions aside from non-increasing marginal values.  This is in contrast to \emph{uniform bidding} which requires that each agent submit a bid in which all non-zero marginal values are equal to each other.  {In this paper we will focus on standard bidding, which is the format used by the EU ETS.}
%%This corresponds to placing a real-valued bid $b_i$ per unit received, up to a specified maximum number of licenses $q_i$.   Both standard bidding and uniform bidding can be combined with either discriminatory or uniform pricing rules.\mbc{consider removing "uniform bidding" (if it is not used by EU ETS).} \mbc{If we keep both - are all the claims we make independent of the two bidding options? If so, we need to say that explicitly: "Our bounds on the welfare of combined markets with primary auction being a uniform-price auction holds for both standard bidding and uniform bidding, and hold for any clearing price used in the auction (price between the $m$-{th} and $(m+1)$-st highest bid)". }

We will also be interested in a simple form of posted-price mechanism.  In the multi-unit setting, the auctioneer selects in advance a price $p$ per item.  The buyers are then approached sequentially in a fixed order.  Each buyer can choose to buy as many items as desired, up to the amount remaining, at a price of $p$ per item.  That is, if the buyers purchase in the order $1, 2, \dotsc, n$, then each bidder $i$ can choose any non-negative $x_i \leq m - \sum_{j < i}x_j$ and pays $px_i$.  Once all items are sold the mechanism ends.

%\item Combinatorial auctions: 
%	\begin{itemize}
%	\item
%	\emph{simultaneous first-price (all-pay) auction}: run first-price (all-pay) auctions simultaneously for all items $j\in[m]$.
%	\end{itemize}
	
%\end{itemize}

\paragraph{Secondary Markets}
Informally, a secondary market allows users to trade items that they obtained from the auction.
%bring to the market.  
%obtained \mbc{we do not allow items that are obtained from any other source - only from the auction} elsewhere, such as from an auction.  
For example, one might imagine that agents could offer to sell their items at a certain price, and other agents might choose to purchase at the suggested price (or not).  Similar to the auction, we will model the secondary market as a mechanism.  The starting point of the secondary market is the allocation picked by the auction, which is publicly revealed.  The secondary market is therefore parameterized by an allocation $\alloc \in \feasibles$, which we think of as the auction outcome.  We will tend to use $\mechs$ to refer to secondary market mechanisms, and in a slight abuse of notation we will write $\mechs(\mathbf{a}; \mathbf{x}) = (\alloc^{\mechs}(\mathbf{a}; \mathbf{x}), \pay^{\mechs}(\mathbf{a}; \mathbf{x}))$
%$\mech: \actions \times \feasibles \to \Delta(\feasibles \times \reals^n)$ 
for the allocation and payment rules of a secondary market $\mechs$ as a function of the initial allocation $\mathbf{x} \in \feasibles$.\footnote{In some secondary market formats it is more natural to think of actions being chosen sequentially rather than simultaneously.  E.g., in the example above, a seller first chooses a price then buyers choose whether to purchase.  One can model this by having a buyer's ``action'' be a mapping from all possible observations (e.g., prices) to a realized action (e.g., whether to buy).}
%, which assigns for each profile of actions of the agents in the secondary market $\asecond =(\asecond_1, \asecond_2,\ldots, \asecond_n)\in (\asetsecond_1, \asetsecond_2,\ldots, \asetsecond_n)= \asetsecond$ (where $\asetsecond_i$ is the space of actions of agent $i$ in $\mechs$) and auction allocation $\alloc\in \feasibles$, a distribution over the final allocation and a second round of transfer payments.  Note that the action spaces of the two mechanisms can be different, but they share a common set of feasible allocations.

To capture our intuitive notion of a secondary market, we will introduce two mechanism properties that we will assume in secondary markets we consider.  First, we assume voluntary participation, which informally means that each agent can choose not to participate.  More formally, voluntary participation requires that each agent has an ``opt-out'' action that guarantees their utility is not reduced by the secondary market.

\begin{definition}
A secondary market $\mechs$ satisfies \emph{voluntary participation} if for each agent $i$, all valuations $\val_i$, and all feasible allocations $\mathbf{x}$, there exists an action $a^*_i$ such that, for any action profile $\mathbf{a}_{-i}$ of the other agents,
$\val_i(x_i) \leq u_i(\mechs((a^*_i, \mathbf{a}_{-i}); \mathbf{x}))$.
\end{definition}

%\begin{definition}
%A secondary market $\mechs$ satisfies \emph{voluntary participation} if for each agent $i$ and all valuations $\val_i$, 
%there exists an action $\asecond_i \in \asetsecond_i$
%such that $\val_i(\alloc) \leq u_i(\mechs(\alloc, (\asecond_i, \asecond_{-i})))$ 
%for any allocation $\alloc$
%and any action profile $\asecond_{-i}$.
%\end{definition}
We argue that this condition is quite mild.
For example, if the secondary market is one in which license holders can suggest take-it-or-leave-it prices, and trade happens if another user agrees to trade at that price, then a license holder might decide not to make an offer (``not participate"), and 
other agents can decide to decline any offer made (again, ``not participate"). Each agent can therefore ensure that her utility in the secondary market is the same as the utility obtained from the initial allocation $\mathbf{x}$.

We also assume that our secondary market satisfies weak budget balance, which means that it does not run a deficit with respect to payments.

\begin{definition}
A secondary market $\mechs$ satisfies \emph{weak budget balance} if 
$\sum_i \pay^{\mechs}_i(\mathbf{a}; \mathbf{x}) \geq 0$
for any action profile $\mathbf{a}$ and feasible allocation $\mathbf{x} \in \feasibles$.
\end{definition}

A mechanism that satisfies both voluntary participation and weak budget balance is called a \emph{voluntary-non-subsidized-trade mechanism}, or a \emph{trade mechanism} for short. 

\iffalse
In a slight abuse of notation we will write $\mechs: \asetsecond \times \feasibles \to \Delta(\feasibles \times \reals^n)$ for the allocation and payment rules of secondary market $\mechs$ as a function of the initial allocation, which assigns for each profile of actions of the agents in the secondary market $\asecond =(\asecond_1, \asecond_2,\ldots, \asecond_n)\in (\asetsecond_1, \asetsecond_2,\ldots, \asetsecond_n)= \asetsecond$ (where $\asetsecond_i$ is the space of actions of agent $i$ in $\mechs$) and auction allocation $\alloc\in \feasibles$, a distribution over the final allocation and a second round of transfer payments.  Note that the action spaces of the two mechanisms can be different, but they share a common set of feasible allocations.
\fi

\subsubsection{The Combined Market}

%\bjl{Given the intro, I'd recommend we describe the timing of the combined game earlier, then make note of the different auction and secondary market rules.  Goal: model section in 2 pages.}

% After information acquisition is complete the mechanism will begin.
We are finally ready to formally model our setting of an auction followed by a secondary market.
%We consider a setting where an auction mechanism is followed by a secondary market in which agents can trade the goods. 
We model this scenario as a two-round game $\game$ that consists of two mechanisms, $\mechf$ and $\mechs$, run sequentially.  We refer to $\mechf$ as the \emph{auction} and $\mechs$ as the \emph{secondary market}. 
%Our game formulation also allows signals to be revealed between the auction and the secondary market.  

In the first round of the game, the agents participate in the auction $\mechf$.
We denote the action space of $\mechf$ by $\asetfirst = \times_i \asetfirst_i$.
The agents simultaneously choose actions $\afirst \in \asetfirst$, resulting in outcome $\alloc^{\mechf}(\afirst)$ and payments $(\pay_i^{\mechf}(\afirst))_i$.  Each agent observes the outcome of the auction and her own payment.

The second round then starts and the agents participate in the secondary market $\mechs$.  The allocation $\alloc^{\mechf}(\afirst)$ from the auction is used as the initial allocation in the secondary market.
We denote the action space of $\mechs$ by $\asetsecond = \times_i \asetsecond_i$.
Note that the action spaces of the two mechanisms can be different, but they share a common set of feasible allocations.

%\bjl{This description doesn't include the more general signaling; instead I added a remark at the end.  Alternatively we could define the game with respect to general signals, then impose the assumption that the signal includes outcomes and payments.  Not sure yet which is cleaner.}\yl{I think it would be better to add the signal explicitly between step 2 and 3.} \mbc{I agree. And we also need to add a description of that signal in the paragraph above. The signal is simply a mapping from $\afirst$ to a message space.}
%\yl{added signal structure} \bjl{I think the signal actually is more general than a mapping from $\afirst$, since it can also depend on randomness inside the mechanism.  I modified the description so that $s_i$ is now a random variable correlated with the input and output of $\mechf$.  Would it also make sense to just write $\alloc^{\mechf}$ instead of $\alloc^{\mechf}(\afirst)$, etc., to clean up the notation...?}\mbc{If we state that the auction is also outputting the signals, then the signal functions only need the input actions, as it has all information about the random mapping from auctions to distribution over outcomes. }

To summarize, the timing of the two-round game 
$\game(\mechf, \mechs)$ proceeds as follows:  
\begin{enumerate}
    \item Each agent $i$ picks an action $\afirst_i\in \asetfirst_i$ simultaneously. Mechanism $\mechf$ runs on actions $\afirst$. % The agents simultaneously choose actions $\afirst$.  
%    \item Each agent $i$ observes a signal $s_i \in S_i$ that is assumed to contain $\alloc^{\%mech_1}(\afirst)$ and $\pay^{\mechf}_i(\afirst)$.
    \item Each agent $i$ observes $\alloc^{\mechf}(\afirst)$ and $\pay^{\mechf}_i(\afirst)$.
    \item Each agent $i$ picks an action $\asecond_i\in \asetsecond_i$ simultaneously. 
    %\mbc{the "take-it-or-leave-it" secondary market is not really simultaneous (but clearly can be converted to such). Should we ignore this issue or maybe add a footnote ?}\bjl{Maybe add a footnote when describing the take-it-or-leave-it secondary market in Section 3.} \mbc{sounds good to me}
    %
%    The agents simultaneously choose actions $\asecond$.  \bjl{Removed ``depending on...'' since this is implied by the game structure, and also depends on the equilibrium notion.}
      Mechanism $\mechs$ runs on actions $\asecond$, starting from allocation $\alloc^{\mechf}(\afirst)$.
    \item The total payoff to agent $i$ in the combined market is $u_i(\mechs(\asecond;\alloc^{\mechf}(\afirst))) - \pay_i^{\mechf}(\afirst)$.\footnote{Note that this expression includes the payments from both the auction and the secondary market, as the secondary market transfers are included in the utility term.}
\end{enumerate}

Note that an instance of the two-round game $\game(\mechf, \mechs)$ naturally corresponds to a combined mechanism $\mechc$ which we denote by $\combinedGameSimple$, in which an action has two components: (1) an action $\afirst_i \in \asetfirst_i$ for the auction mechanisms, and (2) a mapping for each agent $i$ from the tuple of (allocation, payment) from the auction into an action for the secondary market.

\paragraph{Equilibria} The notions of BNE and PoA extend to a combined mechanism $\mechc$ as before.  We note that since an action of $\mechc$ encodes actions for both markets, the definition of BNE does not require that agents are best-responding in the secondary market given the auction's outcome.  This enables strategies with non-credible threats like ``refuse to trade with any competitors in the secondary market, even if it is beneficial to do so, unless they let me win at least $7$ items in the auction."  For multi-round games like ours, it is standard to refine the BNE solution concept and instead consider the more restrictive class of perfect Bayesian equilibria (PBE).
%
%\subsubsection{Equilibrium Refinements}
%
%Since our combined market is an extensive-form game that proceeds in multiple rounds, one might wonder about natural refinements of the BNE solution concept such as perfect Bayesian equilibrium (PBE).
%or the stronger refinement of sequential equilibrium.  
We formally define this equilibrium notion in Appendix~\ref{sec:sequential}.  Roughly speaking, a PBE requires (a) subgame perfection, where behavior in the secondary market is always rational given any auction outcome, and (b) that agents accurately update their beliefs after the auction outcome and behave in accordance with those beliefs in the secondary market.  As it turns out, our positive results about welfare hold at \emph{any} BNE, whether or not they satisfy these requirements.  So in particular our welfare bounds also hold for any 
PBE
%sequential equilibria 
as well.  Moreover, each example we use to illustrate a negative result will {not only be a BNE, but rather also}
be a PBE.  In fact, the perfect Bayesian equilibria we consider will also satisfy the stronger conditions of sequential equilibria; see Appendix~\ref{sec:sequential} for further discussion.

% \bjl{On second thought, I removed the definition of PBE and instead expanded the discussion and gave an example illustrating a non-PBE type of strategy.}

%be a sequential equilibrium.

%\bjl{The following is not integrated yet into the main notation, since I suspect the section above will change to make it shorter.}
%
%A profile of strategies $\sigma$, together with a belief function $B \colon \Gamma \to \Delta(\Theta)$ for the receiver mapping every observation to a distribution over the profile of types, form a perfect Bayesian equilibrium if: 
%\begin{enumerate}
%    \item For each type $\theta_i$ and observation $\gamma_i$, the secondary market action specified by $\sigma_i$ maximizes the expected utility of agent $i$ in the secondary market given type distribution $B(\gamma_i)$ and strategies $\sigma_{-i}$ of the other agents. \item {$B$ is the \emph{rational belief} with respect to $\sigma$. That is, for each observation $\gamma_i$, $B(\gamma_i)$ is} the posterior distribution of $\Theta$ given $\sigma$ and $\gamma_i$, for each $\gamma_i$ that occurs with positive probability.
%    \item For each $\theta_i$, the action specified by $\sigma_i$ in the primary auction maximizes expected utility given agent strategies $\sigma$, belief functions $B$, and the prior type distribution $F$.
%\end{enumerate}
%At a \emph{Sequential equilibrium}, % is a pair $(\pi^*, \alpha^*)$ such that $\pi^* \in \pi^{\alpha^*}$ and $\alpha^* \in \alpha^{\pi^*}$.

\subsection{Extensions}
\label{sec:extensions}

%{Most of the {positive} results 
%in our paper actually apply in a more general framework than the basic setting described above.} %\mbc{Yingkai, you have changed "All" to "Most" but I think this is not enough as now it is not clear which do and which do not. Maybe say? The post-auction information extension does not work for the non-smoothness results, right?}

So far we have focused on the simpler setting for ease of notation and to more directly connect to the application of allocating emission licenses, but we will now describe two generalizations of the model.  {Most of our {positive} results in the remainder of the paper will actually be proven for this generalized setting.}\footnote{{In particular, our welfare bounds in Sections~\ref{sec:smooth} and~\ref{sec:pricing} apply in this generalized setting, assuming that agent valuations are submodular (which is a natural generalization of the assumption that agents have non-increasing marginal values).}} %\mbc{our claims are do not hold for general CA but only for the submodular case. So we need to add that below to make the claims about extending the results true }
%\mbc{do we want to list these results explicitly (maybe a footnote with pointers to specific claims) }
%Our reduction result for smooth auctions continues to apply under all of these generalizations.  Our result for posted-price mechanisms continues to apply under the first two generalizations (combinatorial allocations and post-auction information).  A more detailed description appears in Appendix~\ref{}.  \bjl{I'm imagining we have an appendix that contains the old model section almost exactly as it previously appeared.}

\paragraph{Combinatorial Allocation and Multiple Items} In the basic model, the items to be allocated are all identical, so each buyer is concerned only with the number of items she obtains.  
%Our main results \mbe{which show that when the auction is smooth then so is the combined market} (Theorem~\ref{} and Theorem~\ref{}) apply 
More generally, we can consider a \emph{combinatorial auction} scenario
%s \mbc{claim seems too strong - we have  positive results only for  submodular valuations (not general CA)}\bjl{Main result is the reduction; clarified this}\mbc{I have added a sentence above (are these the theorems you refer to? why two (independent and correlated?)? Anything else? )}, 
where there is a set of (possibly different) goods to allocate and each buyer has a value for each possible combination of goods.  We then interpret an allocation $x_i$ to buyer $i$ as a subset of the available goods, $x_i \subseteq M$.  
An allocation profile $\mathbf{x} = (x_1, \dotsc, x_n)$ is then feasible if no item is double-allocated, meaning that $x_i \cap x_j = \emptyset$ for all $i \neq j$. 
The basic model is the special case where all of the items are identical {and agents have non-increasing marginal values for the items.} %\mbc{and items have non-increasing marginals. We thus need to add the submodularity here to make the claim about extending the results be true}.
This generalization also captures scenarios where items of different types being sold alongside each other, such as licenses that apply to different calendar years or that permit different forms of emissions.  
The natural generalization of ``non-increasing marginal values'' is then that agent valuations are submodular, meaning that $v_i(S) + v_i(T) \geq v_i(S \cap T) + v_i(S \cup T)$ for all sets of items $S, T \subseteq M$.

%Our smoothness results allow arbitrary correlation and for valuations to depend arbitrarily on the allocation profile.
%\yl{Do we have results for valuation with externalities? Currently in our results it seems that we only cover the situation with correlated valuations, not externalities.}

\paragraph{Post-Auction Information Revealed}
In the basic model, after the auction but before the secondary market, each agent observes the outcome of the auction and her own auction payment. % Our main results \mbe{which show that when the auction is smooth then so is the combined market} (Theorem~\ref{} and Theorem~\ref{}) %\mbc{this claim seems too broad - maybe add "positive" or refer to specific results?} 
%apply to a more general setting where, in 
More generally, each agent might also observe some additional information about the auction before the secondary market begins.  For example, it may be that all agents' payments are revealed, or it might be that all bids are made public.  Formally, we can think of each buyer $i$ as observing a private signal $s_i \in S_i$ after the auction that can be correlated with $\asetfirst$, $\alloc^{\mechf}(\afirst)$, and $(\pay_i^{\mechf}(\afirst))_i$.  
In fact, we could also allow these signals to be correlated with the valuations of the agents, which allows the agents to receive additional information that even the auction has no direct access to.  We can write $\Gamma$ for the (possibly randomized) mapping from $\mathbf{v}$, $\afirst$, $\alloc^{\mechf}(\afirst)$, and $(\pay_i^{\mechf}(\afirst))_i$ to the profile of signals $(s_1, \dotsc, s_n)$ that the agents receive after the auction.  Under this generalization we would include $\Gamma$ in the description of the combined mechanism, so that $\mechc = \game(\mechf, \Gamma, \mechs)$.  In the basic setting the agents receive no signals, so we can think of $\Gamma$ as being the empty mapping that always returns a null signal.

\section{Welfare Loss and Secondary Markets}
\label{sec:example}

In this section we present an example showing that even for uniform-price auctions (which has high welfare in every equilibrium when runs in isolation), the presence of a secondary market can induce a perfect Bayesian equilibrium with low expected welfare.

% \mbc{$F$ is used in the proof as the CDF of $B$, but in the statment as the distribution over all agents. inconsistent } 
\begin{theorem}
\label{thm:lower.bound}
There exists a valuation distribution $F$, a combined market $\mechc$ consisting of a uniform-price auction followed by a trade mechanism, and a PBE $\sigma$ that avoids undominated strategies such that $\wel(\mechc,\sigma,F) \leq \frac{1}{\Omega(\log(m))} \wel(F)$. 
%\mbc{no mention of undominated strategies, but Prop 1 does prove.}
\end{theorem}

The equilibrium $\sigma$ that we use to exhibit Theorem~\ref{thm:lower.bound} will have additional nice properties, such as avoiding weakly dominated strategies.\footnote{This extends a ``no-overbidding'' refinement commonly used when considering auctions in isolation.}  Before discussing this in detail, we first describe the example.

%\mbc{sequential equilibrium require arguing about beliefs. we never do so}

%\mbc{we need to define the notion of Sequential Equilibrium under Undominated Strategies.}

%\mbc{NOTE: changes to address the comments below should also be made to the example in the appendix as appropriate. }

%\mbc{should we change this example no talk about "items" and not "units" and "licenses" ?}

\begin{example}\label{example:IEWDS}
There are $m > 3$ units to be allocated and $3$ agents named $A,B$ and $C$.  %The agent values are deterministic, so it is a full-information setting.  
Agent $A$ has marginal value $2$ for the first unit, value uniformly sample from $[1,1.5]$ for the second unit, and $0$ for any subsequent units. 
% That is, $v_A(S) = \min\{|S|,2\}$. 
Agent $B$ has the following distribution over valuations.  She always has marginal value $2$ for the first unit acquired, then a value $z_B > 0$ for each subsequent unit acquired. 
Here $z_B$ is a random variable drawn from a distribution with CDF $F_B(z) = 1-\frac{1}{1+(2m-1)z}$ for $z \in [0, 1)$ and $F_B(z) = z-\frac{1}{2m}$ for $z\in [1,1+\frac{1}{2m}]$. 
Note that given buyer value with distribution $F$,
the unique revenue maximizing price is $p=1$ with {expected revenue of} $p \cdot \Pr[z_B \geq p] = 1/(2m)$ {per-unit}.
% is a truncated equal revenue distribution with the property that, for any $p \in [1/m, 1]$, the revenue by pricing at $p$ is $p \cdot \Pr[z_B \geq p] = 1/2m$.
Agent $C$ has value $0$ for any number of units; we refer to agent $C$ as a \emph{speculator}.

% The primary auction will be a uniform-price auction with uniform bidding \mbc{need to change to standard bidding.}.  As a reminder, in this auction each agent $i$ specifies a bid $b_i \geq 0$ and a quantity $q_i \geq 0$.  This is interpreted as being willing to pay up to $b_i$ per unit up to a maximum of $q_i$ units.  The auction is resolved by allocating units in order of bids (breaking ties in any arbitrary but consistent manner), up to the requested quantity, until all units have been allocated (or until all bids are resolved). If all agents receive their requested quantity then the price $p$ per unit is $0$, otherwise it is set to the highest unallocated bid.  Each agent then pays $p$ per unit received.

The primary auction is a uniform-price auction with standard bidding.
%The auction is resolved by allocating units in order of bids (breaking ties in any arbitrary but consistent manner) until all units have been allocated.  The price $p$ per unit is set equal to the highest unallocated bid.  Each agent then pays $p$ per unit received.
%
In the secondary market, the speculator $C$ can put some or all of the items that she has acquired in the auction up for sale, at a take-it-or-leave-it price of her choice. Agent $A$ has the first opportunity to purchase any (or all) of the items made available by $C$.  
Then agent $B$ has the option to purchase any items that are still available.
%\footnote{A more general secondary market can allow all agents to offer acquired licenses for sale, not just agent $C$. The result presented extends to such more general version as well, but for simplicity we focus on a secondary market where only agent $C$ can sell.}
\end{example}

We now describe a particular choice of bidding strategies in the primary auction and agent behavior in the secondary market.  We will then prove that these form a perfect Bayesian equilibrium in the combined market.
%(and actually, a sequential equilibrium that avoids weakly dominated strategies).
%in undominated %that survives iterated elimination of weakly dominated 
%strategies).
In the auction, agents $A$ and $B$ each bid $2$ for exactly a single unit and $0$ for the rest of the units. The speculator $C$ bids $1$ 
%\mbc{Is this enough? shouldn't the bid be $1+\frac{1}{2m}$? Seems like $B$ rather win at price 1 if she has higher value}
for $m-2$ units and $0$ for the rest of the units.
Then, in the secondary market, the speculator offers all the units she has acquired for the price of $1$. Agent $A$ buys one unit, and agent $B$ buys all $m-3$ remaining units if $z_B\geq1$, and nothing otherwise.
Note that under this behaviour the price in the auction is 0. 
% Agents $A$ and $B$ have utility of $2$ in every realization of values. 
Agent $C$ has expected utility of $1+(m-3)/(2m)$, as she makes $1$ from selling to $A$, and $(m-3)/(2m)$ in expectation from selling to $B$ {(selling the remaining $m-3$ items to $B$  when $z_B\geq 1$).}

Consider the social welfare obtained in the combined market.  The total expected welfare obtained under this behavior is at most $6$ (as $2+1.5+2+\frac{1}{2m}\cdot (m-3)\cdot \left(1+\frac{1}{2m}\right)\leq 6$), whereas the optimal expected welfare is at least $4 + (m-2)\E[z]$. 
% Since $z$ is a precisely $1/m$ times a standard equal revenue distribution capped at $m$, we have 
It is easy to compute that $\E[z] \geq \int_0^{1} \frac{1}{1+(2m-1)z}\dd z= \Theta((\log m) / m)$, and hence the optimal expected welfare is $\Theta(\log m)$.
Thus, if this behavior occurs at equilibrium, this implies that the price of anarchy for this combined market is $\Theta(\log m)$, {growing unboundedly large with~$m$}.

%We now prove that this behavior does indeed correspond to a natural equilibrium.

%\bjl{I removed the comment about undominated strategies from the proposition statement.  Instead, we describe this following the proof below.  An alternative would be to add something like ``moreover, given the behavior in the secondary market, no agent is using a weakly dominated strategy in the auction.''}\mbc{I think it will be good to add this "moreover" as we claim this property in the first paragraph, and also prove it in the appendix. Added. }

{
To complete the proof of Theorem~\ref{thm:lower.bound} we must show that this behavior forms a perfect Bayesian equilibrium.  To see this, first note that the secondary market is dominance solvable: agents $A$ and $B$ should always accept utility-improving 
%\mbc{this is confusing as at the end of the day the speculator does not care about maximizing the welfare. So maybe change to "welfare-improving" or "utility-improving"?} 
trades offered by agent $C$, and hence it is a dominant strategy for agent $C$ to offer a revenue-maximizing price.  So one can equivalently think of the combined market as a one-shot auction game where payoffs of the primary auction take into account the outcomes of the secondary market, {which are unambiguous (up to zero-measure ties that do not impact utilities)}. %\mbc{this is up to the zero-measure case that $z_B=1$. Do we need to remark about this (footnote)?}.  
We can then show that the bidding strategies described above form an equilibrium of this implied one-shot game.  Moreover, those bidding strategies are not weakly dominated, again thinking of them as strategies in an implied single-shot game.\footnote{Strategy $\sigma_i$ is said to be weakly dominated by strategy $\sigma'_i$ if $\sigma'_i$ results in weakly better utility for agent $i$ for any actions that could be taken by the other agents, and strictly better utility in at least one instance.}  The formal proof of Proposition~\ref{prop:UDS} appears in Appendix~\ref{sec:sequential}.
%that there is no ``unambiguously better'' alternative bidding strategy for any agent given the way that agents will behave in the secondary market given the auction outcome.  The formal definition of sequential equilibrium and weakly dominated strategies, as well as the proof of \cref{prop:UDS}, is provided in Appendix \ref{sec:sequential}.
}

\begin{restatable}{proposition}{propsequential}\label{prop:UDS}
The above behaviour in the combined market forms a perfect Bayesian equilibrium.
{Moreover, no agent is using a weakly dominated strategy in the primary auction with respect to the payoffs implied by the secondary market.}
\end{restatable}

\paragraph{Discussion} 
%\bjl{Todo -- make this more explicit, since it's less obvious with the changed example.  E.g., understanding complementarity in $C$'s implied valuation.}  
Let's interpret this example.  {One thing to notice is that, at equilibrium, the speculator is placing a very high bid for a very large number of items; much higher than the revenue she obtains in the secondary market.  
Of course, the speculator can afford these items because of the low price, but is this a ``reasonable" strategy? At first glance it seems that this sort of behavior is a form of bullying that should be excluded by removing weakly dominated strategies (similar to overbidding in a second-price auction).
%Of course, the speculator can afford these items because of the low price, but what is the point of placing such a high bid?  
%\mbc{I feel this question is not the right one. 
%Maybe "Of course, the speculator can afford these items because of the low price, but is this a ``legitimate" strategy? At first glance it seems that this sort of behavior is a form of ``bullying'' that should be excluded by removing weakly dominated strategies (similar to overbidding in a second-price auction)."
%}
%One might expect that this sort of behavior is a form of ``bullying'' and it may be surprising at first glance that it is not excluded by {removing weakly dominated strategies (similar to overbidding in a second-price auction).}
%OLD: iterated elimination of weakly dominated strategies.  
But we note that, at equilibrium, the speculator behaves as a monopolist in the secondary market.  She generates a modest amount of revenue from each license sold to agent $B$ (namely, $1/(2m)$ each), plus a large amount of revenue (revenue of $1$) by selling a license to agent $A$.  However, this sale to agent $A$ can only occur if agent $A$ obtains fewer than $2$ items at auction.
%latter amount of revenue is only obtained as long as bidder $A$  generating revenue as long as she keeps agent $A$ \mbc{should be $B$, right?} from obtaining too many units in the auction.  
This creates an extra incentive for the speculator $C$ to obtain many items, to prevent agent $A$ from obtaining a second one.  This can cause the items to appear complementary to speculator $C$, depending on the bidding behavior of agent $A$: obtaining one fewer items could dramatically reduce $C$'s utility if that one license is won by agent $A$ instead.
%: she has high value for obtaining many units, but significantly less value if two of those units are lost to agent $A$. %\mbc{why? I do not follow. revenue from $B$ is linear in the number of units $C$ can sell to $B$. please explain.  } .  
This rationalizes the overbidding that occurs at equilibrium in the primary auction, where the speculator makes an effective bid much higher than her obtained revenue in the secondary market.  It is for this reason that the overbidding behavior of the speculator is not {weakly dominated}.%\mbc{this seems to be due to the fact that the price is very low}
}

Another implication of speculator $C$'s monopolistic behavior is that she posts a high price that distorts the allocation to agent $B$.  Although agent $B$ has high expected welfare for the goods that agent $C$ holds, the speculator $C$ maximizes revenue by setting the probability of trade very low and significantly reducing welfare.

Finally, we note that bidders $A$ and $B$ are systematically under-bidding in this equilibrium.  This is driven by demand reduction effects, where the bidders are strictly incentivized to underbid in order to keep prices low.  Importantly, this behavior is not driven by indifference, and {is undominated}.
% OLD: survives iterated elimination of weakly dominated strategies.

% \subsection{OLD: A Low-Welfare Equilibrium that Survives IEWDS}%\label{sec:low-eq-IEWDS}

%\bledit{One thing to note about this example is that a single speculator is obtaining nearly all of the items in the primary auction.  However, we note that this extreme accumulation of market power is not necessary for the result.  
%In Appendix~\ref{sec:lower.bound.gamma} we give a modified example of a market and PBE \mbc{in weakly undominated strategies?} in which no single agent obtains more than a $\gamma$ fraction of the items in the auction, and the welfare gap is $\Omega(\log(\gamma m))$.}

In this example a single speculator is obtaining nearly all of the items in the auction, and the bound is logarithmic in the number of items he acquires.
In Appendix~\ref{sec:lower.bound.gamma} we give a modified example of a market and PBE {in weakly undominated strategies} in which no single agent obtains more than a $\gamma$ fraction of the items in the auction, and the welfare gap is $\Omega(\log(\gamma m))$. Thus,
the lower bound does not hinge on a single speculator becoming a monopolist in the aftermarket, but rather degrades gracefully with the market power that is attained by any single speculator. In particular, the inefficiency persists even though multiple speculators obtain a large quantity of items, %and the inefficiency does not diminish although 
and there is no one speculator monopolizing the aftermarket.

\section{Price of Anarchy via Smooth Framework}
\label{sec:smooth}

%\bjl{Reordered content in this section for flow.}

%When licenses are sold by auction, rational agents that are aware of 
% the possibility of information acquisition and 
%the opportunity to trade in the secondary market will take this opportunity into account when placing bids in the auction.  As we have now seen, 
In the previous section we saw that the expected welfare of a uniform-price auction may decrease catastrophically when there is a secondary market, even in ``natural'' equilibria that are sequentially rational and avoid weakly dominated strategies. {Can the welfare loss due to aftermarkets be bounded for other auction formats?}
% MB: it is not about us being smart enough to find a bound... 
% OLD: Can we bound the impact of this welfare loss for certain auction formats?
Unfortunately, explicitly characterizing the associated welfare loss is a laborious task: it requires one to construct and analyze the {equilibria in the combined market, }
%MB: to bound POA we need to consider not one, but all equilibria 
%OLD: equilibrium 
which can depend on the agents' distributions in subtle ways.

In this section, we circumvent the challenges of explicitly characterizing {all} %the 
equilibrium strategies
by showing that while adding a secondary market might harm welfare, 
%the welfare loss for smooth auctions (such as first-price auction) can be upper bounded.
the worst-case welfare guarantees of several classical mechanisms (including discriminatory auctions) will not decrease in the combined market, as long as the auction mechanism satisfies certain smoothness properties.  In other words, while the equilibrium welfare may decrease in particular market instances, worst-case guarantees %due to smoothness
 are retained for smooth mechanisms.  The following definition captures the notion of smoothness we require.

\begin{definition}[\citealp{syrgkanis2013composable}]\label{def:smooth}
Auction $\mech$ with action space $\actions$ is \emph{$(\lambda,\mu)$-smooth} for $\lambda>0$ and $ \mu\geq 1$,
%For any $\lambda > 0,\mu\geq 1$, auction $\mech$ with action space $\actions$ is $(\lambda,\mu)$-smooth
if for any valuation profile $\mathbf{v}$,
%type profile $\type$,
there exists action distributions $\{\dista_i(\mathbf{v})\}_{i\in [n]}$ such that
for any action profile $\action\in \actions$,
\begin{align*}
\sum_{i\in[n]} \expect[\action'_i\sim \dista_i(\mathbf{v})]{\util_i(\mech(\action'_i, \action_{-i}))} 
\geq \lambda \cdot\wel(\mathbf{v}) - \mu \cdot\rev(\action;\mech)
\end{align*}
\end{definition}

It is known that a smooth auction in isolation achieves approximately optimal welfare at any equilibrium.

\begin{proposition}[\citealp{roughgarden2012price, syrgkanis2013composable}]\label{thm:smooth imply poa}
Let $\dists^{\Pi}$ be the family of all possible product type distributions.
If a mechanism $\mech$ is $(\lambda,\mu)$-smooth for $\lambda>0$ and $ \mu\geq 1$,
then % for any independent prior distribution of the types $F\in \dists^{\Pi}$,
the price of anarchy of $\mech$ within the family of distributions $\dists^{\Pi}$ is at most $\frac{\mu}{\lambda}$,
i.e., $\poa(\mech, \dists^{\Pi}) \leq \frac{\mu}{\lambda}$.
\end{proposition}

{
We now show the main result of this section: if a smooth auction is followed by a
secondary market that satisfies voluntary participation and weak budget balance, then the combined market is smooth as well,
and hence
% In addition, if the combined market is smooth, 
% when it operates in the presence of a pre-auction information-acquisition opportunities,
the price of anarchy is bounded for product type distributions.
}

\begin{theorem}\label{thm:poa combined market}
Let $\dists^{\Pi}$ be the family of all possible product type distributions.
For any signaling protocol $\Gamma$ and any trade mechanism $\mechs$ in the secondary market,
% and any information revelation scheme $\Gamma$,
if an auction mechanism $\mech$ is $(\lambda,\mu)$-smooth for $\lambda\in (0,1]$ and $ \mu\geq 1$,
the combined mechanism $\mechc = \game(\mech, \Gamma, \mechs)$
%$\combinedGame$
is $(\lambda,\mu)$-smooth.
Thus, the price of anarchy of $\mechc$ within the family of distributions $\dists^{\Pi}$ 
for the combined market
is at most $\frac{\mu}{\lambda}$,
i.e., $\poa(\mech, \dists^{\Pi}) \leq \frac{\mu}{\lambda}$.
\end{theorem}

The proof of Theorem~\ref{thm:poa combined market}
is given in \cref{apx:info acquire}.
Here we comment on the implications of the theorem. 
First, note that the welfare bound does not depend on the details of the information revelation structure $\Gamma$ or the trade mechanism $\mechs$ adopted in the secondary market, so the bounds hold for any choice of each.
%or the utility function of the agents. \bjl{The comment about utility functions is a bit confusing. I'm tempted to remove that part or reword.  E.g., the result imposes no additional restrictions on valuation functions.}
Moreover, our reduction framework does not require refinements on the equilibrium such as sequential equilibrium in the combined market to show that the price of anarchy is small --- the result holds for \emph{any} Bayes-Nash equilibrium.

We also note that \Cref{thm:poa combined market} extends directly to settings with multiple secondary markets executed sequentially, with any information released between each market.
This is because combining a smooth auction with a trade mechanism results with a new smooth mechanism (with the same parameters), which we can now view as a smooth auction to be combined with the next trade mechanism.

% \cref{thm:poa combined market} holds by directly combining \cref{thm:poa for invest,thm:smooth}.  Lemma~\ref{thm:poa for invest} establishes a price of anarchy bound for a smooth mechanism in which agents can acquire costly information before the mechanism begins.
% The proof of \cref{thm:poa for invest} is analogous to the smooth arguments in \citet{roughgarden2012price}
% and the details are deferred to \cref{apx:proof smooth}.
% %\mbc{is it possible to have another lemma saying that smooth parameters are maintained  under information acquisition? }

%Thus, for product distributions, 
%Combining \Cref{thm:smooth imply poa} with 
\Cref{thm:poa combined market} establishes a robust welfare guarantee for a combined market as long as the initial auction is $(\lambda,\mu)$-smooth.  For multi-unit allocation problems (such as in our basic setting), it is known that the discriminatory price auction is $(1-1/e, 1)$-smooth~\citep{de2013inefficiency}.  We therefore obtain the following corollary.

\begin{corollary}
\label{cor:discriminatory.smooth}
{Consider any multi-unit auction setting with non-increasing marginal values where the agents' valuations are distributed independently.}
%\mbc{made independence explicit. Should we make it more explicit that this is independence across agents and not across items? maybe a footnote?  }
%For a multi-unit auction setting with non-increasing marginal values, 
Let $\mechc$ be a combined mechanism that runs the discriminatory price auction followed by an arbitrary signaling protocol $\Gamma$ and trade mechanism $\mechs$.  Then at any BNE of $\mechc$ the expected welfare is at least $(1-1/e)$ times the expected optimal welfare.
\end{corollary}

%\mbc{we should consider adding the table (maybe to appendix A) with the different smooth results from prior work and the additional implications of theorem 1. In particular, present the result for submodular combinatorial auction (or only that result? ). 
%Maybe say here
%"In Appendix A we present additional corollaries of Theorem 1 to other auctions that are known to be smooth by prior work. In particular, a result for combined markets for the setting of submodular combinatorial auction (generalizing multi-unit auctions with non-increasing marginals). }

As we know from Section~\ref{sec:example}, a similar welfare bound does not hold for uniform-price auctions.  This is because, unlike the discriminatory price auction, the uniform-price auction is not $(\lambda,\mu)$-smooth for any {positive} constants $\lambda$ and $\mu$.\footnote{The uniform-price auction does satisfy a relaxed version of smoothness: it is weakly $(1-1/e, 1)$-smooth, which implies a constant welfare bound in isolation as long as agents avoid {weakly dominated} ``overbidding'' strategies~\citep{syrgkanis2013composable,de2013inefficiency}. { {In contrast,} our example in Section~\ref{sec:example} shows that eliminating dominated strategies is not sufficient to 
provide good welfare guarantees when this auction is part of a combined market.}
%\mbc{here the reader might wonder if similar "no-overbidding" result can extend to the case of combined market. We might like to say way this approach fails  (point to Section 3?)}
} %\bjl{Add an example showing why?}

% Next, Lemma~\ref{thm:smooth} establishes that smoothness is preserved when combining a primary mechanism with a trade mechanism, regardless of what information about bids is revealed after the primary mechanism concludes. 

% \begin{lemma}\label{thm:smooth}
% % Fix any secondary market mechanism $\mechs$ satisfying voluntary participation and weak budget balance,  and any signaling protocol $\Gamma$.
% %For any $(\lambda,\mu)$-smooth mechanism $\mechf$,
% If a mechanism $\mechf$ is $(\lambda,\mu)$-smooth for some for $\lambda>0$ and $ \mu\geq 1$,
% then
% for any signaling protocol $\Gamma$ and any trade mechanism $\mechs$,
% the combined mechanism $\combinedGame$
% is $(\lambda,\mu)$-smooth.
% % Thus,  the price of anarchy of $\mechc$ within the family of all product distributions $\dists^{\Pi}$  is at most $\frac{\mu}{\lambda}$,
% %  i.e., $\poa(\mechc, \dists^{\Pi}) \leq \frac{\mu}{\lambda}$.
% \end{lemma}

We can also apply \Cref{thm:poa combined market} to other smooth auctions for more general allocation problems, such as submodular combinatorial auctions. This can capture, for example, a scenario where different license types are being auctioned off simultaneously.  See Appendix~\ref{app:smooth.table} for further details on the implied welfare bounds.

In the remainder of this section we will extend \Cref{thm:poa combined market} in two ways.  In Section~\ref{sec:correlated} we show that if the auction $\mech$ satisfies a stronger notion of smoothness known as semi-smoothness, 
the price of anarchy of the combined market is bounded even for correlated type distributions.  It turns out that the discriminatory auction is $(1-1/e,1)$-semi-smooth, so the welfare bound from Corollary~\ref{cor:discriminatory.smooth} applies even if agent valuations are correlated.  Second, in Section~\ref{sec:info acquire} we show that our welfare bound continues to hold even if agents are allowed to purchase signals correlated with the value realizations of agents in advance of the auction.  For example, this captures settings in which a speculator could invest in market research before participating in the auction.

\subsection{Extension: Correlated Valuations}
\label{sec:correlated}

As \Cref{thm:smooth imply poa} is proven only for independent value distributions, \Cref{thm:poa combined market} likewise applies only to product distributions.
As it turns out, we can extend \Cref{thm:poa combined market} to derive similar results for correlated distributions based on semi-smoothness \citep{lucier2011gsp}.

%One limitation for \Cref{thm:smooth} and above applications for $(\lambda,\mu)$-smoothness is that it only implies price of anarchy for independent valuation distributions. However, similar analysis framework for correlated distributions also holds with the presence of secondary markets.

\begin{definition}[\citealp{lucier2011gsp,roughgarden2017price}]\label{def:smooth private deviation}
Auction $\mech$ with action space $\actions$ is \emph{$(\lambda,\mu)$-semi-smooth} for $\lambda>0$ and $\mu\geq 1$,
if for any valuation profile $\mathbf{v}$,
there exists action distributions $\{\dista_i(v_i)\}_{i\in [n]}$ such that
for any action profile $\action\in \actions$,
\begin{align*}
\sum_{i\in[n]} \expect[\action'_i\sim \dista_i(v_i)]{\util_i(\mech(\action'_i, \action_{-i});v_i)} \geq \lambda \wel(\mathbf{v}) - \mu \rev(\action;\mech)
\end{align*}
\end{definition}

The main difference between the definition of semi-smooth and smooth is that for each agent~$i$, 
the deviating action distribution $\dista_i(v_i)$ in semi-smooth only depends on her private valuation $v_i$,
not the entire valuation profile $\mathbf{v}$.

\begin{proposition}[\citealp{lucier2011gsp}]\label{thm:semi-smooth imply poa}
% Let $\dists$ be the family of all possible distributions.
If a mechanism $\mech$ is $(\lambda,\mu)$-semi-smooth for $\lambda>0$ and $\mu\geq 1$,
then % for any correlated prior distribution of the types $F\in \dists$,
the price of anarchy of $\mech$ within the family of all distributions $\dists$
is at most $\frac{\mu}{\lambda}$,
i.e., $\poa(\mech, \dists) \leq \frac{\mu}{\lambda}$.
\end{proposition}

Similarly to \Cref{thm:poa combined market}, we next show that combining a $(\lambda,\mu)$-semi-smooth auction with 
% any pre-auction information acquisition, 
and any signaling protocol and any trade mechanism happening aftermarkets,
the resulting mechanism in the combined market has small price of anarchy for arbitrary distributions.

\begin{theorem}\label{thm:smooth private deviation}
Let $\dists$ be the family of all possible type distributions.
For any signaling protocol $\Gamma$ and any trade mechanism $\mechs$ in the secondary market,
% and any information revelation scheme $\Gamma$,
if a mechanism $\mech$ is $(\lambda,\mu)$-semi-smooth for $\lambda\in (0,1]$ and $ \mu\geq 1$,
the combined mechanism $\combinedGame$
is $(\lambda,\mu)$-semi-smooth.
Thus, the price of anarchy of $\mech$ within the family of distributions $\dists$ 
for the combined market
is at most $\frac{\mu}{\lambda}$,
i.e., $\poa(\mech, \dists) \leq \frac{\mu}{\lambda}$.
\end{theorem}
The proof of \Cref{thm:smooth private deviation} is essentially identical to
\Cref{thm:poa combined market} (up to replacing $\dista_i(\mathbf{v})$ by  $\dista_i(v_i)$) and hence omitted here. 
We can now use results regarding semi-smooth auction from the literature to prove that the price of anarchy of the corresponding combined markets is bounded.
% \begin{proposition}[\citealp{roughgarden2017price}]\label{prop:semi-smooth first}
% For the single-item setting, the first-price auction is $(1-\sfrac{1}{e}, 1)$-semi-smooth.
% \end{proposition}
\begin{proposition}[\citealp{de2013inefficiency}]\label{prop:semi-smooth discriminatory}
For multi-unit auctions with non-increasing marginal values, the discriminatory auction is $(1-\sfrac{1}{e}, 1)$-semi-smooth.\footnote{In \citet{de2013inefficiency}, the authors only explicitly state that the discriminatory auction is smooth.
However, their construction directly implies that the discriminatory auction is semi-smooth.} 
%\begin{itemize}
%    \item $(1-\sfrac{1}{e}, 1)$-semi-smooth for submodular valuations;
%    \item $(\frac{1}{2}(1-\sfrac{1}{e}), 1)$-semi-smooth for subadditive valuations.
%\end{itemize}
%
\end{proposition}

\subsection{Extension: Acquiring Additional Information}
\label{sec:info acquire}

We consider the extension where
agents can acquire costly information about other agents' private types before the auction starts. 
This captures the application where speculators gather information on the demands in the carbon market, 
and use the acquired information to improve their utilities through buying items in the auction and reselling them more expensively in the secondary market. %arbitrage in the auction and the secondary markets. 
In general this could have a negative impact on the equilibrium welfare of the combined market. 
In this section, we show that if the designer uses smooth auctions, 
then the welfare guarantees we obtained in 
\Cref{thm:poa combined market,thm:smooth private deviation} 
%preserve 
hold even when agents can acquire costly information. 

Information is captured by a signal from the types of the others to some signal space.
Specifically, for each agent $i$, 
let $\signals_i$ be the set of feasible signal structures for agent $i$.
Any signal structure $\signal_i\in \signals_i$ is a mapping from the opponents' valuations $\mathbf{v}_{-i}$ to a distribution over the signal space $S_i$.
%\footnote{Conceptually, we can include the model where the agent does not observe his own type in ex ante, 
%and the agent can pay additional cost to learn his own type. 
%However, mathematically this can be expressed equivalent to the investment problem illustrated in \cref{sec:invest intro}, thus is omitted in this sub-section to avoid repetition.
%}
Note that the information acquisition is potentially costly, 
i.e., 
there is a non-negative cost $\cost_i(\signal_i;v_i)$ for any $\signal_i\in \signals_i$ and any type $v_i$ of $i$. 
Let $\nosignal$ be the signal that acquires no information with zero cost. 
We assume that $\nosignal \in \signals_i$ for any agent $i$.
In our model, both the cost function $\cost_i$ and the set $\signals_i$ of any agent $i$ are
common knowledge among %known to 
all agents.
% public information for any agent $i$.
%The information acquisition occurs before the investment decisions for each agent $i$, 
%and the agents can only acquire information about others private types, 
%not their investment decisions or the induced valuations after the investments.
%\mbc{maybe point to an example showing that the result fails if such information can be acquired.} 

% Note that the ability to acquire additional information does not affect the optimal welfare we considered in \cref{sec:basic} since the optimal welfare is attained when no agent acquires costly information. 
% Thus the price of anarchy is defined analogously as the ratio between the optimal welfare %\mbe{(which stays the same)} 
% and the equilibrium welfare when information acquisition is an option. 

In the following theorem, we extend Theorem \ref{thm:poa combined market} and show that the price of anarchy for smooth auctions in the combined market is bounded when agents can acquire information on the competitors.
% generalize when there is asecondary market after the auction
\begin{restatable}{proposition}{lemmapoainfo}\label{thm:poa for invest}
Let $\dists^{\Pi}$ be the family of all possible product type distributions.
For any set of signals $\signals$ and any cost function $\cost$,
if mechanism $\mech$ 
% for auction and secondary market
is $(\lambda,\mu)$-smooth for $\lambda\in (0,1]$ and $ \mu\geq 1$,
then the price of anarchy of $\mech$ within the family of distributions $\dists^{\Pi}$ 
for the combined market with information acquisition
is at most $\frac{\mu}{\lambda}$,
i.e., $\poa(\mech, \dists^{\Pi}) \leq \frac{\mu}{\lambda}$.
\end{restatable}

The proof of \cref{thm:poa for invest} is provided in \cref{apx:info acquire}.
Similarly, one can extend \cref{thm:smooth private deviation} for all family of all possible type distributions when the mechanism $\mech$ is $(\lambda,\mu)$-semi-smooth. As the proof is similar to the proof of \cref{thm:poa for invest}, we omit it.

\section{Welfare Guarantees under Posted Pricing}
\label{sec:pricing}

%\bjl{Todo: some of the following motivation should be in the intro; revisit this text after the intro is settled to avoid duplication.}
In the previous section we showed that the welfare guarantees derived from smooth auctions, such as discriminatory price auctions, are robust to the presence of a secondary market. 
Such smooth auctions have the advantage of being agnostic to the prior distributions from which agent valuations are drawn. 
A downside is that bidding in such auctions can be quite complex: constructing an optimal bidding strategy requires sophisticated reasoning and the ability to predict market conditions. 
If the designer (i.e., government) has a sense of the market conditions, then a tempting alternative to running a smooth auction is to sell a fixed quantity of items at a pre-specified unit price, while supplies last.  
%Such a price would ideally be chosen to maximize welfare by screening out inefficient allocations \mbc{maybe replace by "by setting the price to be the market clearing price (by the First Welfare Theorem)"}. 
Such posted-price mechanisms have the advantage of being very simple to participate in, since each potential buyer can simply purchase her utility-maximizing bundle of the remaining items at the given prices.
%\bjl{Say something more specific to carbon markets?}

A recent literature on static posted pricing and Prophet Inequalities has illustrated that such pricing methods can provide strong welfare guarantees in a variety of allocation problems, even when the order of buyer arrival is adversarial.  See \cite{lucier2017economic} for a recent survey.  For example, consider the special case where there is a single indivisible good to be sold and each buyer's value is drawn independently.  It is known that if there is no secondary market, then setting a single take-it-or-leave-it price and {(while the item is still available) letting buyers in sequence each} choose whether to purchase, guarantees half of the expected maximum.  Does this guarantee still hold in the presence of a secondary market?

As it turns out, the answer depends on how the fixed price is selected.  The canonical solution to the single-item problem, based on the classic Prophet inequality, chooses price
$p^* = \sup\{ p : \Pr[ \max_i v_i > p ] \geq 1/2\}$.  That is, the median of the distribution over maximum values.  This choice of price guarantees half of the expected optimal welfare with no secondary market~\citep{samuel1984comparison}, but the following example shows that this is no longer the case in the presence of a secondary market.\footnote{The 
{guarantee of half of the expected maximum welfare}
% OLD: $2$-approximation result 
with no secondary market requires some care in the case that a value is precisely equal to $p$.  But this occurs with probability $0$ in our example, so our negative result holds regardless of how this is handled.}

\begin{example}\label{example:posted-fails}
There is a single item to be sold and two potential  buyers.  The first buyer has value $v_1$ drawn uniformly from $[0,1]$.  The second buyer has value $v_2$ equal to $0$ with probability $1-\epsilon<1$, and with the remaining probability $\epsilon$ the value $v_2$ is set equal to $\epsilon^{-1}$ times a random variable drawn from the equal-revenue distribution capped at $H$.  That is, with probability $\epsilon>0$, $v_2 = \epsilon^{-1} z$ where $z$ is drawn from a distribution with CDF $F(z) = \frac{z-1}{z}$ for $z \in [1,H]$ and $F(H) = 1$, {and with the remaining probability $v_2=0$}.

The efficient allocation gives the good to buyer $2$ whenever $v_2 > 0$, leading to an expected welfare of at least $\epsilon \times \epsilon^{-1} \times \text{\bf E}[z] = \Theta(\log H)$. 

Note that the median price is $p^*=\sup\{ p : \Pr[ \max_i v_i > p ] \geq 1/2\} = \frac{1}{2(1-\epsilon)}$.  Suppose we offer this price to each buyer in sequence, starting with buyer 1.  If there is no secondary market, then the first buyer will purchase only if $v_1 \geq p^*$, which occurs with probability less than $1/2$.  The item is therefore available for purchase for the second buyer with probability at least $1/2$, leading to an expected welfare of $\Omega(\log(H))$. { This mechanism therefore obtains a constant fraction of the optimal welfare when running in isolation}.

Now suppose that {the posted-price mechanism is followed by} % OLD: we allow 
a secondary market in which the winning buyer (if any) can post a take-it-or-leave-it price offer to the losing buyer.  In this case the first buyer would always prefer to purchase the item at price ${p^*=\frac{1}{2(1-\epsilon)}}$, and
then offer to resell it to the second buyer in the secondary market at price $p' = \epsilon^{-1} H$.  Note that this choice of $p'$ is revenue-maximizing, assuming that no extra information about buyer valuations is revealed between the auction and the secondary market, and obtains expected revenue $1 \geq v_1$.
%(Assume that no information about buyer valuations is revealed  between the auction and the secondary market, aside from the purchasing decisions).  
This is therefore a sequential equilibrium.  The expected welfare at this equilibrium is $O(1)$, since $v_1 = O(1)$ and $\Pr[v_2 \geq \epsilon^{-1}H] = \epsilon H^{-1}$.  Taking $H$ sufficiently large leads to an arbitrarily large welfare gap.
\end{example}

This example shows that the pricing strategy based the median of maximum values can lead to significant incentives for a low-value buyer to purchase with the intention to resell.
Then, due to monopolist distortions, significant welfare is subsequently lost in the secondary market.

Another approach to setting static posted prices is based on so-called ``balanced prices"~\citep{kleinberg2012matroid,dutting2020prophet}.  
In the single-item example described above, this corresponds to setting a price equal to $\frac{1}{2}\text{\bf E}[\max_i v_i]$.  This approach likewise guarantees half of the optimal welfare for the single-item prophet inequality problem~\citep{kleinberg2012matroid}.
{
We will show that, unlike Example~\ref{example:posted-fails}, this guarantee continues to hold even in the presence of a secondary market.
} % old: As we next show, the approximation factor due to balanced prices still holds in a setting with a secondary market.
This is true not only for single-item auctions, 
{but for the broader broad class of multi-unit auctions, and even to combinatorial allocation problems. }
%OLD: but for a broad class of multi-unit and combinatorial allocation problems. 

%\bjl{Todo: language to connect this to carbon markets more explicitly?  Mirror the language we come up with in earlier sections.}
To formalize this claim, we make use of a general definition of balanced prices due to~\citet{dutting2020prophet}.
We will state this definition in a general combinatorial auction setting in which the set of items $M$ are not necessarily identical.  A pricing function $p$ assigns to each set of items $x \subseteq M$ a price $p(x) \geq 0$.  For example, this function might assign a price to each individual item 
{and set $p(x)$ to be the sum of item prices,}
%OLD: then $p(x)$ is the sum of item prices, 
or more generally $p(x)$ might assign arbitrary prices to each bundle.\footnote{This definition naturally extends to fractional or randomized allocations, 
{but in order to keep notation simple, in this section we will restrict attention to deterministic allocations. 
Theorem \ref{thm:poa.pricing} can be  extended to fractional or randomized allocations with the appropriate notational adjustments.}
%OLD but in this section we will restrict attention to deterministic allocations.
}

%To state what we mean for a price function to be balanced, we write $\opt(\mathbf{v})$ for the optimal allocation given valuations $\mathbf{v}$.  We also write $\opt(\mathbf{v} | \mathbf{x})$ for the optimal \emph{residual} allocation: the welfare-optimal allocation of all items not allocated in $\mathbf{x}$. \bjl{Todo -- switch to the $\wel()$ notation from Section 2}
{First some notation.  Write $\opt(\mathbf{v}, S)$ for the welfare-optimal allocation of the items in $S \subseteq M$ {when the valuations are $\mathbf{v}$}, and denote $\opt(\mathbf{v}) = \opt(\mathbf{v}, M)$.  In a slight abuse of notation, write $M\backslash \mathbf{x}$ for $M \backslash (\cup_i x_i)$, the set of items that are unallocated in feasible solution $\mathbf{x}$.  Then, in particular, $\opt(\mathbf{v}, M\backslash \mathbf{x})$ is the welfare-optimal allocation of the items that are unallocated in $\mathbf{x}$.}

The following is a definition from~\citet{dutting2020prophet} specialized to our setting.

\begin{definition}
For a given valuation realization $\mathbf{v}$, a pricing function $p$ is $(\alpha,\beta)$-balanced if, for any pair of allocations $\mathbf{x}$, $\mathbf{x}'$ that are disjoint and jointly feasible (i.e., they allocate disjoint sets of items), we have
\begin{itemize}
    \item $\sum_{i\in \agents} p(x_i) \geq \frac{1}{\alpha}(\mathbf{v}(\opt(\mathbf{v})) - \mathbf{v}(\opt(\mathbf{v}, M\backslash \mathbf{x})))$,
    \item $\sum_{i\in \agents} p(x'_i) \leq \beta \cdot \mathbf{v}(\opt(\mathbf{v}, M\backslash \mathbf{x}))$
\end{itemize}
\end{definition}

That is, prices $p$ are balanced if the price paid for any allocation $\mathbf{x}$ is at least the loss in optimal welfare due to losing the items in $\mathbf{x}$ (up to a factor of $\alpha$). Secondly,  the total price of any allocation $\mathbf{x}'$ that remains feasible after $\mathbf{x}$ is removed, is at most $\beta$ times the optimal welfare achievable using the items not allocated in $\mathbf{x}$.  {Note that $\mathbf{x}'$ need not be the welfare-optimal allocation {of the items in $M\backslash \mathbf{x}$.}}
% OLD: That is, prices $p$ are balanced if the price paid for any allocation $x$ offset the loss in optimal welfare due to losing the items in $x$, and the total price of any $x'$ that remains feasible after $x$ is removed is at most $\beta$ times the welfare achievable after the items in $x$ are lost (which is at most the welfare of allocation $x'$).

Note that this definition of balanced prices is with respect to a particular realization {$\mathbf{v}$} of agent types.  {We will use this to construct a static pricing rule {(i.e., prices that are independent of realizations)} by taking an expectation over types.  Formally, we say that a posted-price mechanism uses $(\alpha,\beta)$-balanced prices if it (a) defines an $(\alpha,\beta)$-balanced price function $p^\mathbf{v}$ for all valuation profiles $\mathbf{v}$, then (b) sets its actual static price function $p$ according to $p(x) = \frac{\alpha}{1+\alpha\beta}E_{\mathbf{v}}[p^{\mathbf{v}}(x)]$.}  That is, prices are set by taking expectations {of the type-specific prices} over the buyer types and multiplying by the constant $\frac{\alpha}{1+\alpha\beta}$.  For example, in the single-item prophet inequality setting, a $(1,1)$-balanced price for a given realization of values $\mathbf{v}$ is {the price $p^\mathbf{v} = \max_i v_i$}.  The appropriate choice of posted price for the prophet inequality is then $\frac{1}{2}\text{\bf E}[\max_i v_i]$, the expectation of the balanced prices times $\tfrac{1}{2}$. 

It is known that in the absence of a secondary market, using balanced prices leads to a strong welfare guarantee.

\begin{proposition}[\citealp{dutting2020prophet}]
Fix the valuation distributions $F = \times_i F_i$.
If a posted-price mechanism $\mech$ uses $(\alpha,\beta)$-balanced prices for $\alpha, \beta\geq 1$,
then the expected welfare obtained by $\mech$ is at least $\frac{1}{1+\alpha\beta}$ times the expected optimal welfare.
\end{proposition}

As we now show, this result extends to any equilibrium in the combined market setting with an arbitrary trade mechanism being used as a secondary market.  This theorem applies to multi-unit and combinatorial auction allocation problems.

%\mbc{need to be explicit about the setting: do we want a theorem for the case of single item only? Maybe better to state the theorem for the multi-unit case upfront.}\bjl{This is for the fully general setting in our model section.  Clarified above, but we probably want to revisit the language once we've settled on the presentation in the model section.}

\begin{theorem}\label{thm:poa.pricing} 
Fix the valuation distributions $F = \times_i F_i$.
For any signaling protocol $\Gamma$ 
and any trade mechanism $\mechs$ in the secondary market,
if a posted-price mechanism $\mechf$ uses $(\alpha,\beta)$-balanced prices for $\alpha, \beta\geq 1$,
then the expected welfare obtained by the combined mechanism $\combinedGame$ at any Bayesian-Nash equilibrium is at least $\frac{1}{1+\alpha\beta}$ times the expected optimal welfare.
\end{theorem}
Our argument is an adaptation of a proof method due to \citet{dutting2020prophet}. 
The details of the proof is given in \cref{apx:price proof}.

\subsection{Application: Balanced Prices for Carbon Markets}\label{sec:balanced-prices}

%\bjl{Say more about other applications of Theorem~\ref{thm:poa.pricing}?  E.g., list out known balanced-price bounds for combinatorial auctions, matroids, etc.?}

We can apply Theorem~\ref{thm:poa.pricing} to any allocation problem for which balanced prices exist.  For example, it is known that one can design $(1,1)$-balanced item prices for submodular combinatorial auctions~\citep{feldman2014combinatorial}.  A multi-unit auction with weakly decreasing marginal values is a special case of a submodular combinatorial auction in which all items are identical.  We can therefore conclude from Theorem~\ref{thm:poa.pricing} the existence of item prices that guarantee half of the expected optimal welfare in our model of a carbon license market with an arbitrary aftermarket.\footnote{The approximation ratio is tight even for the single-item setting \citep{kleinberg2012matroid}.}
%\mbc{as far as i understand this half is tight (in the sense that no posted-price mechanism gives a better constant. Should we remark this anywhere? )}
%\yl{added a footnote.}

One thing to note is that in the standard construction for submodular combinatorial auctions, there is no guarantee that all items will be assigned the same price even if the items are identical.  We next argue that in fact there always exist $(1,1)$-balanced prices that are identical across items, so in fact there is only a single price that can be interpreted as a per-unit price offered to all prospective buyers.

%For the multi-unit setting, we can derive particularly natural balanced prices as follows.\mbc{It is unclear from this wording if we consider this a contribution or known results from prior work (If so, need to cite)}\bjl{I'm not aware of it appearing before as stated.  So I think it's new.} \mbc{so it is confusing that we first say "it is known" and than say "we derive" (for the same multi-unit setting)}
Our construction is as follows.  For each valuation profile $\mathbf{v}$, calculate the optimal allocation $\mathbf{x}^*(\mathbf{v})$
by greedily allocating in order of highest marginal value. 
Define {$w^{\mathbf{v}} \triangleq \frac{1}{m}\sum_i v_i(x^*_i(\mathbf{v}))\geq 0$} to be the {average} per-unit welfare of this optimal allocation. 
We use the average per-unit welfare $w^{\mathbf{v}}$ as a price for each unit of the item, so in particular all units have the same price.  That is, the price to acquire $k$ units is $p^{\mathbf{v}}(k) = w^{\mathbf{v}} \times k$. 
%\mbc{aren't we missing a factor of half? without this I cannot see how this can be robust to small errors upwards in the estimation. I see, this is only the (1,1) balanced part, we later add the factor half. }
We claim that this choice of unit prices is $(1,1)$-balanced.

\begin{claim}
\label{claim:multiunit.pricing}
For a multi-unit allocation problem, 
for any profile~$\mathbf{v}$ of valuations with non-increasing marginal values, 
the price function $p^{\mathbf{v}}(x) = w^{\mathbf{v}} \cdot |x|$ where $w^{\mathbf{v}} = \frac{1}{m}\sum_i v_i(x^*_i(\mathbf{v}))$
%by setting the price for each item \mbc{better say "per unit"?} to be $w^{\mathbf{v}} = \frac{1}{m}\sum_i v_i(x^*_i(\mathbf{v}))$ 
is $(1,1)$-balanced.
% MB: the prices are a function of $v$, so I have revered the order of the sentence. 
%OLD: the prices described above are $(1,1)$-balanced for any profile~$\mathbf{v}$ of valuations with non-increasing marginal values.
\end{claim}
\begin{proof}
Choose any pair of % \mbc{are these deterministic? can they be fractional allocations?}\bjl{As long as we define the valuation of a fractional allocation appropriately, yes.  We'd use the fact that the optimal fractional allocation is equivalent to a deterministic allocation with fractional coordinates.  But I'm not sure how much we want to get into this for applied modeling.} 
allocations $(\mathbf{x},\mathbf{x}')$ as in the definition of balancedness.  Suppose $\mathbf{x}$ allocates $k \leq m$ items in total, meaning that $\sum_i x_i = k$ %\mbc{if we allow for fractional allocations, then this sum might not be an integer (yet the proof seems to go through - looking at per-unit contributions)}\bjl{Agreed, it would go through.}
%allocation as a distribution over deterministic.  
Since allocation $\mathbf{x}'$ is feasible after the items allocated under $\mathbf{x}$ are removed, we have that $\mathbf{x}'$ allocates at most $m-k$ items.

For the first condition of balancedness, note that $\sum_i p^{\mathbf{v}}(x_i)$ is $k$ times the per-unit welfare of the optimal allocation.  On the other hand, $\mathbf{v}(\opt(\mathbf{v})) - \mathbf{v}(\opt(\mathbf{v}, M\backslash \mathbf{x}))$ is the loss of welfare due to removing any $k$ items from the auction, which is precisely the sum of the $k$ smallest marginal values in the optimal allocation.  The average of these $k$ smallest marginal values is at most the overall average, and hence $\sum_i p^{\mathbf{v}}(x_i) \geq \mathbf{v}(\opt(\mathbf{v})) - \mathbf{v}(\opt(\mathbf{v}, M\backslash \mathbf{x}))$.

For the second condition, note that $\sum_i p^{\mathbf{v}}(x'_i)$ is at most $m-k$ times the per-unit welfare of the optimal allocation.  On the other hand, $\mathbf{v}(\opt(\mathbf{v}, M\backslash \mathbf{x}))$ is the value of the optimal allocation of $m-k$ elements, which is the sum of the largest $m-k$ marginals in the overall optimal allocation.  The average of these $m-k$ largest marginals is at least the overall average, and hence $\sum_i p^{\mathbf{v}}(x'_i) \leq \mathbf{v}(\opt(\mathbf{v}, M\backslash \mathbf{x}))$.

We conclude that both conditions of balancedness are satisfied with $(\alpha,\beta) = (1,1)$.
\end{proof}

Given this choice of $(1,1)$-balanced prices, Theorem~\ref{thm:poa.pricing} then implies that setting a static price per unit that equals to half of the expected per-unit welfare, $\frac{1}{2m}\text{\bf E}[\sum_i v_i(x_i(\mathbf{v}))]$, guarantees half of the optimal expected welfare in any equilibrium, under any secondary market implementable as a trade mechanism and for any set of information revealed between the posted-price mechanism and the secondary market. 

{
\begin{example}
Recall the lower bound market example in Section~\ref{sec:example}, which illustrated a low-welfare equilibrium in a uniform auction with an aftermarket.  In that example, the welfare-maximizing allocation always allocates two goods to bidder $A$ and the remaining $m-2$ goods to bidder $B$, {for a total expected welfare of $\Theta(\log m)$.} 
%OLD: for a total welfare of $3 + (m-3)z_B$.  The expected optimal welfare is therefore $3 + (m-3)\E[z_B]$ \mbc{the example is changing, we will need to check the details here once done with Example 3.}.  
In this example,  {\cref{thm:poa.pricing} together with the $(1,1)$-balanced prices we present in \cref{claim:multiunit.pricing} }
%OLD:our construction from Claim~\ref{claim:multiunit.pricing} 
would set a per-unit price of {$E[v(\opt(v))]/(2m)= \Theta((\log m)/m) $} 
%OLD: $(3 + (m-3)\E[z_B])/(2m) = \Theta((\log m)/m)$ 
and allow bidders to purchase as many units as desired at this price, % MB: this got me confused, as we need to use half the prices of \cref{claim:multiunit.pricing} (as per Theorem 3)  
while supplies last.  Theorem~\ref{thm:poa.pricing} implies that, at this price, any purchasing behavior at equilibrium will achieve at least half of the expected optimal welfare. Intuitively, the price is high enough to discourage the speculator agent $C$ from buying more than two items %\mbc{"excessive"? seems he will buy at most two items, right? (can only sell to A)}
(since at most two items can be sold to agent $A$, and any items sold to agent $B$ generate an expected revenue of at most $1/m$ each, which is less than the auction {per-unit} reserve)
but also low enough that agent $B$ will nearly always decide to purchase most of the available items.
\end{example}
}

%\mbc{ It might be good to go back to Example \ref{example:posted-fails} and discuss how the balanced price will be set and what will be the welfare obtained. (not sure where we should put that discussion) } \bjl{Done.}

%\bjl{do we prove the general result for balanced prices, or do we focus on the more particular setting of multi-unit auctions to tie ourselves more closely with the application domain?} \mbc{How general can this be? all matroid domains? even further?} \bjl{Yes, even further.  Any combinatorial allocation problem subject to a downward-closed feasibility constraint. The proof of Theorem~\ref{thm:poa.pricing} above is stated for general combinatorial auctions.}\mbc{we should remark about this.}\bjl{Sure.  It should mirror the way we handle this in the Smoothness section.  I can come back to this after taking a pass over that section.}

\section{Best of Both Worlds: Uniform-Price Auction with Reserves}
\label{sec:reserves}
{We realize that moving from an auction to a posted price mechanism is a radical change in many markets, including the carbon allowances market. We thus consider the problem of finding a minimal change to a uniform-price auction} which will still take care of the problem of {significant efficiency loss} %inefficiency 
that can happen as a result of combing the auction with a secondary market (as illustrated in the example presented in Section \ref{sec:example}). {As we now show, 
%our welfare guarantee can alternatively be obtained by slightly changing the uniform-price auction, 
a sufficient change is to add to the uniform-price auction a per-unit reserve price based on balanced prices.}
%(and using any tie-breaking rule). 
Any marginal bids strictly less than this reserve price are ignored, but bids exactly equal to the reserve are allowed. Formally, in the following theorem we show that the welfare bound from Theorem~\ref{thm:poa.pricing} continues to hold when we add {an appropriately chosen per-unit reserve price} 
% reserve prices 
to the uniform-price auction.
{The proof of \cref{thm:reserve} is given in Appendix \ref{apx:price proof}.
% MB: 
%\mbc{We should consider skipping the proof of Theorem 4 and only prove the more general result that appears in appendix D (as the proof are almost identical, right?) if we do so, we need to point to Appendix \ref{sec:uniform.errors}.}
}
%OLD: The proof of \cref{thm:poa.pricing} is given in Appendix \ref{apx:price proof}.
% at any equilibrium of the resulting combined market.

%\bjl{Removed note about tie-breaking.  We already say that we use an arbitrary tie-breaking rule in the preliminaries when describing multi-unit auctions.}\mbc{might still be good to emphasize the modified auction works with arbitrary tie breaking}

\iffalse
\mbc{OLD (see above for the new paragraph), to be removed:} A policy recommendation: continue running a uniform-price auction, but impose a reserve price based on balanced prices.  The reserve price will be the price identified in Claim~\ref{claim:multiunit.pricing}.  Any marginal bids strictly less than this reserve price are ignored (but bids exactly equal to the reserve are allowed). \mbc{how are ties handled?Assuming any tie breaking is ok (which I think is the case) then we can actually do something "fair" - say, pick an allocation that is proportional to reported demand (with appropriate randomization when needed) - though it will be good to check if it is always also proportional to real demand (say when every agent  has same value per unit till some quantity).  }\bjl{We'll say it works for any arbitrary tie-breaking.} 
\fi

\begin{theorem}\label{thm:reserve}
For a multi-unit allocation problem with non-increasing marginal values {sampled independently from $\dist = \times_i\dist_i$}, 
consider mechanism $\mechf$ that is a uniform-price auction subject to a per-unit reserve price of {$p = \frac{1}{2m}E_{\mathbf{v}\sim \dist }[\mathbf{v}(\opt(\mathbf{v}))]$}.
% $p$ which is $(\alpha,\beta)$-balanced for $\alpha, \beta\geq 1$.
%OLD: $p = \frac{1}{2m}\mathbf{v}(\opt(\mathbf{v}))$ \mbc{shouldn't this be $p = \frac{1}{2m}E_{\mathbf{v}\sim \dist }[\mathbf{v}(\opt(\mathbf{v}))]$?}.
Then for any signaling protocol $\Gamma$ and any trade mechanism $\mechs$ in the secondary market,
%Then for any trade mechanism, 
and at any Bayesian Nash equilibrium %\mbc{should this be stated as "BNE"? If so, we should make this change consistently in other places as well.}
of the combined market $\combinedGame$, the expected welfare is at least 
% $\frac{1}{1+\alpha\beta}$ 
half of the expected optimal welfare.
\end{theorem}

% Setting $c=\frac{1}{2}$, which leads to a $(1,1)$-balance price, 
% and \cref{thm:reserve} indicates that uniform price auction with such reserve price guarantees at least half of the optimal welfare in the combined market.
% Moreover, this reserve price can be estimated efficiently from the empirical data with $\epsilon$ error, 
% i.e., $p\in [\frac{c-\epsilon}{m}\mathbf{v}(\opt(\mathbf{v})), \frac{c+\epsilon}{m}\mathbf{v}(\opt(\mathbf{v}))]$ with high probability given sufficient samples.
% Again by \cref{thm:reserve}, setting any price within that range would lead to at most $\epsilon$ multiplicative loss in optimal welfare. 

\paragraph{Estimation Errors and Calculating Reserve Prices}
{
The reserve price in Theorem~\ref{thm:reserve} depends on the expected optimal welfare attainable by allocating items.  As it turns out, the welfare guarantee in Theorem~\ref{thm:reserve} is robust to mistakes in the prices.  This robustness was noted by~\citet{feldman2014combinatorial} for posted-price mechanisms in submodular combinatorial auctions, and the argument extends to our setting without change. 
Specifically, for any $\epsilon > 0$, if instead of setting reserve $p = \frac{1}{2m}E[\mathbf{v}(\opt(\mathbf{v}))]$ we set a reserve price $p'$ such that $|p' - p| < \epsilon$, then the expected welfare of the resulting combined market is at least half of the expected optimal welfare less $m\epsilon$, where recall that $m$ is the total number of items for sale.  
To see why, recall the proof of Theorem~\ref{thm:reserve} and note that lowering a license's price by $\epsilon$ reduces the revenue obtained by at most $\epsilon$, and increasing a license's price by $\epsilon$ reduces buyer surplus from purchasing that item by at most $\epsilon$.  Collecting up these additive losses leads to a total loss of at most $m\epsilon$. 
A formal statement and proof are deferred to Appendix~\ref{sec:uniform.errors}.

One implication of this robustness result is that if one sets reserve prices based on a noisy estimate of $E[\mathbf{v}(\opt(\mathbf{v}))]$, such as obtained by samples or historical data, then the welfare guarantee one obtains will degrade in proportion to the estimation error.
}

%The construction described above sets prices 

\begin{corollary}
For a multi-unit allocation problem with non-increasing marginal values {sampled independently from $\dist = \times_i\dist_i$}, 
suppose that the seller has access to an estimate $\estimate$ to the optimal welfare such that 
with probability at least $1-\delta$, 
$\estimate\in [(1-\epsilon)\mathbf{v}(\opt(\mathbf{v})), (1+\epsilon)\mathbf{v}(\opt(\mathbf{v}))]$
for constants $\delta,\epsilon\in [0,1]$.
Consider the uniform-price auction subject to a per-unit reserve price of
$p = \frac{\estimate}{2m}$.
Then with probability at least $1-\delta$, for any trade mechanism, and at any Bayesian Nash equilibrium of the combined market, the expected welfare is at least 
% $\frac{1}{1+\alpha\beta}$ 
$\frac{1-\epsilon}{2}$ 
fraction of the expected optimal welfare.
\end{corollary}

%\bjl{Add a comment about connection to Michal-Nick paper}

\newpage
\bibliographystyle{apalike}
\bibliography{ref}

\newpage
\appendix
\section{Missing Proof for Smooth Auctions}
\label{apx:info acquire}

\subsection{Proof of Theorem~\ref{thm:poa combined market}}

First we finish the task of proving Theorem~\ref{thm:poa combined market}.

\begin{proof}[Proof of Theorem~\ref{thm:poa combined market}]
Let $\mechf = \mech$ for clarity of notation. Let $\asetfirst, \asetsecond$ be the action spaces of mechanisms $\mechf, \mechs$ respectively,
and let $\asetcombined$ be the action space of the combined market.
Note that $\acombined\in\asetcombined$ is an extensive form action,
which is choosing action $\afirst\in\asetfirst$ in the first market,
and then choosing action $\asecond\in\asetsecond$ based on the allocation, payment, and the signal realized in the first market.
%and $\asetcombined = \asetfirst \times \asetsecond$ be the action space of the combined market.
For each agent $i$, by voluntary participation,
there exists action $\hat{a}\second_i \in \asetsecond_i$
such that her payoff is at least her value of the initial allocation in the secondary market.
Since mechanism $\mechf$ is $(\lambda,\mu)$-smooth,
for any valuation profile $\mathbf{v}$,
there exists action distributions $\{\dista\first_i(\mathbf{v})\}_{i\in [n]}$ such that
for any action profile $\afirst\in \asetfirst$,
\begin{align*}
\sum_{i\in[n]} \expect[\hat{\action}\first_i\sim \dista^1_i(\mathbf{v})]{\util_i(\mechf(\hat{\action}\first_i, \afirst_{-i}))}
\geq \lambda \wel(\mathbf{v}) - \mu \rev(\afirst;\mechf).
\end{align*}
For any valuation profile $\mathbf{v}$, let $\{\dista\combined_i(\mathbf{v})\}_{i\in [n]}$ be the distributions over actions $\hat{\action}\combined_i$ for each agent $i$
where $\hat{\action}\combined_i$ chooses action $\hat{\action}\first_i$ according to distribution $\dista\first_i(\mathbf{v})$,
and then always chooses action $\hat{\action}\second_i$ regardless of the signal realization $s_i$ from $\Gamma$, or the allocation and payment in the first market.
%$(\afirst_i, \hat{\action}\second_i)$ for each agent $i$
%where $\afirst_i \sim \dista\first_i(\type)$.
For any $\acombined \in \actions\combined$,
we have
\begin{align*}
&\sum_{i\in[n]} \expect[\hat{\action}\combined_i\sim \dista\combined_i(\mathbf{v})]{\util_i(\mechc(\hat{\action}\combined_i, \acombined_{-i}))}
\geq \sum_{i\in[n]} \expect[\hat{\action}\first_i\sim \dista^1_i(\mathbf{v})]{\util_i(\mechf(\hat{\action}\first_i, \afirst_{-i}))} \\
&\geq \lambda \wel(\mathbf{v}) - \mu \rev(\afirst;\mechf)
\geq \lambda \wel(\mathbf{v}) - \mu \rev(\acombined;\mechc)
\end{align*}
where the first inequality holds by the definition of $\hat{\action}^2_i$, the second as  $\mechf$ is $(\lambda,\mu)$-smooth,
and the last inequality holds since mechanism $\mechs$ is weakly budget balanced.
Thus, the combined mechanism $\combinedGame$
is $(\lambda,\mu)$-smooth.
Finally, by \Cref{thm:smooth imply poa}, we have that the price of anarchy in the combined market is at most $\frac{\mu}{\lambda}$.
\end{proof}

\subsection{Proof of \cref{thm:poa for invest}}
Next we formally prove \cref{thm:poa for invest} for the setting in which agents can acquire costly signals about the valuations of other agents prior to participating in a combined mechanism. 

% \lemmapoainfo
\begin{proof}[Proof of \cref{thm:poa for invest}]
Since mechanism $\mech$ is $(\lambda,\mu)$-smooth, 
by definition, 
for any value type $\type$,
there exists action distributions $\{\dista_i(\type)\}_{i\in [n]}$ such that
for any action profile $\action\in \actions$,
\begin{align*}
\sum_{i\in[n]} \expect[\action'_i\sim \dista_i(\type)]{\util_i(\mech(\action'_i, \action_{-i});\type_i)} \geq \lambda \wel(\type) - \mu \rev(\action;\mech).
\end{align*}

Suppose in equilibrium, for any agent $i$ 
the information acquisition strategy is $\bar{\sigma}_i(\type)$
and the bidding strategy is $\hat{\sigma}_i(\type,s_i)$.
Note that since the information acquisition decisions
are not revealed to the opponents, 
the distribution over bids of any agent $i$
is not affected by the information acquisition decisions taken % chosen 
by agent $j\neq i$. 
Let $G_i$ be the distribution over actions in the auction for agent $i$ under equilibrium strategies 
$\bar{\sigma}_i$ 
and $\hat{\sigma}_i$, when her type is distributed according to $\dist_i$.

Now consider the following deviating strategy for agent $i$. 
Agent $i$ will not acquire any information by adopting signal structure $\nosignal$.
Then in the auction, agent $i$ simulates the behavior of the other agents by first sampling
%first simulates 
$\hat{\type}_{j}$ according distribution
$\dist_{j}$ for any $j\neq i$,
and then follow the action distribution $\dista_i(\type_i,\hat{\type}_{-i})$.
The expected utility of all agents given this deviating strategy is 
\begin{align*}
&\sum_{i\in [n]}\expect[\type_i\sim\dist]{
\expect[\action'_i\sim \dista_i(\type_i,\hat{\type}_{-i});
\action_{-i}\sim G_{-i}]{
\util_i(\mech(\action'_i, \action_{-i});\type_i)
}}\\
&= \sum_{i\in [n]}\expect[\type\sim\dist]{
\expect[\action'_i\sim \dista_i(\type);
\action_{-i}\sim G_{-i}]{
\util_i(\mech(\action'_i, \action_{-i});\type_i)
}}\\
&\geq 
\expect[\type\sim\dist]{
\lambda \cdot \wel(\type) - 
\mu \cdot \expect[\action \sim G]{
\rev(\action;\mech)}}\\
&= \lambda\cdot \wel(\dist) 
- \mu \cdot \expect[\action \sim G]{
\rev(\action;\mech)}
\end{align*}
where the first equality holds by renaming the random variables $\hat{\type}_{-i}$ as $\type_{-i}$.
The inequality holds by applying the definition of the smoothness 
and taking expectation over the actions according to distribution $G$.
Note that in every equilibrium
$\bar{\sigma}(\type)= (\bar{\sigma}_1(\type_1), \ldots, \bar{\sigma}_n(\type_n))$ and 
$\hat{\sigma}(\type,s)= (\hat{\sigma}_1(\type_1,s_1), \ldots, \hat{\sigma}_n(\type_n,s_n))$, 
%\mbc{I have made the reference to equilibrium explicit. I have also tried to make it explicit that the strategy profile $\sigma(\ptype)$ does NOT mean $(\sigma_1(\ptype), \ldots, \sigma_n(\ptype))$ but rather $(\sigma_1(\ptype_1), \ldots, \sigma_n(\ptype_n))$. Maybe we should also say it in words?}
the utility of any agent is at least her utility given the above deviating strategy. % \mbc{this is unclear, as above we sum over all agents.}. 
Thus, 
\begin{align}\label{eq:util.bound}
&\sum_{i\in [n]}\expect[\type\sim\dist]{
\expect[\signal\sim\tilde{\sigma}(\type);s\sim \signal(\type)]{
\expect[\action \sim \hat{\sigma}(\type,s)]{
\util_i(\mech(\action);\type_i)}
- \cost_i(\signal_i,\type_i)}}\nonumber\\
&\geq 
\lambda\cdot \wel(\dist) 
- \mu \cdot \expect[\action \sim G]{
\rev(\action;\mech)}.
\end{align}
By rearranging the terms and noting that the sum of expected utility is the difference between 
equilibrium welfare and the expected revenue, we also have
\begin{align}\label{eq:util.eq}
&\sum_{i\in [n]}\expect[\type\sim\dist]{
\expect[\signal\sim\tilde{\sigma}(\type);s\sim \signal(\type)]{
\expect[\action \sim \hat{\sigma}(\type,s)]{
\util_i(\mech(\action);\type_i)}
- \cost_i(\signal_i,\type_i)}} \nonumber\\
&= 
\wel(\mech,(\bar{\sigma},\hat{\sigma}),\dist)
- \expect[\action \sim G]{
\rev(\action;\mech)}.
\end{align}
% \mbc{better to add this explicitly (algebraically)}, 
Multiplying both sides of the equality~\eqref{eq:util.eq} with factor $\mu$ and combining it with the inequality~\eqref{eq:util.bound} above, 
and recalling that $\mu\geq 1$ and the equilibrium utility is non-negative, we have 
\begin{align*}
\mu\cdot \wel(\mech,(\bar{\sigma},\hat{\sigma}),\dist)
\geq \lambda\cdot \wel(\dist), 
\end{align*}
i.e., $\poa(\mech, \dists^{\Pi}) \leq \frac{\mu}{\lambda}$.
\end{proof}

\subsection{Other Smooth Auctions}
\label{app:smooth.table}

In Table~\ref{table:smooth} we list some additional examples of smooth auctions for allocation problems that fall within our framework.  For each auction format we note the smoothness bound as well as the implied welfare bound when the auction is used in a combined market.

\begin{table}[htp]
	\begin{center}
		\begin{tabular}{|c|c|c|c|}
			\hline
			{\bf Auction Environment} & {\bf Mechanism $\mechf$} & {\bf  Smoothness} &
			\begin{minipage}{70pt} \vspace{5pt}
				{\bf PoA in \\Combined Market~$\mechc$}
				\vspace{3pt} \end{minipage}\\
			\hline
			\multirow{2}{*}{single-item} & first-price auction & $(1-\sfrac{1}{e}, 1)^*$ & $\sfrac{e}{(e-1)}$\\
			\cline{2-4}
			& all-pay auction & $(\sfrac{1}{2},1)^*$ & $2$\\
			\hline
%			multi-unit submodular & discriminatory auction & $(1-\sfrac{1}{e}, 1)^\mathsection$ & $\sfrac{e}{(e-1)}$\\
%			\hline
%			multi-unit subadditive & discriminatory auction & $(\sfrac{1}{2}-\sfrac{1}{2e}, 1)^\mathsection$ & $\sfrac{2e}{(e-1)}$\\
%			\hline
			\multirow{2}{*}{combinatorial, submodular } & simultaneous first price & $(1-\sfrac{1}{e}, 1)^\dagger$ & $\sfrac{e}{(e-1)}$\\
			\cline{2-4}
			& simultaneous all pay & $(\sfrac{1}{2},1)^\dagger$ & $2$\\
			\hline
			%position auction & generalized second price & & $(\frac{1}{2},2)^\ddagger$\\
			%\hline
		\end{tabular}
		\begin{minipage}{0.9\textwidth}{
				\center \footnotesize
				$^*$ \cite{roughgarden2012price} \qquad
				$^\dagger$ \cite{syrgkanis2013composable} \qquad
				%$^\ddagger$ \cite{lucier2011gsp} \qquad
%				$^\mathsection$ \cite{de2013inefficiency} \qquad
				%$^\mathparagraph$ \cite{roughgarden2017price} \qquad
			}
		\end{minipage}
	\end{center}
	\caption{\label{table:smooth}
		The first column lists the auction environment and the second the auction mechanism. 	
		The third column lists the $(\lambda,\mu)$-smoothness
		results from the literature, 
%		which is carried over to any combined market
%		(for information released after the auction and any aftermarket trade mechanism), by \Cref{thm:smooth}.
		and by \cref{thm:poa combined market}, this implies the price of anarchy upper bound for the combined market, when every valuations are independently distributed, as listed in the last column. }%for any secondary market satisfying voluntary participation and weak budget balance.
%		These PoA results are formally stated in \cref{cor:poa combined}.}
	\label{default}
\end{table}

\section{Low Welfare Perfect Bayesian Equilibrium}
\label{sec:sequential}

{In this section we prove \cref{prop:UDS}. The claim states that the strategies described form a perfect Bayesian equilibrium (PBE) in the combined market of \cref{example:IEWDS}. 
We actually prove a stronger claim, showing that these strategies form a sequential equilibrium (SE), which is a refinement of perfect Bayesian equilibrium. We start by defining the two notions (PBE and SE) for combined markets. We then present the proof of \cref{prop:UDS}. 
Finally, we present another example of welfare loss that is parameterized by the market power of any speculator. }

\subsection{Sequential Equilibrium in Combined Markets}

{
We now formally define perfect Bayesian equilibrium (PBE) and sequential equilibrium (SE).  
%\mbc{Up to this point we never mentioned SE, so the reader might be puzzled. I think that  
%we need to say that sequential equilibrium is a refinement of PBE, and we plan to prove that the negative result holds even for SE (and thus for PBE). I have added some sentences above}
Fully general definitions of PBE and SE 
%\mbc{two definitions  (SE+PBE)?} is 
are beyond the scope of this paper, so we will provide a definition tailored to our setting of combined mechanisms.

Recall our description of the two-stage combined market game $\game(\mechf, \Gamma, \mechs)$.  Recall also that a BNE consists of a profile of agent strategies $\mathbf{\sigma}^1$ for the auction $\mechf$, where $\sigma_i^1$ maps agent $i$'s type to an action, and a profile of strategies $\mathbf{\sigma}^2$ for the secondary auction $\mechs$, where $\sigma_i^2$ maps agent $i$'s type and the realization of observations $(x^{\mechf}(a^1), p_i^{\mechf}(a^1), s_i)$ to an action.

A perfect Bayesian equilibrium (PBE) additionally includes a belief function $\beta_i$ for each agent $i$, which maps the realization of observations $(x^{\mechf}(a^1), p_i^{\mechf}(a^1), s_i)$ for agent $i$ into a distribution over the types of the other agents.  We think of $\beta_i$ as agent $i$'s posterior belief about the other agents' types given the observed outcome of the auction game.  We write $\beta = (\beta_1, \dotsc, \beta_n)$ for the tuple of belief functions.

The collection $(\sigma^1, \sigma^2, \beta)$ forms a perfect Bayesian equilibrium (PBE) if
\begin{itemize}
\sloppy
    \item $\beta_i$ is the correct posterior distribution over agent types, given the observations $(x^{\mechf}(a^1), p_i^{\mechf}(a^1), s_i)$, assuming agents behave according to $\sigma^1$ in the auction $\mechf$, given the prior type distribution.
    \item For any realization of auction outcomes and observations, strategy $\sigma^2_i$ maximizes the expected utility of agent $i$ in the secondary market $\mechs$ under beliefs $\beta_i$.
    \item Given that agents apply strategies $\sigma^2$ in the secondary market, following strategy $\sigma_i^1$ in the auction $\mechf$ maximizes the expected utility of agent $i$.
\end{itemize}
In words, a PBE satisfies subgame perfection (as agents must best-respond in the secondary market for any realization of the initial auction), and moreover agents update their beliefs rationally and consistently given the strategies of others.

An important subtlety is how Bayes' rule should be applied to events with $0$ probability, such as under deviations from the equilibrium.  
The definition of PBE does not specify what beliefs are to be held by an agent if they observe an auction outcome that has probability $0$ given the type distribution and strategy profile $\sigma^1$.  
In sequential equilibrium, these beliefs are constrained by
thinking of each strategy as a limit of `trembling' strategies in which all possible actions have a positive chance of being observed.
More formally, a sequential equilibrium (SE) is a PBE such that there exists a sequence of totally mixing strategy profiles $\{\sigma^{1,k}\}$ for the auction mechanism converging to $\sigma^1$ 
and a sequence of beliefs $\{\beta^k\}_{k\geq 1}$ converging to $\beta$ such that $\beta^k$ is consistent with each agent applying Bayes' rule to all observations, under the assumption that agents are applying strategy profile $\sigma^{1,k}$.
}

% \propsequential
\subsection{Proof of \cref{prop:UDS}}

\begin{proof}[Proof of \cref{prop:UDS}]
{
We prove a stronger claim
by showing that this behavior forms a sequential equilibrium.
We first describe the beliefs of the agents about the valuation profile when entering the secondary market.  Specifically, we claim that each agent will have a posterior distribution that is simply equal to the prior distribution over all other agents' valuations.  That is, no agent receives an informative signal about the valuations of other agents.  These beliefs are indeed consistent with the suggested strategies, since the proposed bidding behavior of each agent in the auction does not depend on her valuation.  Furthermore, we assume that agents would rationalize any observed off-path behavior as independent of valuations, which is consistent with retaining the prior beliefs.

Now consider behavior in the secondary market given these beliefs. Whichever price speculator $C$ chooses, it is a dominant strategy for agents $A$ and $B$ to purchase items as long as their marginal values exceed the listed price.  Given this, it is dominant for agent $C$ to post a price that maximizes revenue, in expectation over her beliefs about the agent values.  One can verify that the suggested price of $1$ will always maximize revenue for agent $C$, given the prior valuation distribution, no matter what allocation is generated by the auction.  We can therefore assume that, no matter what outcome is generated by the auction, agent $C$ will list all acquired items at a price of $1$.

Given this behavior in the secondary market, we verify that the proposed bidding behavior in the auction is indeed an equilibrium. 
Indeed, any unilateral change in an agent's bid will either not change the outcome of the auction (when the acquired quantity is the same or lower), or turn the clearing price to $1$ or higher (if the agent obtains more items). 
This results in lower utility to every agent in the combined market, since items sell at a price of $1$ in the secondary market. 
}
\end{proof}

{
We note that this sequential equilibrium has an additional desirable property. % should we say here that "agents are not employing weakly dominated strategies in the primary auction"?.  
To describe it, we first note that since the agents have uniquely optimal behavior in the secondary market, it is straightforward define payoffs in the initial uniform-price auction with respect to the utility that will be obtained assuming this secondary market behavior.  Thinking of payoffs in this way, our property is that the agents are not employing \emph{weakly dominated strategies} in the primary auction.

\begin{definition}
A strategy $\sigma^1_i$ is weakly dominated by strategy $\overline{\sigma}^1_i$ if, for every action profile $a^1_{-i}$ of the other agents, the total expected utility of agent $i$ is no greater when using $\sigma^1_i$ than when using $\overline{\sigma}^1_i$, and there exists some action profile for which $\overline{\sigma}^1_i$ is strictly better.  Strategy $\sigma^1_i$ is \emph{not weakly dominated} if there is no other strategy that weakly dominates it.
\end{definition}

For agents $A$ and $B$, it is easy to see that the chosen bid profiles are not weakly dominated.  Indeed, any bid less than $2$ on the first item would be worse against a competing bid strictly between the placed bid and $2$ on each item, and any non-zero bid on a second item would be strictly worse against a competitor that bids $2$ on $m-1$ items, as this would only increase the price paid.

To see why the bidding strategy of $C$ is not weakly dominated, notice that if agent $A$ wins two units in the primary market, then the speculator cannot obtain utility by selling to agent $A$ in the secondary market.  So any bid on fewer than $m-2$ units would be strictly worse than bidding $1$ on $m-2$ units if agent $A$ bids $\epsilon$ on two units and agent $B$ bids $1$ on a single unit.  Likewise, bidding $b < 1$ on any of the $m-2$ units would be strictly worse than bidding $1$ on $m-2$ units if agent $A$ bids $b+\epsilon$ on two units and agent $B$ bids $1$ on a single item, since the per-unit price is effectively unchanged but the speculator wins one fewer unit and agent $A$ wins one more, causing the speculator to lose the opportunity to sell a unit to agent $A$.  Finally, increasing the bids above $0$ on any of the remaining two units can only increase the price paid by agent $C$, which would strictly reduce her utility.
}

\subsection{Additional Example: Welfare Gap with Limited Auction Winnings per Agent}
\label{sec:lower.bound.gamma}

%\mbc{revise this:} The following example shows that the welfare gap can be large under additional equilibrium refinements.  

%\mbc{change "unit" to "item"? }
% Maybe to this:
One aspect of the example proving Proposition~\ref{prop:UDS} is that, on the equilibrium path, a single speculator wins nearly all of the items in the auction.  In practice we may not expect to see such stark auction outcomes, or for a single agent to acquire such strong monopoly power.
The following example shows that the welfare gap from Proposition~\ref{prop:UDS} can be large even if each agent possesses only a bounded fraction $\gamma \in (0,1)$ of the items in the aftermarket.  One corollary is that even if a regulator imposes an auction rule limiting the number of items that can be acquired by any one participant, there can still be significant welfare loss due to speculation.

\begin{example}\label{example:additional refine}
For $\gamma<1$ and let $k=\left \lceil{ \frac{1}{\gamma}}\right \rceil$.
For an integer $r>3$, let $m=r\cdot k$. 
There are $m > 3k$ units to be allocated and $3k$ agents, which we categorize into $k$ groups.
%OLD: 
%For $\gamma<1$ and $m\geq 3$, let $k=\frac{1}{\gamma}$ \mbc{need to round this to an integer}. There are $m > 3k$ units to be allocated and $3k$ agents, which we categorize into $k$ groups. We assume $m$ are multiplies of $k$.
Agents are named $A_i,B_i$ and $C_i$ inside group $i$.   
Agent $A_i$ has marginal value $2$ for the first unit, value uniformly sample from $[1,1.5]$ for the second unit, and $0$ for any subsequent units. 
Agent $B_i$ has the following distribution over valuations.  She always has marginal value $2$ for the first unit acquired, then a value $z_{B_i} > 0$ for each subsequent unit acquired. 
Here $z_{B_i}$ is a random variable drawn from a distribution with CDF $F_B(z) = 1-\frac{1}{1+(2m/k-1)z}$ for $z \in [0, 1)$ and $F_B(z) = z-\frac{k}{2m}$ for $z\in [1,1+\frac{k}{2m}]$. 
Note that given buyer value with distribution $F$,
the unique revenue maximizing price is $p=1$ with {expected revenue of} $p \cdot \Pr[z_B \geq p] = k/(2m)$ {per-unit}.
Agent $C_i$ has value $0$ for any number of units; we refer to agent $C_i$ as a \emph{speculator}.

The primary auction is a uniform-price auction with standard bidding.
In the secondary market, trading only occurs within each group. The speculator $C_i$ can put some or all of the items that she has acquired in the auction up for sale, at a take-it-or-leave-it per-item price $p$ of her choice. Agent $A_i$ has the first opportunity to purchase any (or all) of the items made available by $C_i$ paying $p$ per item acquired.  
Then agent $B_i$ has the option to purchase any item that are still available, at price $p$.
\end{example}

We now describe the bidding strategies in the primary auction and agents' behavior in the secondary market, for agents in each group $i$. 
In the auction, agents $A_i$ and $B_i$ each bid $2$ for exactly a single unit and $0$ for the rest of the units. The speculator $C_i$ bids $1$ for $r-2 = \frac{m}{k}-2$ units and $0$ for the rest of the units.
Note that this implies that each speculator wins at most $\gamma$ fraction of the items in the entire market. 
Then, in the secondary market, the speculator $C_i$ offers all the units she has acquired for the price of $1$, first to $A_i$ and the remaining item to $B_i$. Agent $A_i$ buys one unit, and agent $B_i$ buys all $\frac{m}{k}-3$ remaining units if $z_{B_i}\geq1$, and nothing otherwise.
Note that under this behaviour the price in the auction is 0. 
% Agent $C$ has expected utility of $1+(m-3)/(2m)$, as she makes $1$ from selling to $A$, and $(m-3)/(2m)$ in expectation from selling to $B$.
\begin{proposition}
The above behaviour in the combined market forms a sequential equilibrium.
{Moreover, no agent is using a weakly dominated strategy in the primary auction with respect to the payoffs implied by the secondary market.}
\end{proposition}
The proof is essentially identical to the proof of \cref{prop:UDS}, and hence omitted here.

Consider the social welfare obtained in the combined market. 
The total expected welfare obtained for each group $i$ under this behavior is at most $6$ (as $2+1.5+2+\frac{k}{2m}\cdot (\frac{m}{k}-3)\cdot \left(1+\frac{k}{2m}\right)\leq 6$), 
whereas the optimal expected welfare for group $i$ is at least $4 + (\frac{m}{k}-2)\E[z]$. 
It is easy to compute that $\E[z] \geq \int_0^1 \frac{1}{1+(2m/k-1)z} = \Theta(\frac{k}{m}\cdot\log \frac{m}{k})$, and hence the optimal expected welfare is $\Theta(\log \frac{m}{k}) = \Theta(\log \gamma m)$.
Thus, if this behavior occurs at equilibrium, this implies that the price of anarchy for this combined market is $\Theta(\log \gamma m)$, {growing unboundedly large with~$\gamma m$}.

\section{Missing Proofs for Balanced Prices}
\label{apx:price proof}
\subsection{Proof of \cref{thm:poa.pricing}}
\begin{proof}[Proof of \cref{thm:poa.pricing}]
Our argument is an adaptation of a proof method due to \citet{dutting2020prophet}.
Our approach is to use the definition of balancedness to derive lower bounds on the utility obtained by the buyers and on the revenue collected by the primary auction mechanism.  Since agents achieve non-negative utility at equilibrium, and since the secondary market is budget-balanced, we can add these together to obtain a lower bound on the welfare generated in the combined market.

Since $\mechf$ is a posted-price mechanism that uses $(\alpha,\beta)$-balanced prices, we know that it uses a static pricing function $p$ defined by
%\mbe{do we restrict to deterministic allocations? We should be explicit (throughout)} 
$p(x) = \frac{\alpha}{1+\alpha\beta}\text{\bf E}_\mathbf{v}[p^\mathbf{v}(x)]$  where 
{for each possible type profile $\mathbf{v}$ the price $p^{\mathbf{v}}$ is $(\alpha,\beta)$-balanced for $\mathbf{v}$.}
%OLD: $p^{\mathbf{v}}$ is $(\alpha,\beta)$-balanced for each possible type profile $\mathbf{v}$.  %from the definition of balanced prices, where $x$ is some arbitrary allocation to an agent. 
Fix some Nash equilibrium of the combined mechanism $\mechc$, and write $\mathbf{x}(\mathbf{v})$ for the allocation obtained by the posted-price mechanism $\mechf$, at this equilibrium, when the valuation profile is~$\mathbf{v}$. 
We emphasize that $\mathbf{x}(\mathbf{v})$ is the allocation from the posted-price mechanism, and does not account for any transfer of items that might occur in the secondary market. 
We will also write $u_i(\mathbf{v})$ for the utility obtained by agent $i$ in the combined market (including any transfers that occur in the secondary market) when valuations are $\mathbf{v}$.

We first bound the consumer surplus.  We will have each player consider a strategy that purchases a certain collection of items in the primary market and avoids participating in the secondary market.  This strategy will always be feasible, since $\mechs$ is a trade mechanism.  {The strategy will depend only on the distribution over valuations, so in particular it will be independent of the signaling scheme $\Gamma$ and the details of $\mechs$.  Our utility bound will therefore hold for any choice of $\Gamma$ and any trade mechanism $\mechs$.}

To describe these strategies, we first sample ``phantom'' valuations $\mathbf{v}' \sim F$.  Buyer $i$ will consider buying the set of items $\opt_i( (v_i, \mathbf{v}'_{-i}), M\backslash \mathbf{x}(v'_i, \mathbf{v}_{-i}) )$ at the posted prices.  This is buyer $i$'s part of the optimal allocation, under valuations $(v_i, \mathbf{v}'_{-i})$, of all items that are not allocated (at equilibrium) by the posted price mechanism under valuations $(v'_i, \mathbf{v}_{-i})$. %\mbc{Is this so? I was expecting  "valuations $(v_i, \mathbf{v}'_{-i})$", but I suspect that this is on purpose. Yet, it left me  confused} \bjl{There are two places the valuation profile comes in.  First, it determines what items are already sold.  Second, it determines which subset of the remaining items buyer $i$ wants to purchase.  We want these to be independent.  So we use $(v_i, \mathbf{v}'_{-i})$ for one and $(v'_i, \mathbf{v}_{-i})$ for the other.  We want to use $(v'_i, \mathbf{v}_{-i})$ to determine what items are already sold, since we want to use the ``real'' values for the agents that arrived before agent $i$.  And we want to use $(v_i, \mathbf{v}'_{-i})$ to choose which items agent $i$ buyers, because $v_i$ is the ``real'' valuation for agent $i$.}\mbc{this is a very nice idea! Is this new for this paper? If so we might want to emphasize this as a contribution (else, cite). In any case, might be good to say this more explicitly when we explain in the parentheses below} \mbc{Stopped reading here - let's talk about this}.  
This strange choice of valuation profiles 
%\mbc{unclear as the valuations $\mathbf{v}'$ were sampled from $F$, not picked.}\bjl{This refers to the mix-and-match nature, using $(v_i, \mathbf{v}'_{-i})$ and $(v'_i, \mathbf{v}_{-i})$ instead of $\mathbf{v}$ and $\mathbf{v}'$.} 
is carefully chosen to make the valuation of the allocation independent of the items available, while still coupling with the outcomes under valuation profile $\mathbf{v}$.  Since $\opt_i( (v_i, \mathbf{v}'_{-i}), M\backslash \mathbf{x}(v'_i, \mathbf{v}_{-i}) )$ is a subset of the items actually available to agent $i$ when agents have valuations $\mathbf{v}$ {(as the set of items available to agent $i$ only depends on the valuation profile $\mathbf{v}_{-i}$ in a posted pricing mechanism)}, her expected utility $\text{\bf E}_{\mathbf{v}}[u_i(\mathbf{v})]$ is at least the expected utility of following this strategy (and not participating in the secondary market).  That is,
\begin{align*}
    \text{\bf E}_\mathbf{v}[u_i(\mathbf{v})] & \geq  \text{\bf E}_{\mathbf{v}, \mathbf{v}'}[ v_i(\opt_i( (v_i, \mathbf{v}'_{-i}), M\backslash \mathbf{x}(v'_i, \mathbf{v}_{-i}) )) - p(\opt_i( (v_i, \mathbf{v}'_{-i}), M\backslash \mathbf{x}(v'_i, \mathbf{v}_{-i}) )) ] \\
    & = \text{\bf E}_{\mathbf{v},\mathbf{v}'}[v'_i(\opt_i(\mathbf{v}', M\backslash \mathbf{x}(\mathbf{v}))) - p(\opt_i(\mathbf{v}', M\backslash \mathbf{x}(\mathbf{v})))].
\end{align*}
Here the inequality holds because agent $i$ is choosing her utility-optimal allocation given the prices, and the equality holds by a change of variables (swapping the role of $v'_i$ and $v_i$).
%\mbc{explain why the equality above holds.}
%\yl{The inequality holds since the agent $i$ are choosing the optimal allocation given the prices, 
%and the equality holds by just relabeling the variables.}

Summing the previous inequality over all buyers gives
\begin{equation}
\label{eq:util.1}
    \text{\bf E}_\mathbf{v}\left[\sum_i u_i(\mathbf{v})\right] \geq \text{\bf E}_{\mathbf{v},\mathbf{v}'}\left[\mathbf{v}'(\opt(\mathbf{v}', M\backslash \mathbf{x}(\mathbf{v})))\right] - \text{\bf E}_{\mathbf{v},\mathbf{v}'}\left[\sum_i p(\opt(\mathbf{v}', M\backslash \mathbf{x}(\mathbf{v}))) \right].
\end{equation}
Recalling that we set $p_i$ to be $\frac{\alpha}{1+\alpha\beta}$ times $\text{\bf E}_v[p^v_i]$, we can use linearity of expectation plus the definition of balanced prices (applied pointwise to each realization of $\mathbf{v}'$ and $\mathbf{x}(\mathbf{v})$) to conclude that
\[ \text{\bf E}_{\mathbf{v},\mathbf{v}'}\left[\sum_i p(\opt(\mathbf{v}', M\backslash \mathbf{x}(\mathbf{v})))\right] \leq \beta \cdot \frac{\alpha}{\alpha\beta + 1} \cdot \text{\bf E}_{\mathbf{v},\mathbf{v}'}[\mathbf{v}'(\opt(\mathbf{v}', M\backslash\mathbf{x}(\mathbf{v})))]. \]
Substituting into \eqref{eq:util.1} yields
\begin{equation}
\label{eq:util.2}    
\text{\bf E}_\mathbf{v}\left[\sum_i u_i(\mathbf{v})\right] \geq \left(1 - \frac{\alpha\beta}{1+\alpha\beta}\right) \text{\bf E}_{\mathbf{v},\mathbf{v}'}\left[\mathbf{v}'(\opt(\mathbf{v}', M\backslash \mathbf{x}(\mathbf{v})))\right] = \frac{1}{1+\alpha\beta}\text{\bf E}_{\mathbf{v},\mathbf{v}'}\left[\mathbf{v}'(\opt(\mathbf{v}', M\backslash \mathbf{x}(\mathbf{v})))\right].
\end{equation} 

We next provide a bound on the revenue generated in the posted-price mechanism at equilibrium.  Recall that $\mathbf{x}(\mathbf{v})$ is the equilibrium allocation in the posted-price mechanism under valuations $\mathbf{v}$.  From the definition of balanced prices we have
\begin{align*}
    \sum_i p(x_i(\mathbf{v})) &= \frac{\alpha}{1+\alpha\beta}\sum_i \text{\bf E}_{\mathbf{v}'}\left[p^{\mathbf{v}'}(x_i(\mathbf{v}))\right] \\
    & \geq \frac{1}{1+\alpha\beta} \text{\bf E}_{\mathbf{v}'}[\mathbf{v}'(\opt(\mathbf{v}')) - \mathbf{v}'(\opt(\mathbf{v}', M\backslash \mathbf{x}(\mathbf{v})))].
\end{align*}
Taking expectations over $\mathbf{v}\sim F$ yields
\begin{equation}
\label{eq:rev.1}
\text{\bf E}_\mathbf{v}\left[\sum_i p(x_i(\mathbf{v}))\right] \geq \frac{1}{1+\alpha\beta}\text{\bf E}_\mathbf{v}[\mathbf{v}(\opt(\mathbf{v}))] - \frac{1}{1+\alpha\beta}\text{\bf E}_{\mathbf{v}',\mathbf{v}}[\mathbf{v}'(\opt(\mathbf{v}', M\backslash \mathbf{x}(\mathbf{v})))].
\end{equation}

We are finally ready to bound the expected welfare obtained at equilibrium.  We first argue that the expected welfare obtained at equilibrium is at least the sum of the expected utility bound \eqref{eq:util.2} and expected revenue bound \eqref{eq:rev.1} calculated above.  Recall that the utility bound is a lower bound on the expected utility obtained from the buyers in the combined mechanism.  Moreover, since the secondary market $\mechs$ is a trade mechanism, it is budget-balanced and hence no additional revenue is extracted or lost.  Thus our bound on revenue collected in the auction also bounds the total revenue raised for the combined mechanism.  The total welfare obtained by the combined mechanism $\mechc$ is therefore the sum of the consumer surplus and the revenue generated by the posted price mechanism.

Summing up the utility and revenue bounds~\eqref{eq:util.2} and~\eqref{eq:rev.1}, we conclude that the total expected welfare of $\mechc$ is at least
\[ \text{\bf E}_\mathbf{v}\left[\sum_i u_i(\mathbf{v})\right] + \text{\bf E}_\mathbf{v}\left[\sum_i p(x_i(\mathbf{v}))\right] \geq \frac{1}{1+\alpha\beta}\text{\bf E}_\mathbf{v}[\mathbf{v}(\opt(\mathbf{v}))] \]
as claimed.
%
%(SK\text{\bf E}TCH) This is similar to the smoothness reduction.  The idea of balanced prices is that the sum of prices of any set of items offsets the loss in welfare due to removing those items.  This reasoning does not distinguish between items purchased for consumption and items purchased to resell. Thus even if agents purchase items with the intent of reselling in the secondary market, the amount they spent to acquire those items covers (approximately) any loss in welfare due to misallocation.  Importantly, these payments still contribute to expected welfare: each buyer's expected utility is non-negative by individual rationality, so each agent's payment will be offset by (a) value obtained by direct consumption, and (b) expected revenue obtained in the secondary market.
%
%This argument assumes that values are independent. In order for this to hold, we require that agents do not gain additional information about the realization of other buyers' values, beyond the known distributions.  So this result is not compatible with allowing speculators to purchase signals correlated with the values of others.
\end{proof}

\subsection{Proof of \cref{thm:reserve}}
\begin{proof}[Proof of \cref{thm:reserve}]
Our analysis is similar to the proof of Theorem~\ref{thm:poa.pricing}.  
Fix some Bayes-Nash equilibrium of the combined mechanism. 
Let $\mathbf{b}(\mathbf{v})$ be the bids made in the auction at this equilibrium
and $\opt(\mathbf{v}; k)$ be the welfare-optimal allocation of $k$ units when the valuation profile is $\mathbf{v}$. 
We denote the allocation of the auction under bids $\mathbf{b}$
as $\mathbf{x}(\mathbf{b})$, 
and the number of items left unallocated by the auction when agents bid according to $\mathbf{b}$ as $Z(\mathbf{b})$.
That is, $Z(\mathbf{b}) = m - \sum_i x_i(\mathbf{b})$.

We first bound the buyer surplus at equilibrium. To this end, fix a valuation profile $\mathbf{v}$ and consider a possible deviation by buyer $i$. 
Sample phantom valuations $\mathbf{v}' \sim F$. 
Let $z(v_i)$ be the largest index $j$ such that $v_{ij} \geq p$. 
Note that due to the reserve price, agent $i$ would obtain negative marginal utility for any items won in excess of $z(v_i)$. 
Let $y_i = \min\{z(v_i), \opt_i((v_i, \mathbf{v}'_{-i}); m)\}$ 
be agent $i$'s optimal allocation when others' valuations are $\mathbf{v}'_{-i}$, excluding any items for which her marginal value is less than $p$.
Our proposed deviation for agent $i$ is to place an auction bid of $b'_i$ where $b'_{ij} = p$ for $j \leq y_i$, and $b_{ij} = 0$ for $j > y_i$, 
then not participate in the secondary market. 
Importantly, this deviation depends on $v_i$ but not $v_{-i}$, and this deviation is feasible given any signaling protocol $\Gamma$. 
Note also that the utility obtained under this deviation can only be non-negative since $y_i \leq z(v_i)$ and the price paid per item obtained is exactly $p$.

Under this deviation, either agent $i$ wins $y_i$ items or all items are sold, 
and in the latter case agent~$i$ receives all items not allocated to the other agents under bid profile $b'_i(v'_i), \mathbf{b}_{-i}(\mathbf{v}_{-i})$.
Note that this quantity should be at least the number of items such that the bids of agents from $-i$ are strictly below $p$,
which again is at least 
$Z(b_i(v'_i), \mathbf{b}_{-i}(\mathbf{v}_{-i}))$, the number of items unallocated if we also include agent $i$ bidding at equilibrium as if her valuation is $v'_i$. 
% This is because the latter is exactly the number of items such that the bids of agents from $-i$ are at least $p$.
Thus we conclude that 
\[ x_i(b'_i, \mathbf{b}_{-i}(\mathbf{v}_{-i})) \geq \min\{y_i, Z(\mathbf{b}(v'_i, \mathbf{v}_{-i}))\}.\]

% Note that the number of items not allocated to the other agents under bid profile $b'_i(v'_i), \mathbf{b}_{-i}(\mathbf{v}_{-i})$
% equals the number of items not allocated to the other agents when agent $i$ does not participate in the auction 
% since the bid $b'_i(v'_i)$ is at most the reserve price. 
% Finally, this quantity is at least $Z(b_i(v'_i), \mathbf{b}_{-i}(\mathbf{v}_{-i}))$, the number of items unallocated if we also include agent $i$ bidding at equilibrium as if her valuation is $v'_i$. 

% The number of items not allocated to the other agents is at most $Z(b_i(v'_i), \mathbf{b}_{-i}(\mathbf{v}_{-i}))$, the number of items unallocated if we also include agent $i$ bidding at equilibrium as if her valuation is $v'_i$.  

The right hand side term of the above inequality is at least $\min\{z(v_i), \opt_i((v_i, \mathbf{v}'_{-i}); Z(\mathbf{b}(v'_i, \mathbf{v}_{-i})))\}$, 
agent $i$'s part of the optimal allocation of items left unallocated, up to a maximum of $z(v_i)$.  We conclude that
\begin{align*}
    \E_{\mathbf{v}}[u_i(\mathbf{v})]
    &\geq \E_{\mathbf{v},\mathbf{v}'}[v_i(\min\{z(v_i), \opt_i((v_i, \mathbf{v}'_{-i}); Z(\mathbf{b}(v'_i, \mathbf{v}_{-i})))\}) \\
    &\qquad\qquad\qquad - p \times \min\{z(v_i), \opt_i((v_i, \mathbf{v}'_{-i}); Z(\mathbf{b}(v'_i, \mathbf{v}_{-i})))\}] \\
    &\geq \E_{\mathbf{v},\mathbf{v}'}[v_i(\opt_i((v_i, \mathbf{v}'_{-i}); Z(\mathbf{b}(v'_i, \mathbf{v}_{-i})))) - p \times \opt_i((v_i, \mathbf{v}'_{-i}); Z(\mathbf{b}(v'_i, \mathbf{v}_{-i})))] \\    
    & = \E_{\mathbf{v},\mathbf{v}'}[v_i(\opt_i(\mathbf{v}; Z(\mathbf{b}(\mathbf{v}')))) - p \times \opt_i(\mathbf{v}; Z(\mathbf{b}(\mathbf{v}')))]
\end{align*}
where the second inequality follows because dropping the minimum with $z(v_i)$ can only introduce items with negative marginal utility at price $p$, and the final equality is a change of variables.

Summing over all agents, we have
\begin{align*}
    \sum_i \E_{\mathbf{v}}[u_i(\mathbf{v})]
    & \geq \sum_i \E_{\mathbf{v},\mathbf{v}'}[v_i(\opt_i(\mathbf{v}; Z(\mathbf{b}(\mathbf{v}')))) - p \times \sum_i \E_{\mathbf{v},\mathbf{v}'}[\opt_i(\mathbf{v}; Z(\mathbf{b}(\mathbf{v}')))]\\
    & \geq \E_{\mathbf{v},\mathbf{v}'}\left[\frac{Z(\mathbf{b}(\mathbf{v}'))}{m}\times \mathbf{v}(\opt(\mathbf{v}))\right] - p \times \E_{\mathbf{v}'}[Z(\mathbf{b}(\mathbf{v}'))]\\
    &= \E_{\mathbf{v}'}[Z(\mathbf{b}(\mathbf{v}'))] \times \E_{\mathbf{v}}\left[\frac{1}{m} \mathbf{v}(\opt(\mathbf{v}))\right] - p \times \E_{\mathbf{v}'}[Z(\mathbf{b}(\mathbf{v}'))]\\
    &= p \times \E_{\mathbf{v}}[Z(\mathbf{b}(\mathbf{v}))]
\end{align*}
%where the second inequality follows because (in the first summation) valuations have weakly decreasing marginals, so the value of optimally allocating $Z(\mathbf{b}(\mathbf{v}'))$ items is at least $Z(\mathbf{b}(\mathbf{v}'))/m$ times the value of optimally allocating $m$ items, and (in the second summation) because the total number of items allocated in $\opt_i(\mathbf{v}, Z(\mathbf{b}(\mathbf{v}')))$ is $Z(\mathbf{b}(\mathbf{v}'))$.  The following equality is independence, and the final equality then follows from the definition of $p$ and a change of variables.
%
%Since the bidder pays at most half of her marginal value per item won, the expected welfare obtained by the bidder is at least $\frac{1}{2}E[v_i(y_i)]$ times the expected fraction of the $y_i$ items that she wins.  She will win any item that would be left unsold by the auction, as long as her bid is above the reserve price.  (And she may win more than that, but we omit this from our bound.)
%
%Aggregating over all bidders, the expected buyer surplus is at least half of the expected per-item optimal welfare times the expected number of items left unsold.  \bjl{Need to work this out.}

On the other hand, the expected revenue raised is at least $p$ times the expected number of items sold in the auction, due to the reserve price.  That is, the revenue is at least
\begin{align*}
    p \times (m - \E_{\mathbf{v}}[Z(\mathbf{b}(\mathbf{v}))]).
\end{align*}
As in the proof of Theorem~\ref{thm:poa.pricing}, the total welfare of the combined mechanism is at least the sum of buyer surplus and the revenue generated by the auction.  Thus, summing up expected buyer surplus and expected revenue, we have that the total welfare generated at equilibrium is at least
\begin{align*}
    &p \cdot \E_{\mathbf{v}}[Z(\mathbf{b}(\mathbf{v}))] 
    + p (m - \E_{\mathbf{v}}[Z(\mathbf{b}(\mathbf{v}))]) 
    = p\cdot m 
    = \frac{1}{2}\cdot\E_{\mathbf{v}}[\mathbf{v}(\opt(\mathbf{v}))]
\end{align*}
as required.
%
%\mbc{this proof never explicitly talks about the price being (1,1) balanced. At what extend this result generalizes in the following sense: if we take any pricing that is $(\alpha,\beta)$-balanced as reserves in the uniform-price auction, we get the corresponding approximation in the combined market? I suspect this result should hold in this more general way. }\bjl{Yes, works for any $(\alpha,\beta)$-balanced prices.}
%
%\mbc{Nicole and I had a chat and she have raised the issue of robustness of the mechanism (which is an important practical issue)  - assume the distributions are not known exactly, so we have errors in the price - how will the approximation guarantee change? My sense is that we should be in a good shape (theoretically) as long as we somewhat shave down the price - assume that we take half of the "right" price (half of $p = \frac{1}{m}\mathbf{v}(\opt(\mathbf{v}))$) that was computed with respect to the approximate distributions. I suspect that if we take half the price, the approximation does not deteriorate by much with respect to the input distributions (maybe we get $1/4$ or $1/8$ of the optimum?), but then, if we do that and the true distribution is "not far from the estimate", we will also get something approximately good for the true distributions.  }
\end{proof}

\section{Uniform-Price Auction with Imprecise Reserves}
\label{sec:uniform.errors}

We now state and prove a 
{generalized version}
% MB: "relaxed" sounds like an easier problem to prove, but this is harder (more general) - changed. 
% OLD relaxed version 
of Theorem~\ref{thm:reserve} in which the reserve price used in the uniform-price auction may differ from the suggested value, {proving that the equilibrium welfare guarantee degrades gracefully with the error in price. In the special case that the price is exactly the suggested value, we get Theorem~\ref{thm:reserve} as a special case.} 

%\mbc{Is there any point to have the proof of theorem 4? Can we just say it is a corollary of the next proof?}
%\bjl{Sure, we could replace Appendix C.2 with a pointer to this proof.}

\begin{theorem}\label{thm:reserve.errors}
For a multi-unit allocation problem with non-increasing marginal values, consider a uniform-price auction subject to a per-unit reserve price of $p'$ such that $|p' - p| < \epsilon$ where
% $p$ which is $(\alpha,\beta)$-balanced for $\alpha, \beta\geq 1$.
$p = \frac{1}{2m}\mathbf{v}(\opt(\mathbf{v}))$.
Then for any trade mechanism, and at any equilibrium of the combined market, the expected welfare is at least 
% $\frac{1}{1+\alpha\beta}$ 
%half of the expected optimal welfare.
$\frac{1}{2}E[v(\opt(v))] - m\epsilon$.
\end{theorem}
\begin{proof}
Fix some Bayes-Nash equilibrium of the combined mechanism. 
Let $\mathbf{b}(\mathbf{v})$ be the bids made in the auction at this equilibrium
and $\opt(\mathbf{v}; k)$ be the welfare-optimal allocation of $k$ units when the valuation profile is $\mathbf{v}$. 
We denote the allocation of the auction under bids $\mathbf{b}$
as $\mathbf{x}(\mathbf{b})$, 
and the number of items left unallocated by the auction when agents bid according to $\mathbf{b}$ as $Z(\mathbf{b})$.
That is, $Z(\mathbf{b}) = m - \sum_i x_i(\mathbf{b})$.

We first bound the buyer surplus at equilibrium. To this end, fix a valuation profile $\mathbf{v}$ and consider a possible deviation by buyer $i$. 
Sample phantom valuations $\mathbf{v}' \sim F$. 
Let $z(v_i)$ be the largest index $j$ such that $v_{ij} \geq p$. 
Note that due to the reserve price, agent $i$ would obtain negative marginal utility for any items won in excess of $z(v_i)$. 
Let $y_i = \min\{z(v_i), \opt_i((v_i, \mathbf{v}'_{-i}); m)\}$ 
be agent $i$'s optimal allocation when others' valuations are $\mathbf{v}'_{-i}$, excluding any items for which her marginal value is less than $p'$.
Our proposed deviation for agent $i$ is to place an auction bid of $b'_i$ where $b'_{ij} = p'$ for $j \leq y_i$, and $b_{ij} = 0$ for $j > y_i$, 
then not participate in the secondary market. 
Importantly, this deviation depends on $v_i$ but not $v_{-i}$. 
Note also that the utility obtained under this deviation can only be non-negative since $y_i \leq z(v_i)$ and the price paid per item obtained is exactly $p'$.

Under this deviation, either agent $i$ wins $y_i$ items or all items are sold, 
and in the latter case agent~$i$ receives all items not allocated to the other agents under bid profile $b'_i(v'_i), \mathbf{b}_{-i}(\mathbf{v}_{-i})$.
Note that this quantity should be at least the number of items such that the bids of agents from $-i$ are strictly below $p'$,
which again is at least 
$Z(b_i(v'_i), \mathbf{b}_{-i}(\mathbf{v}_{-i}))$, the number of items unallocated if we also include agent $i$ bidding at equilibrium as if her valuation is $v'_i$. 
% This is because the latter is exactly the number of items such that the bids of agents from $-i$ are at least $p$.
Thus we conclude that 
$$ x_i(b'_i, \mathbf{b}_{-i}(\mathbf{v}_{-i})) \geq \min\{y_i, Z(\mathbf{b}(v'_i, \mathbf{v}_{-i}))\}.$$

% Note that the number of items not allocated to the other agents under bid profile $b'_i(v'_i), \mathbf{b}_{-i}(\mathbf{v}_{-i})$
% equals the number of items not allocated to the other agents when agent $i$ does not participate in the auction 
% since the bid $b'_i(v'_i)$ is at most the reserve price. 
% Finally, this quantity is at least $Z(b_i(v'_i), \mathbf{b}_{-i}(\mathbf{v}_{-i}))$, the number of items unallocated if we also include agent $i$ bidding at equilibrium as if her valuation is $v'_i$. 

% The number of items not allocated to the other agents is at most $Z(b_i(v'_i), \mathbf{b}_{-i}(\mathbf{v}_{-i}))$, the number of items unallocated if we also include agent $i$ bidding at equilibrium as if her valuation is $v'_i$.  

The right hand side term of the above inequality is at least $\min\{z(v_i), \opt_i((v_i, \mathbf{v}'_{-i}); Z(\mathbf{b}(v'_i, \mathbf{v}_{-i})))\}$, 
agent $i$'s part of the optimal allocation of items left unallocated, up to a maximum of $z(v_i)$.  We conclude that
\begin{align*}
    \E_{\mathbf{v}}[u_i(\mathbf{v})]
    &\geq \E_{\mathbf{v},\mathbf{v}'}[v_i(\min\{z(v_i), \opt_i((v_i, \mathbf{v}'_{-i}); Z(\mathbf{b}(v'_i, \mathbf{v}_{-i})))\}) \\
    &\qquad\qquad\qquad - p' \times \min\{z(v_i), \opt_i((v_i, \mathbf{v}'_{-i}); Z(\mathbf{b}(v'_i, \mathbf{v}_{-i})))\}] \\
    &\geq \E_{\mathbf{v},\mathbf{v}'}[v_i(\opt_i((v_i, \mathbf{v}'_{-i}); Z(\mathbf{b}(v'_i, \mathbf{v}_{-i})))) - p' \times \opt_i((v_i, \mathbf{v}'_{-i}); Z(\mathbf{b}(v'_i, \mathbf{v}_{-i})))] \\    
    & = \E_{\mathbf{v},\mathbf{v}'}[v_i(\opt_i(\mathbf{v}; Z(\mathbf{b}(\mathbf{v}')))) - p' \times \opt_i(\mathbf{v}; Z(\mathbf{b}(\mathbf{v}')))]
\end{align*}
where the second inequality follows because dropping the minimum with $z(v_i)$ can only introduce items with negative marginal utility at price $p$, and the final equality is a change of variables.

Summing over all agents, we have
\begin{align*}
    \sum_i \E_{\mathbf{v}}[u_i(\mathbf{v})]
    & \geq \sum_i \E_{\mathbf{v},\mathbf{v}'}[v_i(\opt_i(\mathbf{v}; Z(\mathbf{b}(\mathbf{v}')))) - p' \times \sum_i \E_{\mathbf{v},\mathbf{v}'}[\opt_i(\mathbf{v}; Z(\mathbf{b}(\mathbf{v}')))]\\
    & \geq \E_{\mathbf{v},\mathbf{v}'}\left[\frac{Z(\mathbf{b}(\mathbf{v}'))}{m}\times \mathbf{v}(\opt(\mathbf{v}))\right] - p' \times \E_{\mathbf{v}'}[Z(\mathbf{b}(\mathbf{v}'))]\\
    &= \E_{\mathbf{v}'}[Z(\mathbf{b}(\mathbf{v}'))] \times \E_{\mathbf{v}}\left[\frac{1}{m} \mathbf{v}(\opt(\mathbf{v}))\right] - p' \times \E_{\mathbf{v}'}[Z(\mathbf{b}(\mathbf{v}'))]\\
    &\geq \E_{\mathbf{v}'}[Z(\mathbf{b}(\mathbf{v}'))] \times \E_{\mathbf{v}}\left[\frac{1}{m} \mathbf{v}(\opt(\mathbf{v}))\right] - (p + \epsilon) \times \E_{\mathbf{v}'}[Z(\mathbf{b}(\mathbf{v}'))]\\
    &= (p - \epsilon) \times \E_{\mathbf{v}}[Z(\mathbf{b}(\mathbf{v}))]
\end{align*}
where the last inequality uses the assumption that $|p' - p| < \epsilon$, and the final equality uses the definition of $p$.
%where the second inequality follows because (in the first summation) valuations have weakly decreasing marginals, so the value of optimally allocating $Z(\mathbf{b}(\mathbf{v}'))$ items is at least $Z(\mathbf{b}(\mathbf{v}'))/m$ times the value of optimally allocating $m$ items, and (in the second summation) because the total number of items allocated in $\opt_i(\mathbf{v}, Z(\mathbf{b}(\mathbf{v}')))$ is $Z(\mathbf{b}(\mathbf{v}'))$.  The following equality is independence, and the final equality then follows from the definition of $p$ and a change of variables.
%
%Since the bidder pays at most half of her marginal value per item won, the expected welfare obtained by the bidder is at least $\frac{1}{2}E[v_i(y_i)]$ times the expected fraction of the $y_i$ items that she wins.  She will win any item that would be left unsold by the auction, as long as her bid is above the reserve price.  (And she may win more than that, but we omit this from our bound.)
%
%Aggregating over all bidders, the expected buyer surplus is at least half of the expected per-item optimal welfare times the expected number of items left unsold.  \bjl{Need to work this out.}

On the other hand, the expected revenue raised is at least $p'$ times the expected number of items sold in the auction, due to the reserve price.  That is, the revenue is at least
\begin{align*}
    p' \times (m - \E_{\mathbf{v}}[Z(\mathbf{b}(\mathbf{v}))]) \geq (p - \epsilon) \times (m - \E_{\mathbf{v}}[Z(\mathbf{b}(\mathbf{v}))]).
\end{align*}
As in the proof of Theorem~\ref{thm:poa.pricing}, the total welfare of the combined mechanism is at least the sum of buyer surplus and the revenue generated by the auction.  Thus, summing up expected buyer surplus and expected revenue, we have that the total welfare generated at equilibrium is at least
\begin{align*}
    &(p - \epsilon) \cdot \E_{\mathbf{v}}[Z(\mathbf{b}(\mathbf{v}))] 
    + (p - \epsilon) (m - \E_{\mathbf{v}}[Z(\mathbf{b}(\mathbf{v}))]) 
    = (p - \epsilon)\cdot m 
    = \frac{1}{2}\cdot\E_{\mathbf{v}}[\mathbf{v}(\opt(\mathbf{v}))] - m \epsilon
\end{align*}
as required.
%
%\mbc{this proof never explicitly talks about the price being (1,1) balanced. At what extend this result generalizes in the following sense: if we take any pricing that is $(\alpha,\beta)$-balanced as reserves in the uniform-price auction, we get the corresponding approximation in the combined market? I suspect this result should hold in this more general way. }\bjl{Yes, works for any $(\alpha,\beta)$-balanced prices.}
%
%\mbc{Nicole and I had a chat and she have raised the issue of robustness of the mechanism (which is an important practical issue)  - assume the distributions are not known exactly, so we have errors in the price - how will the approximation guarantee change? My sense is that we should be in a good shape (theoretically) as long as we somewhat shave down the price - assume that we take half of the "right" price (half of $p = \frac{1}{m}\mathbf{v}(\opt(\mathbf{v}))$) that was computed with respect to the approximate distributions. I suspect that if we take half the price, the approximation does not deteriorate by much with respect to the input distributions (maybe we get $1/4$ or $1/8$ of the optimum?), but then, if we do that and the true distribution is "not far from the estimate", we will also get something approximately good for the true distributions.  }
\end{proof}

\section{Efficiency in Single-Item Symmetric Environments}
\label{apx:single efficient}
We have shown that the welfare loss is bounded for various auctions.
In this section, we are interested in understanding conditions under which there will not be any welfare loss.
We present such conditions for some single-item auctions (e.g., first-price auctions),\footnote{Second-price auctions have multiple Nash equilibria, so equilibrium selection is an issue there, even within a standalone market.} 
in the well-studied case of two agents with i.i.d.\ valuation distributions.
Note that without aftermarkets, 
\citet{chawla2013auctions} show that for symmetric distributions, 
the first-price auction has a unique equilibrium, and it is efficient.

In \citet{hafalir2008asymmetric},
the authors consider the setting with two i.i.d.\ agents
and show that when combining the first-price auction with a secondary market taken from a specific family of secondary markets\footnote{Essentially, in the model of \citet{hafalir2008asymmetric},
the agent winning the item in the auction posts a price to the other agent in the secondary market.
}, assuming bids in the auction are not revealed before the secondary market, any BNE of the combined market in which bids in the auction are non-decreasing is welfare-maximizing.
%the equilibrium in the combined market is efficient if (1) actions are not revealed in the secondary market, and (2) the bids in the auction are non-decreasing.
We generalize their BNE efficiency result by showing it holds under significantly weaker conditions.  First, our result allows the aftermarket to be any trade mechanism that satisfies ex post individual rationality.
\begin{definition}\label{def:ex post participate}
A secondary market $\mechs$ satisfies \emph{ex post individual rationality} if 
for any Bayesian Nash equilibrium strategy $\sigma$ in the secondary market
given initial allocation $\alloc$, 
for any type profile~$\type$ and any agent~$i$, 
we have $\val_i(\alloc; \type_i) \leq u_i(\mechs(\alloc, \sigma(\theta));\type_i)$.
\end{definition}
Second, we allow arbitrary information about bids, outcomes, and payments of the auction to be revealed before the secondary market.
%make weaker assumptions on the information revealed before the secondary market, allowing information, as the bids in the auction, to be partially or fully revealed.
Given this information agents update their belief about others' valuation.
As in \citet{hafalir2008asymmetric}, we consider only strategies in which
%(1) valuations of the agents are not directly revealed before the secondary market and cannot be deduced beyond what is implied by the bids (\Cref{asp:no value reveal});
the bids in the auction are non-decreasing.\footnote{We conjecture that this assumption is unnecessary and that any BNE of the combined market is efficient, even without this assumption.}
% Our result applies to all secondary markets satisfying ex post individual rational and weak budget balance.
% and a weak ex post condition specified in \Cref{asp:ex post trade}.
The techniques we use in the proof are similar to the techniques in \citet{chawla2013auctions}.
It shows that any monotone strategies must be essentially identical.

\begin{restatable}{theorem}{thmefficient}\label{thm:symmetric efficiency}
Consider the  single-item setting with two i.i.d.\ buyers with atomless and bounded valuation distribution that has positive density everywhere on the support.
Fix the auction $\mechf$ to be either the first-price auction or the all-pay auction. Fix any signaling protocol $\Gamma$ and
any trade mechanism satisfying  ex post individual rationality.
Then  in the combined market $\combinedGame$,
every Bayesian Nash equilibrium\footnote{Note that there is at least one efficient equilibrium:  agents adopt symmetric strategies in the auction, the resulting auction allocation is efficient,
		and no trade occurs in the secondary market.} in which the bids of the agents in the auction are non-decreasing in values, is efficient.
\end{restatable}
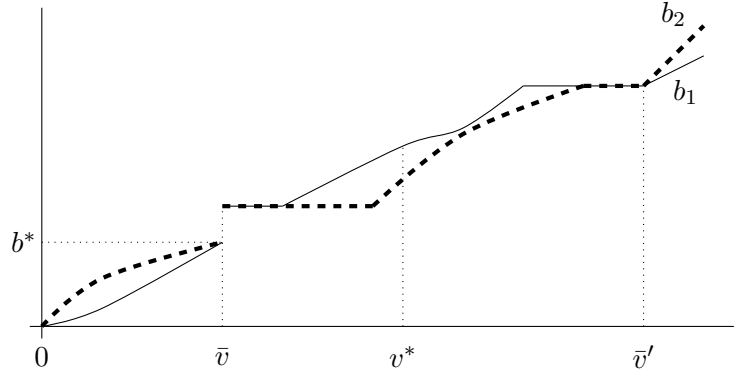
\begin{figure}[t]
\centering
\begin{tikzpicture}[scale = 0.8]

\draw (-0.2,0) -- (11.5, 0);
\draw (0, -0.2) -- (0, 5.3);

\draw plot [smooth, tension=0.6] coordinates { (0,0) (1,0.3) (3, 1.4)};
\draw (3, 2) -- (4, 2);
\draw plot [smooth, tension=0.6] coordinates { (4,2) (6,3) (7, 3.3) (8, 4)};
\draw (8, 4) -- (10, 4);
\draw (10, 4) -- (11, 4.5);

\begin{scope}[ultra  thick]

\draw [dashed] plot [smooth, tension=0.6] coordinates { (0,0) (1,0.8) (3, 1.4)};
% \draw [dashed] (0, 0) -- (3, 0);
% \draw [dashed] (3, 0) -- (3, 2);
\draw [dashed] (3, 2) -- (5.5, 2);
\draw [dashed] plot [smooth, tension=0.6] coordinates { (5.5,2) (7,3.2) (9, 4)};
\draw [dashed] (9, 4) -- (10, 4);
\draw [dashed] (10, 4) -- (11, 5);
\end{scope}

\draw (0, -0.5) node {$0$};

\draw [dotted] (6, 0) -- (6, 3);
\draw (6, -0.5) node {$\val^*$};
\draw [dotted] (3, 0) -- (3, 2);
\draw (3, -0.5) node {$\bar{\val}$};
\draw [dotted] (10, 0) -- (10, 4);
\draw (10, -0.5) node {$\bar{\val}'$};

\draw (10.7, 3.9) node {$\bid_1$};
\draw (10.5, 5.2) node {$\bid_2$};

\draw (-0.3, 1.4) node {$\bid^*$};
\draw [dotted] (0, 1.4) -- (3, 1.4);

\end{tikzpicture}
\vspace{-6pt}
\caption{\label{f:symmetric dist}
Asymmetric bid functions $\bid_1$ and $\bid_2$.
}
\end{figure}

In \Cref{thm:symmetric efficiency}, we obtain efficiency for both first-price auction 
and all-pay auction.
The proof for those auctions are similar and in the following parts, we only show it for the first-price auction. 

First we observe that Myerson payment identity result for the BNE expected payment of an agent as function of her expected allocation, 
holds in the combined market
when agents do not acquire costly information.

\begin{lemma}[\citealp{myerson1981optimal}]
\label{lem:payment identity}
In single-item setting, for any combined market, in any BNE and any agent $i$, 
if $\tilde{\alloc_i}(\cdot)$ is the interim allocation of $i$, and $\tilde{\pay}(\cdot)$ her interim payment function, then
\begin{align*}
\tilde{\pay}_i(\val_i) = \val_i\tilde{\alloc}_i(\val_i) - \int_0^{\val} \tilde{\alloc_i}(z) \dd z + \tilde{\pay}_i(0)
\end{align*}
where $\tilde{\pay}_i(0)$ is the expected payment of agent $i$ with value $0$.
\end{lemma}

\begin{proof}[Proof of \cref{thm:symmetric efficiency}]
Let $\underline{v}$ be the lowest value in the support,  and let $H$ be the highest value in the support  ($H$ is finite as we assume that the support is bounded). 
Since the secondary market satisfies 
ex post individual rational and weak budget balance, 
for any valuation profile such that the allocation in the auction environment is efficient, 
no trade occurs in the secondary market, 
and hence the allocation is efficient in the combined market.
If both agents use the same strictly increasing strategy,
the allocation is efficient in the auction environment for all valuation profiles, 
and hence it is efficient in the combined market. 
Moreover, if both agents use the same monotone strategy while there exists a non-trivial interval $(\val',\val'')$ with the same bid $b''$ for all values in this interval,
%for each value $\val \in (\val',\val'')$, 
%the utility of agent with value $\val$ is at least $(\val - b'' -\epsilon) \cdot F(\val'')$
%for any $\epsilon >0$ by bidding $b''+\epsilon$ in the auction and opt out the secondary market. 
%Note that the 
at least one agent with value in the interval can strictly increase her bid by an infinitesimal amount to strictly increase her utility, 
contradicting to the assumption that the strategies are sustained under equilibrium.

It remains to consider the case where the strategies of the two agents are monotone yet asymmetric.
For any value $\val^*$ such that $\bid_1(\val^*) > \bid_2(\val^*)$,
let $\bar{\val}\leq \val^*$ be the infimum value such that $\bid_1(\bar{\val}) \geq \bid_2(\bar{\val})$, 
and $\bar{\val}'\geq \val^*$ be the supremum value such that $\bid_1(\bar{\val}') \geq \bid_2(\bar{\val}')$. 
This is illustrated in \Cref{f:symmetric dist}.
We first consider the case that $\bar{\val} > \underline{v}$ and $\bar{v}' <H$. 
The boundary case where 
we have $\bid_1(\val) \geq \bid_2(\val)$
for any $\val < v^*$,
or $\bid_1(\val) \geq \bid_2(\val)$
for any $\val > v^*$,
is considered at the end of the proof.
Note that as we assume that $\bar{\val} > \underline{v}$, the definition of $\bar{\val}$ implies that there exists $\epsilon' > 0$
such that $\bid_2(\val) > \bid_1(\val)$ for 
any $\val\in(\bar{\val}-\epsilon', \bar{\val})$.
This further implies that 
$\lim_{\val\uparrow\bar{\val}} \bid_2(\val) 
\geq \lim_{\val\uparrow\bar{\val}}\bid_1(\val)$.
% First we show that $\util_1(\bar{\val}) = \util_2(\bar{\val})$
% and $\util_1(\bar{\val}') = \util_2(\bar{\val}')$. 

%\mbc{need to define the notation $\util_i(v)$. Is it the utility in the combined market? Assuming it is, the utility is a function of $v_i$ and of the strategies - but here the strategies are implicit. We need to say that we omit the strategies from the notation. here is my suggestions:  }

For the given strategies of the agents in the combined market, we denote the expected utility of agent $i\in \{1,2\}$ with value $v$ after the aftermarket by $\util_i(v)$. 
We wish to show that $\util_1(\bar{\val}) \geq \util_2(\bar{\val})$. 
The claim that $\util_1(\bar{\val}') \leq \util_2(\bar{\val}')$ holds similarly. 
The first step to showing that $\util_1(\bar{\val}) \geq \util_2(\bar{\val})$ is establishing that
$\lim_{\val\uparrow\bar{\val}} \bid_1(\val) 
= \lim_{\val\uparrow\bar{\val}}\bid_2(\val)$.\footnote{$\lim_{\val\uparrow\bar{\val}}$ means taking the limit as $\val$ increases to $\bar{\val}$ (limit from below).} 

\begin{claim}\label{claim:bidlimit}
	It holds that 
	$\lim_{\val\uparrow\bar{\val}} \bid_1(\val) 
	= \lim_{\val\uparrow\bar{\val}}\bid_2(\val)$.
\end{claim}
\begin{proof}
Recall that $\lim_{\val\uparrow\bar{\val}} \bid_2(\val) 
\geq \lim_{\val\uparrow\bar{\val}}\bid_1(\val)$ and assume in contradiction that 
$\epsilon = \lim_{\val\uparrow\bar{\val}} \bid_2(\val) 
- \lim_{\val\uparrow\bar{\val}}\bid_1(\val)>0$.
There exists a value $\hat{\val}_2$ for agent $2$ satisfying 
$\hat{\val}_2 \in (\bar{\val}-\frac{\epsilon}{4}, \bar{\val})$
such that 
$\bid_2(\hat{\val}_2) \in 
[\lim_{\val\uparrow\bar{\val}}\bid_2(\val) - \frac{\epsilon}{2}, 
\bid_2(\bar{\val})]$. 
%\yl{and $\Pr_{v\sim F_2}[\hat{\val} \leq v \leq \bar{v}] \leq \frac{\epsilon \cdot F_2(\bar{v})}{4\bar{v}}$}. 
%\mbc{$\val$ is in use, so we should replace the variable in the limit: $\lim_{\val\uparrow\bar{\val}}\bid_1(\val) $ better be $\lim_{\tilde{\val}\uparrow\bar{\val}}\bid_1(\tilde{\val})  $ }.
Thus agent 2 with value $\hat{\val}_2$ can lower her bid $\bid_2(\hat{\val}_2)$ by $\frac{\epsilon}{2}$ without affecting the allocation in the auction, 
and increase the utility in the auction environment by $\frac{\epsilon}{2}$ when she wins the item. 
Thus in any case she wins with bid $\bid_2(\hat{\val}_2)$, 
her utility from deviating and bidding $\bid_2(\hat{\val}_2)-\frac{\epsilon}{2}$ instead, is at least $\hat{\val}_2 - \bid_2(\hat{\val}_2) + \frac{\epsilon}{2}$ in the auction, and this cannot decrease in the combined market as the secondary market satisfies voluntary participation.
%ex post individual rational. 

We next argue that her utility in the combined market when bidding $\bid_2(\hat{\val}_2)$ in the auction is at most $\hat{\val}_2 - \bid_2(\hat{\val}_2) + \frac{\epsilon}{4}$, and thus smaller than with the deviation.
Indeed, if she wins with bid $\bid_2(\hat{\val}_2)$ in the auction, her utility from the auction is $\hat{\val}_2 - \bid_2(\hat{\val}_2)$.
If she ends up reselling the item in the secondary market, 
then since it satisfies ex post individual rationality and weak budget balance, 
her additional utility is bounded by 
$ {\val}_1-\hat{\val}_2 $, where $\val_1$ is the value of agent $1$.
To complete the proof we show that when agent $2$ wins in the auction with bid $\bid_2(\hat{\val}_2)$ it holds that ${\val}_1\leq \bar{\val}$, and this will prove the claim since $\val_1-\hat{\val}_2\leq \bar{\val}-\hat{\val}_2\leq \frac{\epsilon}{4}$.

Note that if $\bid_1(\val) > \bid_2(\hat{\val}_2)$ for any $\val > \bar{\val}$, the above claim holds. 
Otherwise, let $\epsilon''>0$ be the supremum value such that $\bid_1(\val) \leq \bid_2(\hat{\val})$ for any $\val \in [\bar{\val}, \bar{\val} + \epsilon'')$. % the supremum on gives the open interval from above.
Moreover, by the assumptions we made on the equilibrium strategy, we have 
$\bid_1(\val) \geq \bid_2(\val) \geq \bid_2(\hat{\val})$ 
for any $\val \in [\bar{\val}, \bar{\val} + \epsilon'')$.
Thus $\bid_1(\val) = \bid_2(\val)= \bid_2(\hat{\val})$ 
for any $\val \in [\bar{\val}, \bar{\val} + \epsilon'')$.
First we claim that $\bid_2(\hat{\val}) < \bar{\val} + \epsilon''$. 
This is because otherwise for agent $2$ with value $\hat{\val}$, 
her utility under equilibrium is negative since conditional on winning, 
her payment in the auction is at least $\bar{\val} + \epsilon''$, 
the value of agent $1$ is $\val_1 \leq \bar{\val} + \epsilon''$ (as agent 1 losses the auction),
and strictly smaller with positive probability, % \mbc{why?}\yl{with additional continuity assumption added to statement, this holds since $v< \bar{\val} + \epsilon''$ has positive measure.}, 
and the payment from resale is at most $\val_1$.
In the case that $\bid_2(\hat{\val}) < \bar{\val} + \epsilon''$, 
there exists sufficiently small $\delta>0$
such that agent $2$ with value $\bar{\val} + \epsilon''-\delta$ 
can simply increase the bid by an infinitesimal amount, which would win the item for all opponent's value $\val_1< \bar{\val} + \epsilon''$, which increases her utility. 
Therefore, we must have $\lim_{\val\uparrow\bar{\val}} \bid_1(\val) = \lim_{\val\uparrow\bar{\val}}\bid_2(\val)$.
\end{proof}

%\mbc{GOT TO HERE, PLEASE LEAVE THIS}

% \mbc{add another encapsulated claim about  equality of utilities. Be explicit about your assumptions (do you assume $\bar{\val}$ is not the infimum value?) . Then add another claim about the edge case if it is not covered.} \yl{I didn't add claim for the equality because it seems weird and break up the coherent proof into unnecessary pieces. I add a claim for other parts. See claim 2 below.}

Given \cref{claim:bidlimit} we can define $\bid^*= \lim_{\val\uparrow\bar{\val}} \bid_1(\val) = \lim_{\val\uparrow\bar{\val}}\bid_2(\val)$.
Fix some $\delta > 0$ such that $\bid_1(\val), \bid_2(\val) \in (b^*-\delta, b^*]$, and
let $\delta' \in (0, \delta)$ be such that
%the constant such that $\bid_1(\val), \bid_2(\val) \in (b^*-\delta, b^*]$ 
%and 
$\bid_2(\val) > \bid_1(\val)$
for any $\val \in (\bar{\val}-\delta',\bar{\val})$.
The utility of agent 1 with any value $\val \in (\bar{\val}-\delta', \bar{\val})$ is 
\begin{align*}
\util_1(\val) \geq (\val - b^* - \delta) \cdot \Pr_{\val\sim\dist_2}[\val \leq \bar{\val}]
\geq (\bar{\val} - b_2(\val)) \cdot \Pr_{\val\sim\dist_1}[\val \leq \bar{\val}] - 3\delta
\geq \util_2(\val) - 4\delta.
\end{align*}
The first inequality holds because player 1 can always choose to bid $b^*+\delta$ 
and wins with probability at least $\Pr_{\val\sim\dist_2}[\val \leq \bar{\val}]$ in the auction,
and retain at least the same utility in the combined market since the secondary market satisfies voluntary participation. 
The second inequality holds because $\bar{v} - \val \leq \delta'\leq \delta$, $b^*-b_2(v)\leq \delta$
and $\Pr_{\val\sim\dist_1}[\val \leq \bar{\val}] = \Pr_{\val\sim\dist_2}[\val \leq \bar{\val}]$ due to the i.i.d.\ assumption.
The final inequality follows because, similar to the argument in the previous paragraph, 
the payment to agent 2 from reselling the item to agent~1 is at most $\bar{\val}$, 
and the utility increase of agent 2 in the secondary market is at most $\delta'\leq\delta$.
%The latter implies the last inequality above. 
Hence $\lim_{\val\uparrow\bar{\val}} \util_1(\val) \geq \lim_{\val\uparrow\bar{\val}}\util_2(\val)$.
Since the interim utility is a continuous function of the valuation, we conclude that
$\util_1(\bar{\val}) \geq \util_2(\bar{\val})$ as desired.

A symmetric argument establishes that $\util_1(\bar{\val}') \leq \util_2(\bar{\val}')$ as well.  The next claim shows that these inequalities imply the desired efficiency.

\begin{claim}\label{clm:efficient}
For any pair of values $\bar{\val}' > \bar{\val}$ such that 
$b_1(v) \geq b_2(\val)$ for any $v\in (\bar{\val}, \bar{\val}')$,
if $\util_1(\bar{\val}) \geq \util_2(\bar{\val})$
and $\util_1(\bar{\val}') \leq \util_2(\bar{\val}')$,
the allocation is efficient for values between $\bar{\val}$ and $\bar{\val}$.
\end{claim}
\begin{proof}
By \Cref{lem:payment identity}, 
the interim allocations of the agents in the combined market under Bayesian Nash equilibrium satisfy 
\begin{align*}
u_1(\bar{\val}') - u_1(\bar{\val})
= \int_{\bar{\val}}^{\bar{\val}'} \tilde{\alloc}_1(\val) \dd \val
\geq \int_{\bar{\val}}^{\bar{\val}'} \tilde{\alloc}_2(\val) \dd \val
= u_2(\bar{\val}') - u_2(\bar{\val}).
\end{align*}
The inequality holds because (1) the allocation in the auction environment satisfies $\alloc_1(\val) \geq \alloc_2(\val)$ since $\bid_1(\val) > \bid_2(\val)$ for any $\val\in (\bar{\val},\bar{\val}')$;
and (2) no trade happens in the secondary market if agent $i$ wins the item and $\val_i > \val_{-i}$ since the secondary market satisfies ex post individual rationality  and weak budget balance. 
The two claims implies that the interim allocation satisfies 
$\tilde{\alloc}_1(\val) \geq \tilde{\alloc}_2(\val)$
for any $\val\in (\bar{\val},\bar{\val}')$, 
and hence the inequality holds for the integration of the interim allocation.
Moreover, $\util_1(\bar{\val}) \geq \util_2(\bar{\val})$
and $\util_1(\bar{\val}') \leq \util_2(\bar{\val}')$
implies that 
$\util_1(\bar{\val}') - \util_1(\bar{\val}) 
\leq \util_2(\bar{\val}') - \util_2(\bar{\val})$,
and hence both inequalities must hold with equality. 

In order for the integral of interim allocations to coincide, 
we have $\tilde{\alloc}_1(\val) = \tilde{\alloc}_2(\val)$ for any $\val\in (\bar{\val},\bar{\val}')$, 
which implies that 
the item is sold from agent $-i$ to $i$ in the secondary market 
if agent $i$ loses the item, $\val_i < \val_{-i}$, and $\val_i \in (\bar{\val}, \bar{\val}')$.
Therefore, the allocation in the combined market is efficient in range $(\bar{\val}, \bar{\val}')$.
\end{proof}

Finally, we address the boundary cases where 
we have either $\bid_1(\val) \geq \bid_2(\val)$
for any $\val < v^*$,
or $\bid_1(\val) \geq \bid_2(\val)$
for any $\val > v^*$.
If $\bid_1(\val) \geq \bid_2(\val)$
for any $\val < v^*$,
%the bid of agent $1$ is always weakly above the bid of agent $2$. 
%then we know that $u_1(\underline{v})\geq u_2(\underline{v})$ since 
then $u_1(\underline{v})\geq 0$ (as the utility of agent $1$ is always non-negative under equilibrium), 
and $u_2(\underline{v}) = 0$ (since agent $2$ with value $\underline{v}$ wins with probability zero, as the distribution is atomless,
%\mbc{what about ties at bids? I think in every NE the agents with value $\underline{v}$ must bid their value, no?  if so the utility is 0 for both agents. Note that it is not clear that "continuous distribution" necessarily means no atom at the infimum (the CDF is continuous even with such an atom)- if you want this you need to assume it explicitly  (but I do not think it is needed)
	%Does "continuous distribution" necessarily means no atom at the infimum, or do we need to state this explicitly ? Assuming not, I think we need to say this happens with probability zero, not use "never"
%}\yl{I think for continuous distributions we cannot have atom anywhere.}
 and no trade occurs in the secondary market when agent $2$ loses in the auction).
 This implies that $u_1(\underline{v})\geq u_2(\underline{v})$.
In the other case that $\bid_1(\val) \geq \bid_2(\val)$ for any $\val > v^*$, 
for the highest value $H$, we have $u_1(H) = u_2(H)$.
First note that under equilibrium we have $b_1(H)=b_2(H)$, because otherwise agent $1$ with value $H$ can decrease her bid and retain the same probability of winning in the auction. 
Since in this case each of the agents when her value is $H$ wins the item with probability $1$ (ties have probability 0 as the distribution has no atom at $H$)
with payment $b_1(H)=b_2(H)$, and no trade occurs in the secondary market, we have $u_1(H) = u_2(H)$. 
Again by applying \cref{clm:efficient}, 
the interim allocations for both agents coincide in the combined market, 
and the allocation is efficient.

We note that we did not make use of the structure of the signal distribution at any point in the argument above.  Thus it is apparent that the argument does not depend on the content of the signals released after the auction.  In particular, our conclusion that the allocation is efficient holds even if the auction bids are revealed before the secondary market begins.      
\end{proof}

\end{document}